\documentclass[lettersize,journal]{IEEEtran}
\usepackage{amsmath,amsfonts}
\usepackage{amsthm}
\usepackage{xcolor}
\usepackage{algorithm}
\usepackage{algpseudocode}
\usepackage{array}
\usepackage{caption}
\usepackage[caption=false,font=normalsize,labelfont=sf,textfont=sf]{subfig}
\usepackage{textcomp}
\usepackage{stfloats}
\usepackage{url}
\usepackage{verbatim}
\usepackage{graphicx}
\usepackage{cite}
\usepackage{amssymb}
\usepackage{bbold}
\DeclareMathOperator{\E}{\mathbb{E}}
\newcommand{\probP}{\text{I\kern-0.15em P}}

\newtheorem{theorem}{Theorem}[section]
\theoremstyle{remark}
\newtheorem*{example}{Example}

\theoremstyle{definition}
\newtheorem{definition}{Definition}[section]

\theoremstyle{plain}
\newtheorem{conjecture}{Conjecture}[section]

\newtheorem*{setting}{Setting}

\begin{document}

\title{Bidding efficiently in Simultaneous Ascending Auctions with budget and eligibility constraints using Simultaneous Move Monte Carlo Tree Search}

\author{Alexandre Pacaud, Aurelien Bechler and Marceau Coupechoux
\thanks{A.Pacaud and A.Bechler are with Orange Labs, France (e-mail: alexandre.pacaud@orange.com, aurelien.bechler@orange.com).}
\thanks{A.Pacaud and M.Coupechoux are with LTCI, Telecom Paris, Institut Polytechnique de Paris, France (e-mail: marceau.coupechoux@telecom-paris.fr). The work of M. Coupechoux has been performed at LINCS (lincs.fr).}
}

\markboth{Bidding efficiently in SAA with budget and eligibility constraints using SM-MCTS}%
{Shell \MakeLowercase{\textit{et al.}}: A Sample Article Using IEEEtran.cls for IEEE Journals}

\maketitle

\begin{abstract}
For decades, Simultaneous Ascending Auction (SAA) has been the most popular mechanism used for spectrum auctions. It has recently been employed by many countries for the allocation of 5G licences. Although SAA presents relatively simple rules, it induces a complex strategic game for which the optimal bidding strategy is unknown. Considering the fact that sometimes billions of euros are at stake in an SAA, establishing an efficient bidding strategy is crucial. In this work, we model the auction as a $n$-player simultaneous move game with complete information and propose the first efficient bidding algorithm that tackles simultaneously its four main strategic issues: the \textit{exposure problem}, the \textit{own price effect}, \textit{budget constraints} and the \textit{eligibility management problem}. Our solution, called $SMS^\alpha$, is based on Simultaneous Move Monte Carlo Tree Search (SM-MCTS) and relies on a new method for the prediction of closing prices. By introducing a new reward function in $SMS^\alpha$, we give the possibility to bidders to define their own level of risk-aversion. Through extensive numerical experiments on instances of realistic size, we show that $SMS^\alpha$ largely outperforms state-of-the-art algorithms, notably by achieving higher expected utility while taking less risks. 
\end{abstract}


\section{Introduction}

In order to provide high quality service and develop wireless communication networks, mobile operators need to have access to a wide range of frequencies. These frequencies are obtained in the form of licences. A licence is defined by four features: its frequency band, its geographic coverage, its period of usage and its restrictions on use. Nowadays, spectrum licences are mainly assigned through auctions. \textit{Simultaneous Ascending Auction} (SAA), also known as \textit{Simultaneous Multi Round Auction} (SMRA), has been the privileged mechanism used for spectrum auction since its introduction in 1994 by the US Federal Communications Commission (FCC) for the allocation of wireless spectrum rights. For instance, it has been used in Portugal \cite{5G_portugal}, Germany \cite{5G_germany}, Italy \cite{5G_italy} and the UK \cite{5G_UK} to sell 5G licences. SAA is also expected to play a central role in future spectrum allocations, e.g. for 6G licenses. The popularity of SAA is mainly due to the relative simplicity of its rules and the generation of substantial revenue for the regulator. Both of its creators, Paul Milgrom and Robert Wilson, received the 2020 Sveriges Riksbank Prize in Economic Sciences in Memory of Alfred Nobel mainly for their contributions to SAA. Establishing an efficient bidding strategy for SAA is crucial for mobile operators, especially considering the large amount of money involved, e.g. Deutsche Telekom spent 2.17 billion euros in the 5G German SAA. This is the aim of this work. 

SAA has a dynamic multi-round auction mechanism where bidders submit their bids simultaneously on all licences each round. It offers the freedom to adjust bids throughout the auction while taking into account the latest information about the likelihood of winning different sets of licences. Hence, a great number of bidding strategies can be applied. Unfortunately, selecting the most efficient one is a difficult task. Indeed, SAA induces a $n$-player simultaneous move game with incomplete information with a large state space for the solution of which no generic exact game resolution method is known \cite{reeves2005exploring}. 

In addition to the complexities tied to its general game properties, SAA presents a number of complex strategic issues. Its four main strategic issues are the \textit{exposure problem}, the \textit{own price effect}, \textit{budget constraints} and the \textit{eligibility management problem}. The exposure problem corresponds to the situation where a bidder pursues a set of complementary licences but ends up by paying more than its valuation for the ones it actually wins. The own price effect refers to the fact that bidding on a licence inevitably increases its price and, hence, decreases the utility of all bidders willing to acquire it. On the contrary, it is in the interest of all bidders to keep prices as low as possible. Budget constraints correspond to a fix budget that caps the maximum amount that a bidder can bid during an auction and, thus, can hugely impact an auction's outcome. The eligibility management problem is introduced by activity rules which penalise bidders that do not maintain a certain level of bidding activity. At the beginning of the auction, each bidder is given a certain level of eligibility. Each round a bidder fails to satisfy the activity rule, its eligibility is reduced. As bidders are forbidden to bid on sets of licences which exceed their eligibility, managing efficiently one's eligibility during the course of an auction is crucial to obtain a favourable outcome. In this work, we propose the first efficient bidding algorithm which tackles simultaneously the four strategic issues of SAA. 

\subsection{Related works}

Most works on SAA, such as \cite{milgrom2000putting,cramton2002spectrum,cramton2006simultaneous}, have focused on its mechanism design, its efficiency and the revenue it generates for the regulator. Only a few works have addressed the bidder's point of view. These studies generally consider one of the two following formats of SAA: its original format \cite{cramton2006simultaneous} and its corresponding clock format defined hereafter. In neither of these formats, an efficient bidding strategy tackling simultaneously its four main strategic issues has yet been proposed. Generally, research has focused on trying to solve one of these strategic issues in specific simplified versions of these formats. Moreover, the solutions proposed can often only be applied to small instances. 

As the original format of SAA is generally too complex to draw theoretical guarantees, a simplified clock format of SAA \cite{goeree2014equilibrium} with two types of bidders (\textit{local} and \textit{global}) is often considered. It presents the advantage of being a tractable model where bidders have continuous and differentiable expected utilities. Standard optimisation methods can then be applied to derive an equilibrium. 

In the literature, the clock format is mainly used to analyse the exposure problem. Global bidders all have super-additive value functions. Goeree et al \cite{goeree2014equilibrium} consider the case of identical licences for which they compute the optimal dropout level of each global bidder using a Bayesian framework. They extend their work to a larger class of value functions (regional complementarities) but with only two global bidders. By modifying the initial clock format of SAA with a pause system that enables jump bidding, Zheng \cite{zheng2012jump} builds a continuation equilibrium which fully eliminates the exposure problem in the case of two licences and one global bidder. Using a different pause system, Brusco and Lopomo \cite{brusco2009simultaneous} study the effect of binding public budget constraints on the structure of the unique noncollusive equilibria in the case of two licences and two global bidders. They show that such constraints can be a great source of inefficiency.

Regarding the original format of SAA, Wellman et al. \cite{wellman2008bidding} propose an algorithm which uses probabilistic predictions of closing prices to tackle exposure. Results seemed promising but were only obtained for a specific class of super-additive value functions. 

The own price effect has also been studied in the original format of SAA. In a simple example of SAA with two licences between two bidders having the same public value function, Milgrom \cite{milgrom2000putting} describes a collusive equilibrium. This work was then pursued by Brusco and Lopomo \cite{brusco2002collusion} who build a collusive equilibrium based on signalling for SAA with two licences between two bidders having super-additive value functions. Similarly to the algorithm built to tackle exposure, Wellman et al. \cite{wellman2008bidding} propose an algorithm to tackle the own price effect based on the probabilistic prediction of closing prices when all licences are identical and bidders have subadditive value functions. However, obtained results were unsatisfactory as they are significantly inferior to a simple demand reduction algorithm. 

Regarding budget constraints and the eligibility management problem, little work has been done in the original format of SAA. However, it is commonly accepted that one should gradually reduce its eligibility to avoid being trapped in a vulnerable position if other bidders do not behave as expected~ \cite{weber1997making}. 

In our previous work \cite{pacaud2022monte}, we presented a bidding strategy computed by Monte Carlo Tree Search (MCTS) that we applied to a deterministic version of the original format of SAA with complete information. In this paper, we extend our work to simultaneous moves, budget constraints, activity rules, risk-averse rewards and larger instances. All four MCTS phases have been modified.

\subsection{Contributions}

In this paper, we consider the original format of SAA with complete information for which we propose the first bidding algorithm, named $SMS^\alpha$, tackling simultaneously its four main strategic issues. We make the following contributions: 
\begin{itemize}
    \item We model the auction as a $n$-player simultaneous move game with complete information that we name SAA-c. No specific assumption is made on the bidders' value functions. 
    \item We present an efficient bidding strategy ($SMS^\alpha$) that tackles simultaneously the \textit{exposure problem}, the \textit{own price effect}, \textit{budget constraints} and the \textit{eligibility management problem} in SAA-c. $SMS^\alpha$ is based on a Simultaneous Move Monte Carlo Tree Search (SM-MCTS) \cite{tak2014monte,bovsansky2016algorithms}. To the best of knowledge, it is the first algorithm that tackles the four main strategic issues of SAA. 
    \item We introduce a hyperparameter $\alpha$ in $SMS^\alpha$ which allows a bidder to arbitrate between expected utility and risk-aversion. 
    \item We propose a new method based on the convergence of a specific sequence for the prediction of closing prices in SAA-c. This prediction is then used to enhance the expansion and rollout phase of $SMS^\alpha$.  
    \item Through typical examples taken from the literature and extensive numerical experiments on instances of realistic size, we show that $SMS^\alpha$ outperforms state-of-the-art algorithms by achieving higher expected utility and better tackling the exposure problem and the own price effect in budget and eligibility constrained environments. 
\end{itemize}

The remainder of this paper is organised as follows. In Section \ref{SAA}, we define our model SAA-c and provide its game and strategic complexities. We then introduce our performance indicators. In Section \ref{prediction}, we present our method for the prediction of closing prices. In Section \ref{MCTS bidding strat}, we present our algorithm $SMS^\alpha$. In Section \ref{Experiments}, we show on typical examples taken from the literature  the empirical convergence of our method for the prediction of closing prices and that $SMS^\alpha$ tackles efficiently the four main strategic issues. Then, by comparing $SMS^\alpha$ to state-of-the-art algorithms, we show through extensive numerical experiments on instances of realistic size the main increase in performance of our solution. 

\section{Simultaneous Ascending Auction} \label{SAA}

\subsection{Simultaneous Ascending Auction model with complete information}

Simultaneous Ascending Auction (SAA) \cite{milgrom2000putting,cramton2006simultaneous,wellman2008bidding} is one of the most commonly used mechanism design where $m$ indivisible goods are sold via separate and concurrent English auctions between $n$ players. Bidding occurs in multiple rounds. At each round, players submit their bids simultaneously. The player having submitted the highest bid on an item $j$ becomes its temporary winner. If several players have submitted the same highest bid on item $j$, then the temporary winner is uniformly chosen at random amongst them. The \textit{bid price} of item $j$, noted $P_j$, is then set to the highest bid placed on it. The new temporary winners and bid prices are revealed to all players at the end of each round. The auction closes if no new bids have been submitted during a round. The items are then sold at their current bid price to their corresponding temporary winners. 

In our model, at the beginning of the auction, the bid price of each item is set to $0$. New bids are constrained to $P_j+\varepsilon$ where $\varepsilon$ is a fixed bid increment. This reduction of the bidding space is common in the literature on SAA \cite{goeree2014equilibrium,wellman2008bidding,pacaud2022monte}. We make the classical assumption that players won't bid on items that they are currently temporarily winning \cite{wellman2008bidding,pacaud2022monte}. Hence, in our model, a winner will always pay a price for an item at most $\varepsilon$ above the highest opponent bid.

Activity rules are introduced in SAA to penalise bidders which do not maintain a certain level of bidding activity. In our model, bidders are subject to the following simplified activity rule: the number of items temporarily won plus the number of new bids (also known as \textit{eligibility}) by a bidder can never rise \cite{goeree2014equilibrium, milgrom2004putting}. For instance, suppose a bidder $i$ is temporarily winning a set of items $Y$ and bids on a set of items $X$ at a given round. Its eligibility is defined as $e_i=|Y|+|X|$ and is revealed to all bidders at the end of the round. In the next round, if bidder $i$ is temporarily winning a set of items $Y'$, it can only bid on a set of items $X'$ of size $|X'|\leq e_i-|Y'|$. Its eligibility is then set to $e'_i=|X'|+|Y'|\leq e_i$. At the beginning of the auction, the eligibility of each player is set to $m$.

We assume that the value function $v_i$ and budget $b_i$ of each player $i$ are common knowledge \cite{szentes2003beyond,szentes2003three,pacaud2022monte}. Players are not allowed to bid on a set of items that exceeds their budget. In spectrum auctions, obtaining such knowledge about competitors is considered a very difficult task. Indeed, main mobile operators invest substantial effort to refine as much as possible their estimations. Nevertheless, the complete information framework remains still particularly interesting as it provides a strategic benchmark under ideal conditions. An efficient bidding strategy within this framework can be a significant asset for mobile companies as it facilitates the analysis of possible SAA scenarios based on point-wise estimates of opponents' private information.

This simplified version of SAA induces an $n$-player simultaneous move game with complete information that we name SAA-c. 

\subsection{Budgets, Utility and Value functions} \label{utility section}

A player $i$ in SAA-c is defined by its budget $b_i$, its value function $v_i$ and its utility function $\sigma_i$. Without loss of generality, $b_i$ and $v_i$ are chosen independently. If the current bid price vector is $P$, a player $i$ temporarily winning a set of items $Y$ with current eligibility $e_i$ can bid on a set of items $X$ if and only if 
\begin{equation}
    \begin{cases}
	|X|+|Y|\leq e_i\\
    \sum_{j\in X}(P_j+\varepsilon)\leq b_i-\sum_{j\in Y} P_j
	\end{cases}
\end{equation}

At the end of the auction, the utility obtained by player $i$ after winning the set of items $X$ at bid price vector $P$ is equal to its profit, i.e.: 

\begin{equation}
    \label{profit}
    \sigma_i(X,P)=v_i(X)-\sum_{j\in X} P_j 
\end{equation}

To respect common reinforcement learning conventions, we will sometimes denote by $R^\pi$ the random variable corresponding to the utility obtained by playing policy $\pi$.

Value functions are assumed to be normalised ($v_i(\emptyset)=0$), finite and verify the free disposal condition, i.e. for any two sets of goods $X$ and $Y$ such that $X\subset Y$, then $v(X)\leq v(Y)$ \cite{lehmann2006combinatorial,milgrom2000putting}. Two disjoint sets $X$ and $Y$ of goods are said to be complements if $v(X+Y)>v(X)+v(Y)$ \cite{wellman2008bidding}.  

\subsection{Extensive form}

\begin{figure*}[t]
\centering
\includegraphics[width=1.3\columnwidth]{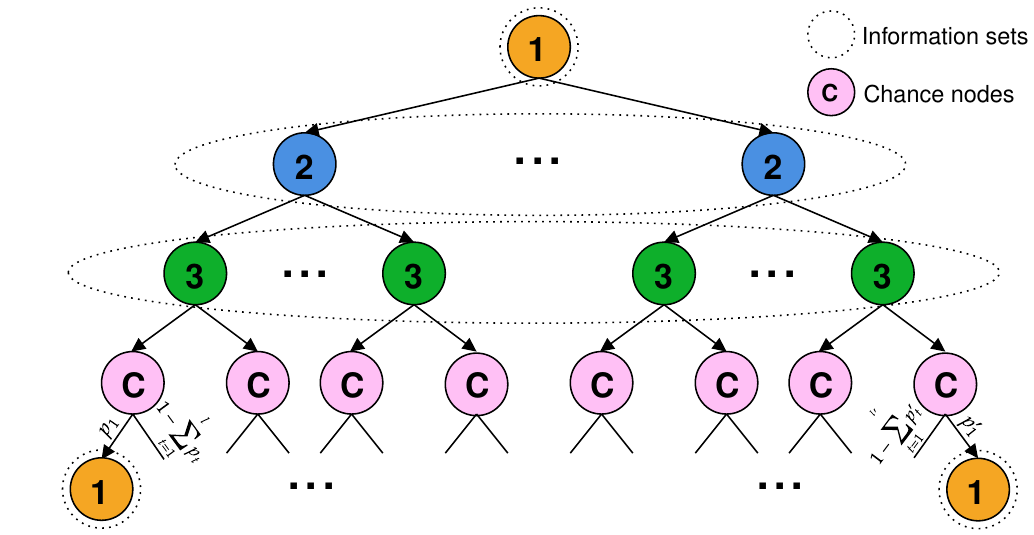}
\caption{Extensive form of a three player SAA-c game with information sets and chance nodes} 
\label{fig:SAA-c}
\end{figure*}

The standard representation for multi-round games is a tree representation named extensive form \cite{maschler2020game}. The game tree is a finite rooted directed tree admitting two types of nodes: \textit{decision nodes} and \textit{chance nodes}. At each decision node, a player has the choice between many actions each represented by a directed edge. A chance node has a fixed probabilistic distribution assigned over its outgoing edges. An information set is a set of decision nodes which are indistinguishable for the concerned player at the current position of the game \cite{cowling2012information}. This means that a player, given its current information, does not know exactly at which decision node it is playing. It only knows that it is playing at one of the decision nodes of the corresponding information set. Games where information sets are not all singletons are known as imperfect information games \cite{swiechowski2022monte}. 

We represent the SAA-c game in this form with the decision nodes representing the different states of the game and the chance nodes representing the random draws of temporary winners in case of ties. At each decision node, an outgoing edge represents a set of items on which the concerned player bids if it selects this edge. Each decision node or state is defined by five features: the concerned player, the eligibility vector revealed at the end of the last round, the temporary winner of each item, the current bid price vector and the bids already submitted during the current round. The four first features are common knowledge and the last feature is hidden information for the concerned player. Therefore, all decision nodes which differ only by the last feature belong to the same information set. In Figure \ref{fig:SAA-c}, we represent an SAA-c game between three players with their information sets and chance nodes.

\subsection{Game and strategic complexities}
\subsubsection{Game complexities}
To highlight the complexity of the SAA-c game, we focus on two metrics: \textit{information set space complexity} and \textit{game tree complexity} \cite{van2002games}. We define the first as the number of different information sets which can be legally reached in the game. It acts as a lower bound of the \textit{state space complexity} \cite{van2002games}. The second corresponds to the number of different paths in its extensive form. We compute both complexities for a given number of rounds $R$, unlimited budgets and without any activity rule.  

\begin{theorem}
Let $\Gamma$ be an instance of the SAA-c game with no activity rule. Let $n$, $m$ and $R$ be respectively the number of players, the number of items and the number of rounds in $\Gamma$. Suppose that all players have unlimited budgets. The number of possible information sets in $\Gamma$ is:
\begin{equation}
    n(Rn+1)^m
\end{equation}

\end{theorem}

\begin{proof}
    Each information set is defined by three components: the player to bid, the temporary winner and bid price of each item. If no player has bidded on an item, then it remains unsold and is handed back to the auctioneer. Otherwise, its bid price is included in $\{\varepsilon,2\varepsilon,...,R \varepsilon\}$ and the item is allocated to one of the $n$ players. Therefore, the number of different allocations and bid prices of an item in $\Gamma$ is $Rn+1$. Under the unlimited budget assumption, all items are mutually independent. Thus, the number of different allocations and bid prices for all items is $(Rn+1)^m$. As there are $n$ different players who can bid, the number of possible information sets is:
    $$n(Rn+1)^m$$
\end{proof}

\begin{theorem}
Let $\Gamma$ be an instance of the SAA-c game with no activity rule. Let $n$, $m$ and $R$ be respectively the number of players, the number of items, and the number of rounds in $\Gamma$. Suppose that all players have unlimited budgets. A lower bound of the game tree complexity of $\Gamma$ is:
\begin{equation}
    \Omega(2^{m(n-1)R})
\end{equation}

\end{theorem}

\begin{proof}
We consider $\Gamma$ with a deterministic tie-breaking rule. This eliminates chance nodes and reduces the number of paths in the game's extensive form. Let's first compute a lower bound of the number of different branches created in $\Gamma$ during a given round.

Suppose player $i$ is the temporary winner of $m_i$ items. Thus, during this given round, player $i$ can bid $2^{m-m_i}$ different ways as it can either bid or not bid on each of the remaining $m-m_i$ items. Hence, during this round, there are $2^{nm-\sum_{i=1}^n m_i}$ different bidding scenarios. Thus, this given round creates $2^{nm-\sum_{i=1}^n m_i}-1$ new branches all leading to non-terminal nodes of $\Gamma$. Moreover, as $\sum_{i=1}^n m_i\leq m$, the number of different branches created during any round is lower bounded by $2^{m(n-1)}-1$.

A lower bound of the game tree complexity of $\Gamma$ can then easily be calculated by induction. Indeed, every non-terminal node of $\Gamma$ starting a bidding round induces at least $2^{m(n-1)}-1$ new branches during this round. Therefore, the game tree complexity of $\Gamma$ is lower bounded by:

\begin{equation}
    \sum_{l=0}^R (2^{m(n-1)}-1)^l=\frac{(2^{m(n-1)}-1)^{R+1}-1}{2^{m(n-1)}}
\end{equation}

Thus, a lower bound of the game tree complexity of $\Gamma$ is $\Omega(2^{m(n-1)R})$.
\end{proof}

\begin{example}
An SAA for 12 spectrum licences (5G) between 5 telecommunication companies was held in Italy in 2018 and ended after 171 rounds \cite{5G_italy}. The number of possible information sets as well as a lower bound of the game tree complexity of the corresponding SAA-c game with no activity rule are respectively $10^{35}$ and $10^{2470}$.
\end{example}

Adding activity rules decreases the game tree complexity as a bidder can no longer bid on a set of items which exceeds its eligibility. However, it increases the information set space complexity as a new feature (eligibility) is added to every information set.

\subsubsection{Strategic complexities}
SAA-c game also admits a number of strategic issues. The four main ones are presented below. 

\begin{itemize}
    \item \textbf{Exposure:} It is a phenomenon which happens when a player tries to acquire a set of complementary items but ends up by paying too much for the subset it actually wins at the end of the auction. Hence, the player obtains a negative utility. For instance, Table~\ref{Exemple1} presents a well-studied example, see e.g. \cite{wellman2008bidding}, in a 2-item SAA-c game with a bid increment of $1$ between two players with unlimited budgets (referred to as Example 1). Player 1 considers both items as perfect substitutes, i.e. it values both items equally and desires to acquire only one of the two, while player 2 considers them as perfect complements, i.e. each item is worthless without the other and desires to acquire both of them. If player 1 is temporarily winning no items and the bid price of the cheapest item is lower than $11$, it should bid on it. Otherwise, it should pass. Hence, if player 2 decides to bid on both items, it will end up exposed as it won't be able to obtain both items for a price inferior to 22. Moreover, if after a few rounds, player 2 decides to give up an item, it will still end up by paying for the other item and, hence, incur a loss. 
    \begin{table}[h]
    \centering
    \caption{Example of exposure ($\varepsilon=1$)\label{Exemple1}}
    \resizebox{.4\textwidth}{!}{
    \begin{tabular}{|c c c c|} 
    \hline
    & $v(\{1\})$ & $v(\{2\})$ & $v(\{1,2\})$\\ [0.5ex] 
    \hline
    Player 1 & 12 & 12 & 12 \\ 
    Player 2 & 0 & 0 &  20\\
    \hline
    \end{tabular}
    }
    \end{table}

\item \textbf{Own price effect:} Competing on an item causes inevitably the rise of its bid price and, hence, the decrease in utility of all players wishing to acquire it. Thus, players have all a strong interest in maintaining the bid price of all items as low as possible. To avoid this rise, a player can concede items to its opponents hoping that they will not bid on the items it is temporarily winning in exchange. This strategy is known as \textit{demand reduction} \cite{weber1997making,ausubel2014demand}. Dividing items between players to avoid this issue is called \textit{collusion} \cite{brusco2002collusion}. No communication is allowed between players. In SAA-c, players should be able to use the common knowledge of valuations and budgets to agree on a same fair split of items to tackle this issue without any communication.

\item \textbf{Budget constraints:} Capping the maximum amount a bidder can spend during an auction can highly impact the auction's outcome. Indeed, it can prevent players from bidding on certain sets of items and be a source of exposure. Moreover, given this information, players can drastically change their bidding strategy. For instance, in the auction presented in Table \ref{Exemple1}, if player 1 and 2 have respectively a fixed budget of $8$ and $20$, player 2 should bid on both items as this situation no longer presents any risk of exposure. 

\item \textbf{Eligibility management:} Efficient management of its own eligibility is a key factor to ensure a favourable outcome. Bidding on a high number of items to maintain high eligibility induces the own price effect. However, reducing its eligibility to form collusions can trap a bidder in a vulnerable position if the other bidders do not behave as expected. Hence, a tradeoff must be found. 
\end{itemize}

\subsection{Performance indicators} \label{Perf indic}

The natural metric used to measure the performance of a strategy is the \textit{expected utility}. However, given the fact that a specific instance of a spectrum auction (i.e. same frequency bands, same operators, etc ...) is generally only held once and an operator just participates to a few different instances, comparing strategies only on the basis of their \textit{expected utility} is not sufficient. Indeed, given the huge amount of money involved, potential losses due to exposure should also be taken into account. To measure this risk, we decompose the expected utility  as follows: 

\begin{equation}\label{expected utility}
\begin{split}
    \hspace*{-0.2cm}\E(R^\pi)=\probP(R^\pi\geq0)\E(R^\pi|R^\pi\geq0)+
    \underbrace{\probP(R^\pi\textless 0)\E(R^\pi|R^\pi\textless 0)}_{\text{Exposure}}
\end{split}
\end{equation}
\noindent where $\pi$ is a policy and $R^\pi$ is a random reward obtained by playing $\pi$ in a SAA-c game. We introduce the term $-\probP(R^\pi\textless0)\E(R^\pi|R^\pi\textless0)$ as a metric of potential exposure which should be minimised. We name it \textit{expected exposure} and estimate it by taking the opposite of all losses incurred by a strategy divided by the number of plays. Moreover, we define the \textit{exposure frequency} as $\probP(R^\pi\textless0)$. This is estimated by the number of times a strategy incurs a loss divided by the number of plays. To increase its expected utility, one can either try to acquire a set of items with higher value or reduce the price paid for the items won. Hence, to highlight the rise in expected utility due to better tackling the own price effect, we use the \textit{average price paid per item won}. To ensure that a strategy divides efficiently items between bidders and that no item is returned to the auctioneer unnecessarily, we consider the \textit{ratio of items won}.

\section{Predicting closing prices} \label{prediction}

$SMS^\alpha$ is based on a SM-MCTS whose expansion and rollout phases rely on the following bidding strategy and prediction of closing prices, i.e., an estimation of the price of each item at the end of the auction. 

\subsection{Constrained point-price prediction bidding}

We start by extending the definition of point-price prediction bidding (\textit{PP}) \cite{wellman2008bidding} to budget and eligibility constrained environments. 
\begin{definition}\label{constr PP}
In a SAA-c game with $m$ objects and a current bid price vector $P$, a point-price prediction bidder with budget $b$, a current eligibility $e$, an initial prediction of closing prices $P^{init}$ and a set of temporarily won items $Y$ computes the subset of goods 
\begin{equation}
    X^*=\underset{\substack{X\subset \{\emptyset\}\cup\{1,...,m\}\backslash Y \\ \sum_{j\in X\cup Y} \rho_j(P^{init},P,Y)\leq b  \\ |X|+|Y|\leq e}}{\text{arg max}}\; \sigma(X\cup Y,\rho(P^{init},P,Y))
\end{equation}
breaking ties in favour of smaller subsets and lower-numbered goods. It then bids $P_j+\varepsilon$ on all items $j$ belonging to $X^*$. The function $\rho:(P^{init},P,Y)\rightarrow \mathbb{R_+}^m$ maps an initial prediction of closing prices, a current bid price vector and a set of items temporarily won to an estimation of closing prices. For any item $j$, it follows the below update rule:  
\begin{equation}
\rho_j(P^{init},P,Y)=\left\{
\begin{array}{ll}
      \max(P^{init}_j,P_j)  \quad &  \text{if }j\in Y\\
      \max(P^{init}_j,P_j+\varepsilon) \quad &   \text{otherwise}\\
\end{array} 
\right. 
\end{equation}
\end{definition}

A point-price prediction bidder only considers sets of items \textit{within budget} $b$ given its prediction of closing prices $\rho(P^{init},P,Y)$, i.e., only sets of items $X$ such that $\sum_{j\in X\cup Y} \rho_j(P^{init},P,Y)\leq b$. Moreover, it can only bid on sets of items which does not exceed its eligibility $e$.  

If closing prices are correctly estimated and independent of the bidding strategy, then playing \textit{PP} is optimal for a player. However, in practice, closing prices are usually tightly related to a player's bidding strategy. Playing \textit{PP} with a null prediction of closing prices ($P^{init}=0$) is known as straightforward bidding (SB) \cite{milgrom2000putting}. The efficiency of the bidding strategy \textit{PP} highly depends on the accuracy of the initial prediction of closing prices $P^{init}$. For instance, if $P^{init}$ largely underestimates the actual closing price of each item, then when the current bid price $P\geq P^{init}$ component-wise, playing \textit{PP} with initial prediction $P^{init}$ gives the same strategy as SB. However, if $P^{init}$ overestimates too much the actual closing price of each item, then the bidder might stop playing prematurely in order to avoid exposure.

\subsection{Computing an initial prediction of closing prices} \label{init prediction}

Several methods exist in the literature for computing an initial prediction of closing prices $P^{init}$ in budget constrained environments. However, they all seem to present some limitations in SAA-c. For instance, the well known Walrasian price equilibrium \cite{arrow1971general} does not always exist when preferences exhibit complementarities as it is the case in Example 1. Standard tâtonnement processes, such as the one used to compute \textit{expected price equilibrium} \cite{wellman2008bidding}, return the same price vector regardless of the auction's specificities (e.g., bid increment $\varepsilon$). The final prediction is then completely independent of the auction mechanism of SAA-c which is problematic. Computing an initial prediction by using only the outcomes of a single strategy profile is relevant only if bidders actually play according to this strategy profile. For instance, simulating SAA-c games where all bidders play SB and using the average closing prices as initial prediction is relevant if the actual bidders play SB. We propose hereafter a prediction method based on the convergence of a specific sequence which aims at tackling all of these issues.

\begin{conjecture}
Let $\Gamma$ be an instance of an SAA-c game. Let $f_{\Gamma}(P)$ be a random variable returning the closing prices of $\Gamma$ when all bidders play \textit{PP} with initial prediction $P$. The sequence $p_{t+1}=\frac{1}{t+1}\E[f_{\Gamma}(p_t)]+(1-\frac{1}{t+1})p_t$ with $p_0$ the null vector of prices converges to a unique element $p^*$.
\end{conjecture}

The fact that $f_{\Gamma}$ is a random variable comes from the tie-breaking rule which introduces stochasticity in $\Gamma$. By taking its expectation $\E[f_{\Gamma}(p_t)]$ at each iteration $t$, we ensure our deterministic sequence $p_t$ to always converge to the same fixed point $p^*$. Hence, all players using our method share the same prediction of closing prices $p^*$. In practice, we perform a Monte-Carlo estimation of $\E[f_{\Gamma}(p_t)]$ by simulating many SAA-c games. In small instances, it is possible to obtain a closed-form expression of $\E[f_{\Gamma}(p_t)]$ and, from that, prove the convergence of sequence $p_t$.

\begin{example}
Suppose that both players play \textit{PP} with $P^{init}=p_0$ in Example 1. During the first round, player 1 bids on item 1 and player 2 bids on both items. There is $50\%$ chance that player 1 temporarily wins item 1 and $50\%$ chance that player 2 temporarily wins item 1. If player 1 wins item 1 during the first round, player 2 bids on item 1 during the second round while player 1 passes. In the third round, player 1 bids on item 2 while player 2 passes. In the fourth round, player 2 bids on item 2 while player 1 passes. Hence, the bid price of item 1 (respectively item 2) is odd (respectively even) if temporarily won by player 1. When the bid price $P=(12,11)$ and both items are temporarily won by player 2, player 1 drops out of the auction as, by definition of \textit{PP}, it prefers smaller subsets of items for a same predicted utility. If player 2 wins item 1 during the first round, the bid price of item 1 (respectively item 2) is even (respectively even) if temporarily won by player 1. The closing price are then $P=(11,11)$. Therefore, $f_\Gamma(p_0)$ has 50\% chance of returning $(12,11)$ and 50\% chance of returning $(11,11)$. Hence, $\E[f_{\Gamma}(p_0)]=(11.5,11)$. By performing a similar analysis, we can show that $\forall p\in\mathbb{R}^2,\E[f_{\Gamma}(p)]\in [0,11.5]^2$ and obtain the following closed-form expression for any $p\in [0,11.5]^2$: 
\begin{equation}
{\E[f_{\Gamma}}(p)]=\left\{
\begin{array}{ll}
      (1,0)  \quad &  \text{if }p_1+p_2\geq 20\text{ and }p_1\leq p_2\\
      (0,1) \quad &   \text{if }p_1+p_2\geq 20\text{ and }p_1> p_2\\
      (11.5,11)  \quad &  \text{if }p_1+p_2<20\text{ and }p_1\leq p_2\\
      (11,11.5)  \quad &  \text{if }p_1+p_2<20\text{ and }p_1> p_2 \\     
\end{array} 
\right.
\end{equation}
From there, it is easy to show that sequence $p_t$ converges to $p^*=(10,10)$ in Example 1.    
\end{example}

The general proof of the conjecture is left for future work. 

Computing an initial prediction of closing prices as above has mainly three advantages compared to other methods in the literature. (1) We observe that this sequence converges in all undertaken SAA-c game instances. (2) This method takes into account the auction's mechanism through $f_\Gamma$. (3) This prediction of closing price is not based only on the outcomes of a single specific strategy profile. Indeed, depending on the value of $p_t$, different strategy profiles are used across iterations. At a fixed iteration $t$, a single strategy profile is used to compute $\E[f_\Gamma(p_t)]$ as the strategy returned by \textit{PP} only depends on its initial prediction $P^{init}=p_t$.

\section{SM-MCTS bidding strategy} \label{MCTS bidding strat}

\subsection{Brief presentation of MCTS}

Given the large state space and game tree complexities, it is practically impossible to explore the SAA-c game tree exhaustively as soon as we depart from very small instances. Thus, only a small portion of the game tree, called the search tree, can be explored. MCTS is a search technique that iteratively builds a search tree using simulations through a process named search iteration (see Figure~\ref{fig:MCTS scheme}). Each search iteration is divided into four steps. (1) The \textit{selection phase} selects a path from the root to a leaf node of the search tree. (2) The \textit{expansion phase} chooses one or more children to be added to the search tree from the selected leaf node according to the available actions. (3) The \textit{simulation phase} simulates the outcome of the game from the newly added node. (4) The \textit{backpropagation phase} propagates backwards the outcome of the game from the newly added node to the root in order to update the diverse statistics stored in each selected node of the search tree. This process is repeated until some predefined computational budget (time, memory, iteration constraint) is reached. Before running $SMS^\alpha$, we compute our initial prediction of closing prices $p^*$ as presented in Section \ref{init prediction}. 

\begin{figure*}[t]
\centering
\includegraphics[width=1.5\columnwidth]{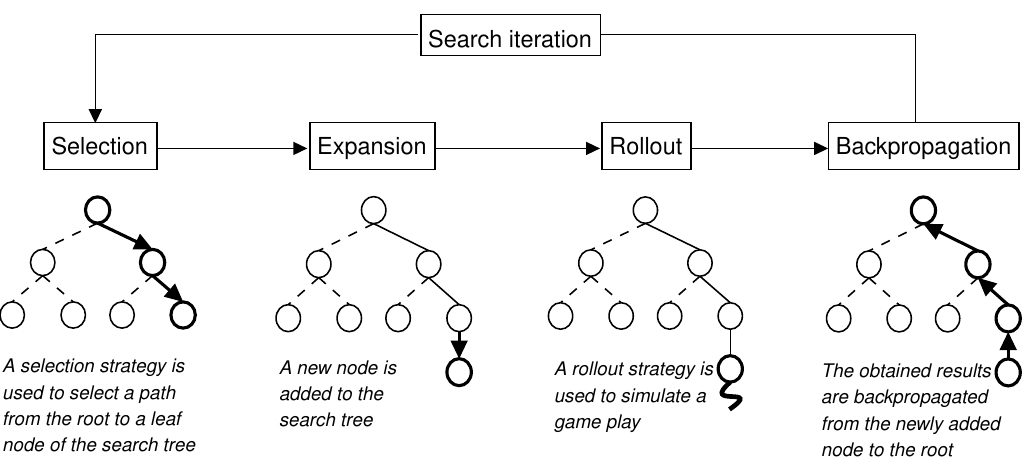}
\caption{MCTS scheme} 
\label{fig:MCTS scheme}
\end{figure*}

\subsection{Risk-averse rewards}

Ideally, one would want to maximise its expected profit while minimising its risk of exposure. However, achieving both objectives simultaneously may present potential conflicts. Indeed, taking risks can either be highly beneficial or lead to exposure depending on how the other players react. To do so, we introduce a new risk-averse reward incorporating both targets. For any strategy $\pi$, we define: 

\begin{equation}
    R_\alpha^\pi=(1+\alpha \mathbb{1}_{R^\pi<0})R^\pi
\end{equation}

where $\alpha$ is a hyperparameter which controls the risk aversion of $SMS^\alpha$. Note that 
\begin{equation}
 \E(R_\alpha^\pi)=\E(R^\pi)+\alpha \probP(R^\pi\textless0)\E(R^\pi|R^\pi\textless0)   
\end{equation}
where $\probP(R^\pi\textless0)\E(R^\pi|R^\pi\textless0)$ is the term corresponding to the losses induced by exposure in Equation~\eqref{expected utility}. Moreover, we define for any vector of price $P$ and any set of items $X$, $\sigma^\alpha(X,P)=(1+\alpha \mathbb{1}_{\sigma(X,P)<0})\sigma(X,P)$ which is a modified utility taking into account both of our objectives. We name it risk-averse utility. When the term utility is not explicitly specified as risk-averse, we refer to the utility described in Equation~\eqref{profit}.

The use of a linear scalarization function is a classical approach in multi-objective optimisation, multi-objective reinforcement learning~\cite{barrett2008learning}, constrained MDP~\cite{Lee2018MonteCarloTS} or POMDP~\cite{lee2018monte}.

\subsection{Search tree structure}

In order to maintain the simultaneous nature of SAA-c in the selection phase of $SMS^\alpha$, we use a Simultaneous Move MCTS (SM-MCTS) \cite{tak2014monte} (Figure~\ref{fig:Decoupled}). At each selection step, we select an $n$-tuple where each index $i$ corresponds to the action maximising the selection index of player $i$ given only its information set. By doing so, bids are selected simultaneously and independently. Each selection step corresponds to a complete bidding round of SAA-c. Hence, the depth of our search tree corresponds to how many rounds ahead $SMS^\alpha$ can foresee. The search tree nodes are defined by the eligibility of each bidder, the temporary winner and current bid price of each item. The vertices correspond to players' joint actions. Chance nodes are explicitly included in the search tree to break ties. There are three main advantages of using an SM-MCTS instead of an MCTS applied to a serialised game tree, i.e. turning SAA-c into a purely sequential game with perfect information. The first advantage is that it maintains the simultaneous move nature of SAA-c. The second advantage is that it does not increase the number of information sets making our learning process more efficient. The third advantage is that the number of selection steps to complete a bidding round of SAA-c is reduced from $n$ to $1$. Thus, the number of players $n$ is no longer a burden for planning a bidding strategy over many rounds.

\begin{figure*}[t]
\centering
\includegraphics[width=1.5\columnwidth]{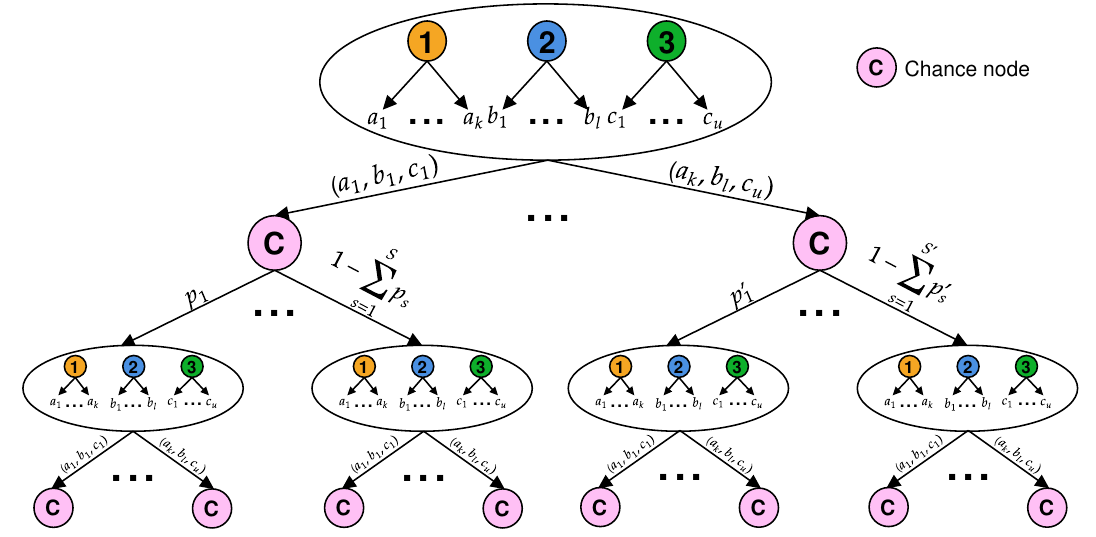}
\caption{SM-MCTS tree structure with explicit chance nodes for SAA-c game with 3 players} 
\label{fig:Decoupled}
\end{figure*}

\subsection{Selection}\label{selection}

At each selection step, players are asked to bid on the set of items which maximises their selection index. The selection phases ends when a terminal state of the SAA-c game or a non-expanded node, i.e. configuration of temporary winners, bid prices and eligibilities not yet added to the search tree, is reached. Our selection index is a direct application of the \textit{Upper Confidence bound applied to Trees} (UCT) \cite{kocsis2006bandit} to risk-averse rewards. Unlike usual applications of UCT, the size of the risk-averse reward support is unknown so we proceed to an online estimation of it. Each player $i$ chooses to bid on the set of items $x_i$ with highest score $q_{x_i}$ at information set $I_i$:

\begin{equation}
    \begin{split}
    q_{x_i}=\frac{r^\alpha_{x_i}}{n_{x_i}}+\max(c^\alpha_{x_i}-a^\alpha_{x_i},\varepsilon) \sqrt{\frac{2\log(\sum_{x'_i}n_{x'_i})}{n_{x_i}}}
    \end{split}
\end{equation}

where $r^\alpha_{x_i}$ is the sum of risk-averse rewards obtained after bidding on $x_i$ at $I_i$, $n_{x_i}$ is the number of times player $i$ has bidded on $x_i$ at $I_i$, $\varepsilon$ is the bid increment, $a^\alpha_{x_i}$ is the estimated lower bound and $c^\alpha_{x_i}$ is the estimated higher bound of the risk-averse reward support when bidding on $x_i$ at $I_i$. Thus, $\max(c^\alpha_{x_i}-a^\alpha_{x_i},\varepsilon)$ acts like the size of the risk-averse reward support when bidding on $x_i$ at $I_i$.

\subsection{Expansion}

The high branching factor due to the exponential growth of the game tree's width with the number of items $m$ prevents in-depth inspection of promising branches. Thus, it is necessary to reduce the action space at each information set of the search tree \cite{swiechowski2022monte}. To do so, each time a non-expanded node is added to the search tree, we select a maximum number $N_{act}$ of promising actions per information set. In order to obtain a relevant estimation of the risk of exposure when playing a specific move, it is important that we select individual actions and not joint actions. If a too low value is chosen for $N_{act}$, then a bidder can only play a very limited number of moves which affects its coordination with other bidders as well as its possibilities to avoid exposure. Thus, given thinking time constraints, a tradeoff between in-depth inspection and a wide range of possible moves must be found.

Passing its turn without bidding on any item is always included in the $N_{act}$ selected actions. This enables $SMS^\alpha$ to obtain shallow terminal nodes in its search tree which correspond to potential collusions between bidders and, thus, reduces the own price effect. The remaining $N_{act}-1$ actions correspond to the moves leading to the $N_{act}-1$ highest predicted utilities in strategy \textit{PP} with initial prediction $p^*$. Hence, the actions are chosen independently for each player, ensuring that the choices made for one player have no influence on the selections made by another. More formally, for each player $i$ at information set $I_i$ temporarily winning set of items $Y_i$ with eligibility $e_i$, the action of bidding on set of items $X_i$ is selected if $\sigma_i^\alpha(Y_i\cup X_i,\rho(p^*,P,Y_i))$ is one of the $N_{act}-1$ highest values with $P$ the current bid price. Only sets of items $X_i$ verifying $\sum_{j\in X_i\cup Y_i} \rho_j(p^*,P,Y_i)\leq b_i$ and $|X_i|+|Y_i|\leq e_i$ are considered. Statistics for each action are then initialised as follows:  
\begin{itemize}
    \item[$\ast$] $r^\alpha_{x_i}\gets 0$
    \item[$\ast$]  $n_{x_i}\gets 0$
    \item[$\ast$] $a^\alpha_{x_i}\gets +\infty$
     \item[$\ast$] $c^\alpha_{x_i}\gets -\infty$
\end{itemize}

\subsection{Rollout}

From the newly added node, an SAA-c game is simulated until the game ends. Players are asked to bid at each round of the rollout. The default strategy is usually to bid on a random set of items. However, it leads to absurd outcomes in this case with very high prices as player rarely all pass. Therefore, we propose an alternative approach. At the beginning of each rollout phase, we set $p^*_i=p^*+\eta_i$ with $\eta_i\sim U([-\varepsilon,\varepsilon]^m)$. Each player $i$ then plays \textit{PP} with initial prediction of closing prices $p^*_i$ during the entire rollout. Noise is added to our initial prediction $p^*$ to diversify players' bidding strategy and, hence, improve the quality of our sampling. To ensure that the prediction $p^*_i$ is always coherent with $p^*$ in all auctions, $\eta_i$ is chosen in $[-\varepsilon,\varepsilon]^m$. Hence, no absurd prediction can be drawn which could negatively impact our sampling. At the end of the rollout, an $n$-tuple is returned corresponding to the risk-averse utility obtained by each player. 

\subsection{Backpropagation}

The results obtained during the rollout phase are propagated backwards to update the statistics of the selected nodes. Let $V^\alpha_i$ be the risk-averse utility obtained by player $i$ at the end of the rollout. Let $x_i$ be the set of items on which player $i$ bidded at information state $I_i$ for one of the selected nodes. The statistics stored for $I_i$ are updated as follows:

\begin{itemize}
    \item[$\ast$] $r^\alpha_{x_i}\gets r^\alpha_{x_i}+V^\alpha_i$
    \item[$\ast$] $n_{x_i} \gets n_{x_i}+1$
    \item[$\ast$] $a^\alpha_{x_i}\gets \min(a^\alpha_{x_i},V^\alpha_i)$
     \item[$\ast$] $c^\alpha_{x_i}\gets \max(c^\alpha_{x_i},V^\alpha_i)$
\end{itemize}

\subsection{Transposition table}\label{transposition table}

Transposition tables are a common search enhancement used to considerably reduce the size of the search tree and improve performance of MCTS within the same computational budget \cite{childs2008transpositions}. By using such tables, we prevent the expansion of redundant nodes in our search tree and share the same statistics between transposed information states. This results in a significant improvement in performance of $SMS^\alpha$ for the same amount of thinking time. 

To identify each information set in the search tree, our hash function is based on two functions $h_1$ and $h_2$. The first returns a different integer for each combination of bid prices and allocations. The second returns a different integer for each eligibility vector. Hence, our hash function assigns a unique value to each information set in the search tree. More precisely, due to computational constraints, we can only assign a unique value for every node in the search tree with a depth lower than $R_{max}$. $R_{max}$ is a hyperparameter corresponding to an upper bound of the maximal depth (or rounds) in the final search tree. An example of function $h_1$ assigning a different integer for each combination of bid prices and allocations in a search tree of maximal depth $R_{max}$ is given in Algorithm \ref{alg:hash_d-SAA}. It uses as inputs the bid price vector $P^0$ at the root of the search tree, the bid price vector $P$ and the temporary winner $A_j$ of each item $j$ at a given node. If $A_j=0$, then item $j$ is temporarily allocated to the auctioneer. 

In practice, given the thinking time constraints in our experimental results, choosing $R_{max}=10$ is more than sufficient to guarantee a final search tree with maximal depth lower than $R_{max}$. Hence, our hash function acts as a perfect hash function as no type-1 error or type-2 error occurs \cite{zobrist1990new}.

\begin{algorithm}
\caption{Example of function $h_1$}\label{alg:hash_d-SAA}
\begin{algorithmic}
\State \textbf{Inputs Game:} $n$, $m$, $\varepsilon$
\State \textbf{Inputs Root Node:} Bid price vector $P^0$ 
\State \textbf{Inputs Node:} Bid price vector $P$, Allocation vector $A$ 
\State \textbf{Hyperparameter:} $R_{max}$
\State $h=0$
\State $step=R_{max}\times n$
\For{$j=1,2,...,m$}
\If{$A_j>0$} 
\small{$h+=(R_{max}\times (A_j-1)+\frac{P_j-P^0_j}{\varepsilon})step^{j-1}$}
\EndIf
\EndFor \\
\Return {h}
\end{algorithmic}
\end{algorithm}

\subsection{Final move selection}

The final move which is returned by $SMS^\alpha$ is the action which maximises the player's expected risk-averse reward at the root node. More formally, $SMS^\alpha$ returns $\underset{x_i}{\text{arg max}} \frac{r_{x_i}^\alpha}{n_{x_i}}$ for player $i$.

\section{Experiments} \label{Experiments}

In this section, we start by analysing the convergence rates of sequence $p_t$, notably through Example 1. Then, we show that our algorithm $SMS^\alpha$ largely outperforms state-of-the-art existing bidding algorithms in SAA-c, mainly by tackling own price effect and exposure more efficiently. This is first shown through typical examples taken from the literature and, then, through extensive experiments on instances of realistic size. We compare $SMS^\alpha$ to the following four strategies: 
\begin{itemize}
    \item $MS^\lambda$: An MCTS algorithm described in \cite{pacaud2022monte} which relies on two risk-aversion hyperparameters $\lambda^r$ and $\lambda^o$. 
    \item EPE: A \textit{PP} strategy using expected price equilibrium \cite{wellman2008bidding} as initial prediction. 
    \item SCPD: A distribution price prediction strategy using self-confirming price distribution \cite{wellman2008bidding} as initial distribution prediction.
    \item SB: Straightforward bidding \cite{milgrom2000putting}.
\end{itemize}

The four strategies $MS^\lambda$, EPE, SCPD and SB initially rely on the definition of \textit{PP} for unconstrained environments \cite{wellman2008bidding}. We extend them to budget and eligibility constrained environments in the same way as it is done in Definition \ref{constr PP}. In all experiments, none of the bidders are aware of their opponents' strategy.

Each algorithm is given respectively $150$ seconds of thinking time. Initial prediction of closing prices are done offline before the auction starts and, therefore, are excluded from the thinking time. This step usually takes a few minutes. All experiments are run on a server consisting of Intel\textregistered Xeon\textregistered E5-2699 v4 2.2GHz processors. In all upcoming experiments, the hyperparameter $\alpha$ of $SMS^\alpha$ takes the value $7$ and the risk-aversion hyperparameters $\lambda^r$ and $\lambda^o$ of $MS^\lambda$ both take the value $0.025$. The maximum number of expanded actions per information set $N_{act}$ of $SMS^\alpha$ is set to $20$. The choice of the hyperparameters is motivated in Section \ref{choice hyperparameters}.

\subsection{Convergence of sequence $p_t$}

One of the main advantages of using our method to compute an initial prediction is the convergence of sequence $p_t$. Even though this convergence has only been observed and not proven, it is possible to derive rates of convergence in small instances. For instance, in Example 1, it can be shown that $\forall t\geq 1$, $p_t$ belongs to the diamond defined by the points $(10-\frac{10}{t},10-\frac{9}{t})$, $(10-\frac{9}{t},10-\frac{10}{t})$, $(10+\frac{7}{4t},10+\frac{3}{4t})$ and $(10+\frac{3}{4t},10+\frac{7}{4t})$ which converges to $p^*=(10,10)$. We represent in Figure \ref{fig:Guarantee} the sequence $p_t^1$ with its corresponding lower bound and upper bound.

\begin{figure}[t]
\centering
\includegraphics[width=0.75 \columnwidth]{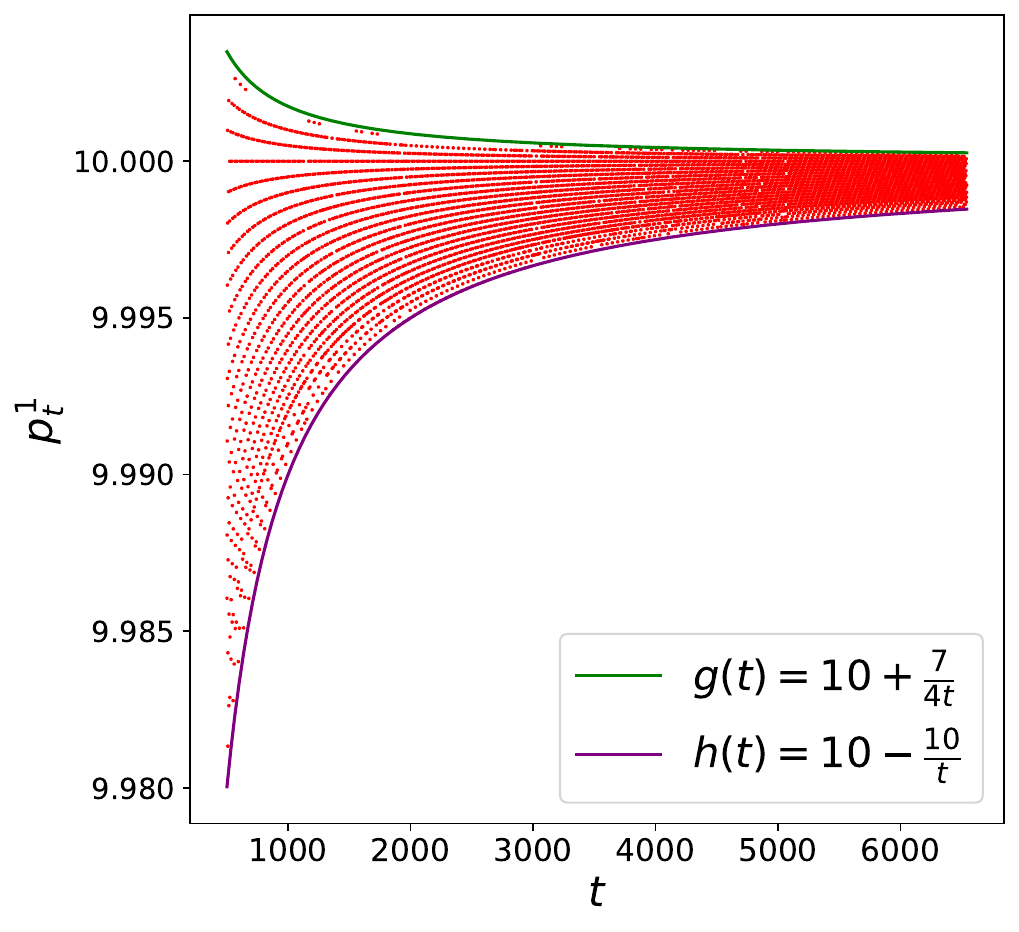}
\caption{Convergence of sequence $p_t^1$ with its respective upper bound $g(t)$ and lower bound $h(t)$ in Example 1.} 
\label{fig:Guarantee}
\end{figure}

In larger instances, we observe similar rates of convergence. However, computing such bounds seems unrealistic as obtaining a closed-form expression of $\E[f_\Gamma(p_t)]$ seems untractable.

\subsection{Test experiments}

One of the greatest advantages that MCTS methods have over other bidding algorithms is the capacity to judge pertinently in which situations adopting a demand reduction strategy is more beneficial. Indeed, through the use of its search tree, an MCTS method is capable of determining if it is more profitable to concede items to its opponents to keep prices low or to bid greedily. To highlight this feature, we propose the following experiment in a 2-item auction between two players with additive value functions. Each player values each item at $l=10$. Player 1 has a budget $b_1\geq 20$. Given that, the optimal strategy for player 2 is to bid on the cheapest item if it is not temporarily winning any item. Otherwise, it should pass. The optimal strategy for player 1 fully depends on its opponent's budget $b_2$. For an infinitesimal bid increment $\varepsilon$,

\begin{itemize}
    \item If $b_2\leq\frac{l}{2}$, player 1's optimal strategy is to play straightforwardly and it obtains an expected utility of $l-2b_2$.
    \item If $b_2\geq\frac{l}{2}$, player 1's should adopt a demand reduction strategy and it obtains an expected utility of $l$.
\end{itemize}

We plot in Figure \ref{fig:Test_exp} the expected utility $\E(\sigma_1)$ of player 1 for each strategy given player 2's budget $b_2$. The three algorithms SB, EPE, SCPD always suggest to player 1 to bid greedily and never propose a demand reduction strategy even when it is highly profitable ($b_2>\frac{l}{2}$). However, both MCTS methods perfectly adopt the appropriate strategy. This experiment highlights the fact that $SMS^\alpha$ selects the most profitable strategy and tackles own price effect, at least in simple budget and eligibility constrained environments. 

\begin{figure}[h]
\centering
\includegraphics[width=0.78 \columnwidth]{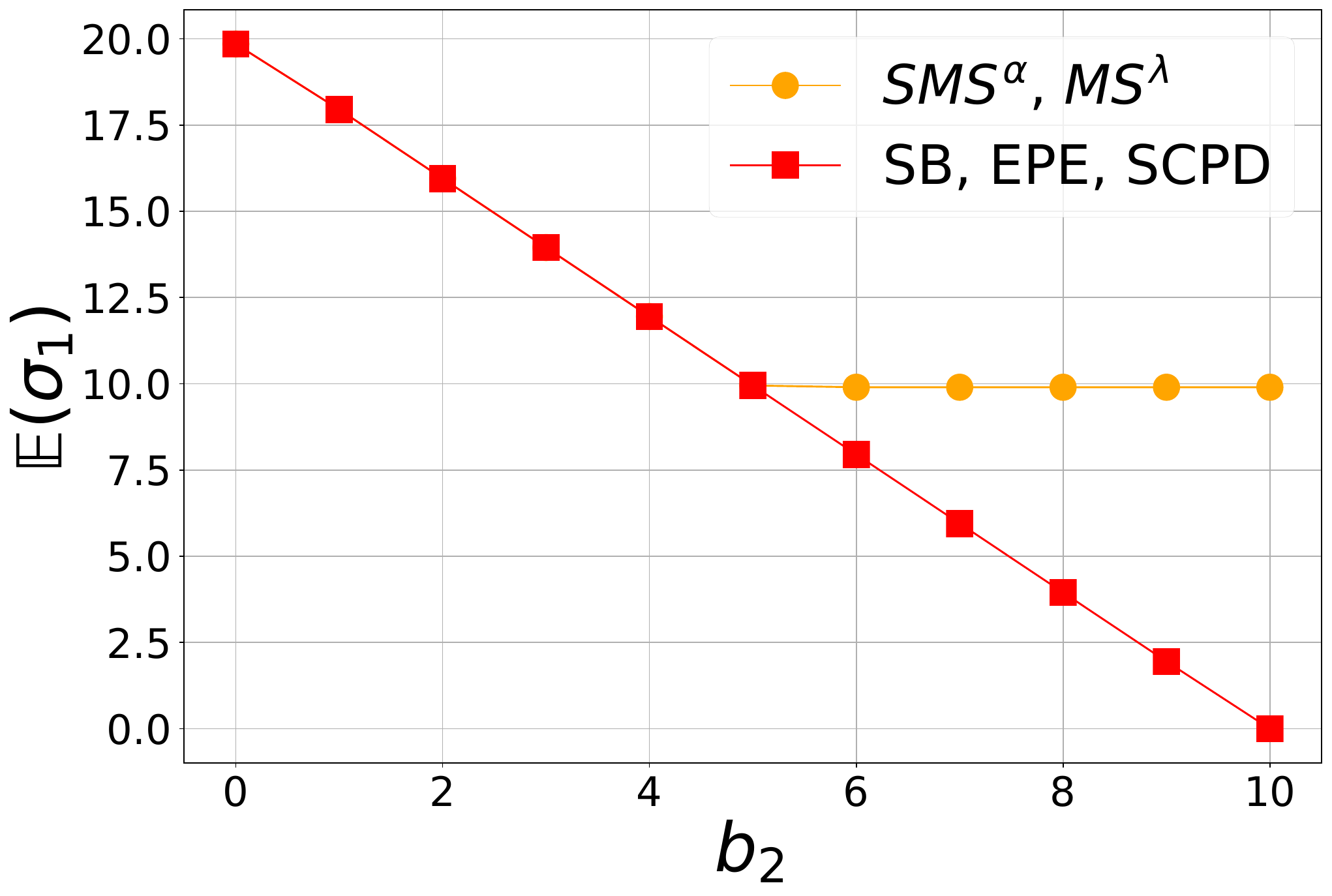}
\caption{Evolution of player 1's expected utility $\E(\sigma_1)$ depending on strategy versus player 2's budget $b_2$ given that player 2 plays optimally ($\varepsilon=0.1$).} 
\label{fig:Test_exp}
\end{figure}

Furthermore, $SMS^\alpha$ is capable of avoiding obvious exposure. To highlight this feature, we use the SAA-c game presented in Example 1 where player 2's budget $b_2=16$. The optimal strategy for player 1 is to play straightforwardly. Similarly to the preceding experiment, the optimal strategy for player 2 fully depends on its opponent's budget $b_1$. 
\begin{itemize}
    \item If $b_1<8$, player 2's optimal strategy is to play straightforwardly. 
    \item If $b_1\geq 8$, player 2's optimal strategy is to drop out of the auction to avoid exposure. 
\end{itemize}

We plot in Figure \ref{fig:Test_exp_exposure} the expected utility $\E(\sigma_2)$ of player 2 for each strategy given player 1's budget $b_1$. The two algorithms SCPD and SB always suggest to player 2 to bid straightforwardly leading player 2 to exposure when $b_1\geq 8$. $MS^\lambda$ never leads player 2 to exposure. However, it suggests to drop out prematurely of the auction in some situations with no risk of exposure and, hence, incurs a loss of easy profit ($b_1=7$). $SMS^\alpha$ and $EPE$ perfectly adopt the optimal strategy. This experiment highlights the fact that $SMS^\alpha$ perfectly adopts the most profitable strategy and tackles efficiently exposure, at least in simple budget and eligibility constrained environments.

\begin{figure}[ht]
\centering
\includegraphics[width=0.78 \columnwidth]{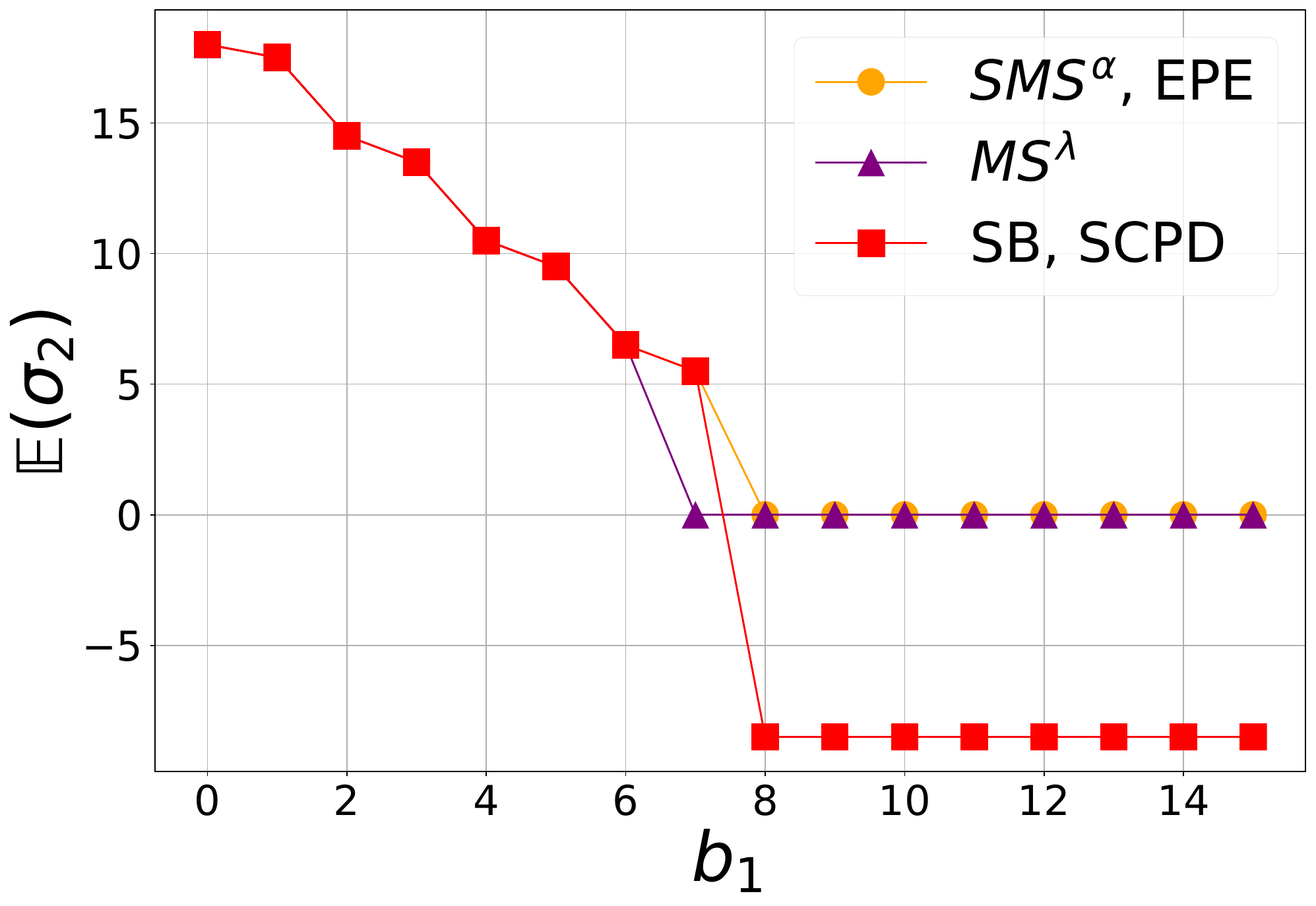}
\caption{Evolution of player 2's expected utility $\E(\sigma_2)$ depending on strategy versus player 1's budget $b_1$ given that player 1 plays optimally in Example 1 ($\varepsilon=1$).} 
\label{fig:Test_exp_exposure}
\end{figure}

\subsection{Extensive experiments} \label{extensive experiments}

In this section, we study instances of realistic size with $n=4$ and $m=11$. Each experimental result has been run on 1000 different SAA-c instances. With the exception of \cite{pacaud2022monte}, all experimental results in the literature are obtained for specific settings of SAA, i.e., using value functions with some specific property such as superadditivity \cite{goeree2014equilibrium,wellman2008bidding,reeves2005exploring}. Hence, it is difficult to conclude on the effectiveness of a method in more generic settings. Therefore, we propose a more general approach to generate value functions by making no additional assumption on its form. Budgets are drawn randomly. 

\begin{setting}\label{general setting}
Let $\Gamma$ be an instance of SAA-c with $n$ bidders, $m$ items and bid increment $\varepsilon$. Each player $i$ has a budget $b_i\sim U([b_{min},b_{max}])$ with $U$ the uniform distribution. Its value function $v_i$ is built as follows: $v_i(\emptyset)=0$ and, for any set of goods $X$,
\begin{equation}
    v_i(X)\sim U([\max_{j\in X}v_i(X\backslash \{j\}),V+\max_{j\in X}v_i(X\backslash \{j\})+v_i(\{j\})])
\end{equation}
\end{setting}

Drawing value functions through a uniform distribution is widely used for creating auction instances \cite{wellman2008bidding,reeves2005exploring}. In our setting, the lower-bound ensures that $v_i$ respects the \textit{free disposal} \cite{milgrom2000putting} condition. The upper-bound caps the maximum surplus of complementarity possibly gained by adding an item $j$ to the set of goods $X\backslash \{j\}$ by $V$. As valuations are always finite, any value function can be represented by our setting for a sufficiently large $V$. For $V=0$, only subadditive functions are considered. For $V>0$, goods can either be complements or substitutes. In our experimental results, value functions and budgets are generated for each instance as above with $\varepsilon=1$, $b_{min}=10$, $b_{max}=40$ and $V=5$. 

In the upcoming analysis, the average price paid per item won, the ratio of items won, the expected exposure and the exposure frequency are obtained by confronting a strategy $A$ to a strategy $B$. To facilitate our study, each measure of $A$ against $B$ is obtained by averaging the results obtained for the three following strategy profiles: ($A$,$B$,$B$,$B$), ($A$,$A$,$B$,$B$) and ($A$,$A$,$A$,$B$). For instance, if $A=SMS^\alpha$ and $B=\text{SB}$, the
average price paid per item won by $SMS^\alpha$ in these three strategy profiles is respectively: $5.96$, $5.46$ and $4.62$. Hence, the average price payer per item won by $SMS^\alpha$ against $SB$ is $5.35$. 

\subsubsection{Expected Utility}

\begin{figure*}[!t]
\centering
\subfloat[$SMS^\alpha$ vs $MS^\lambda$]{\includegraphics[width=2.5in]{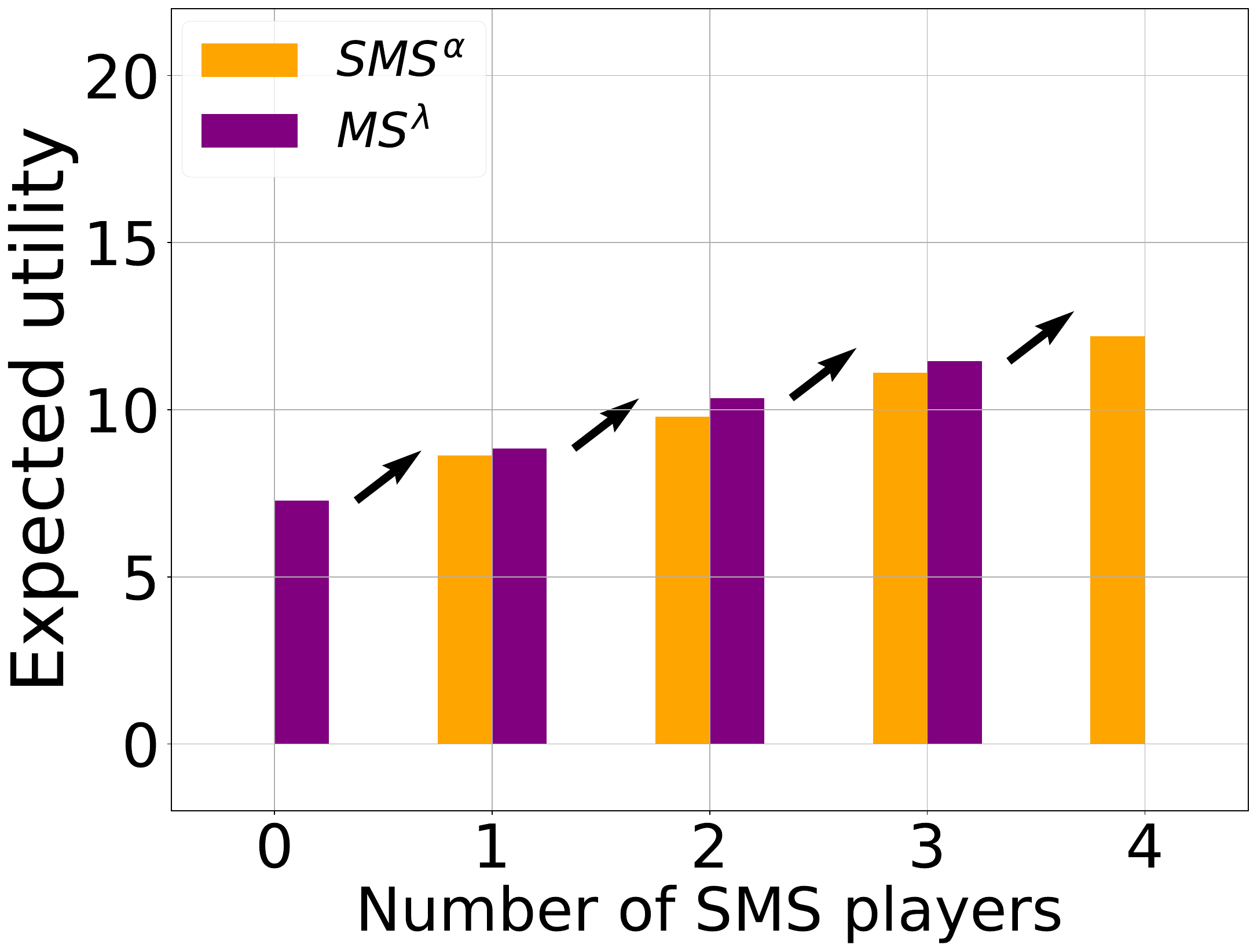}%
\label{1}}
\hfil
\subfloat[$SMS^\alpha$ vs EPE]{\includegraphics[width=2.5in]{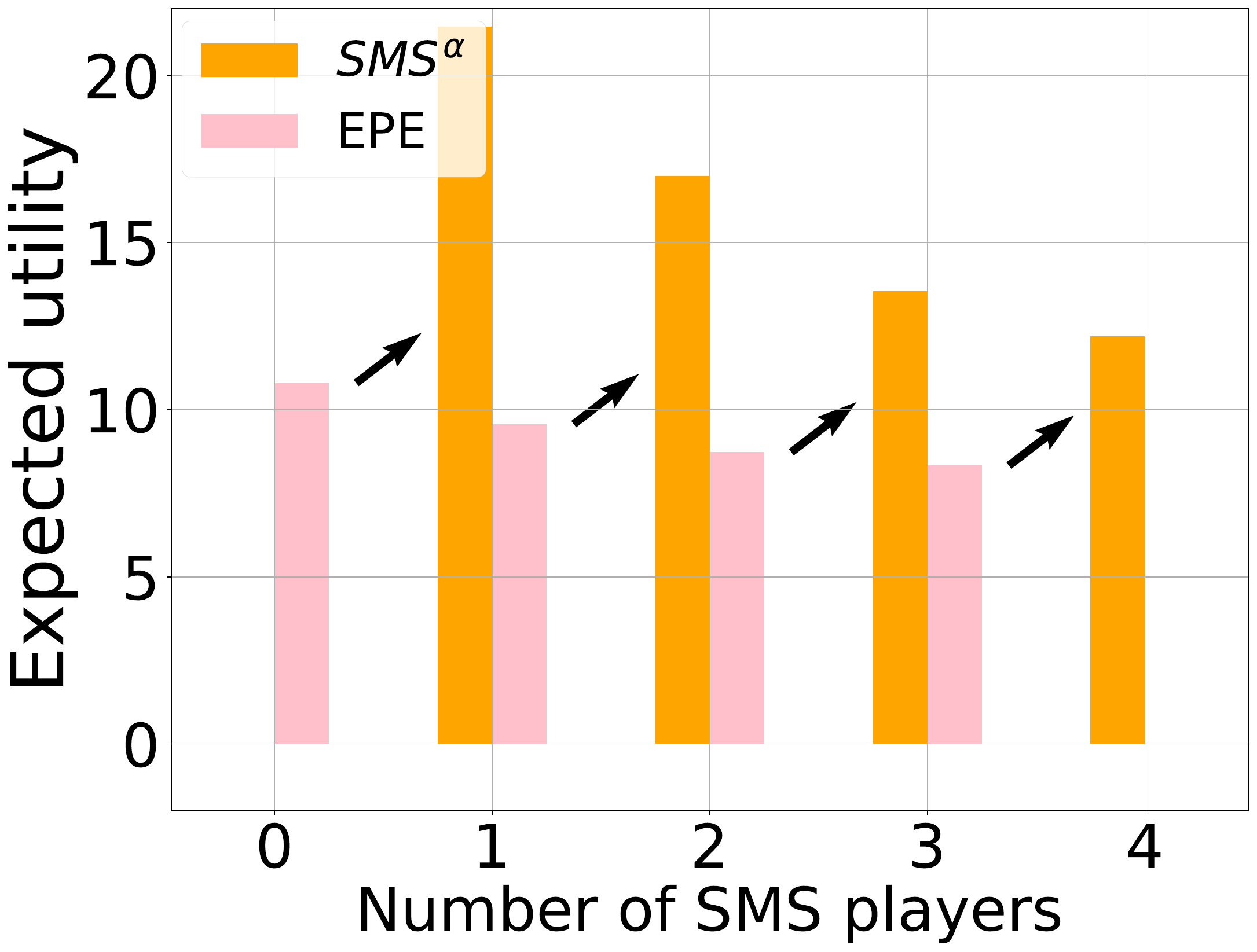}%
\label{Expected utility EPE}}

\subfloat[$SMS^\alpha$ vs SCPD]{\includegraphics[width=2.5in]{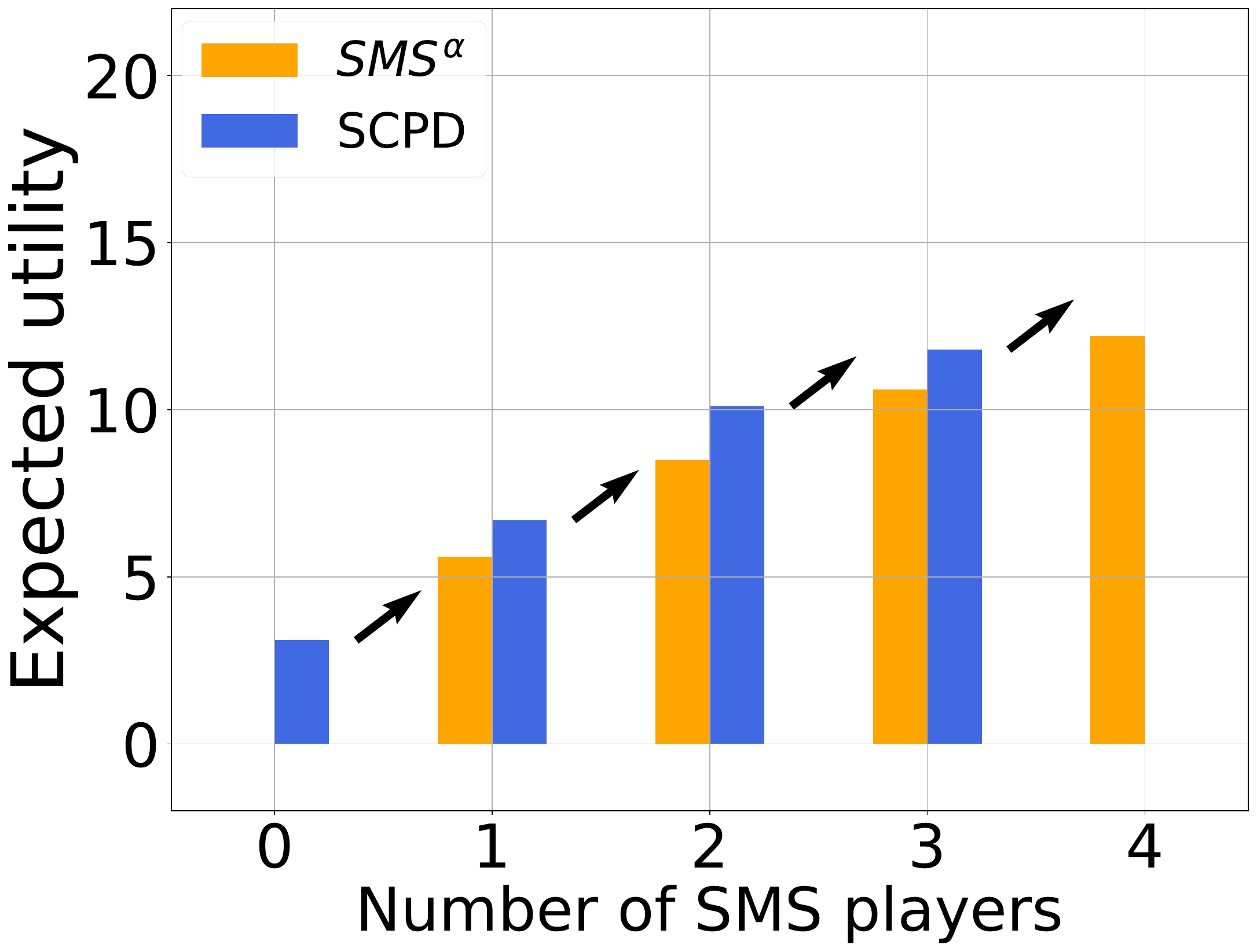}%
\label{4}}
\hfil
\subfloat[$SMS^\alpha$ vs SB]{\includegraphics[width=2.5in]{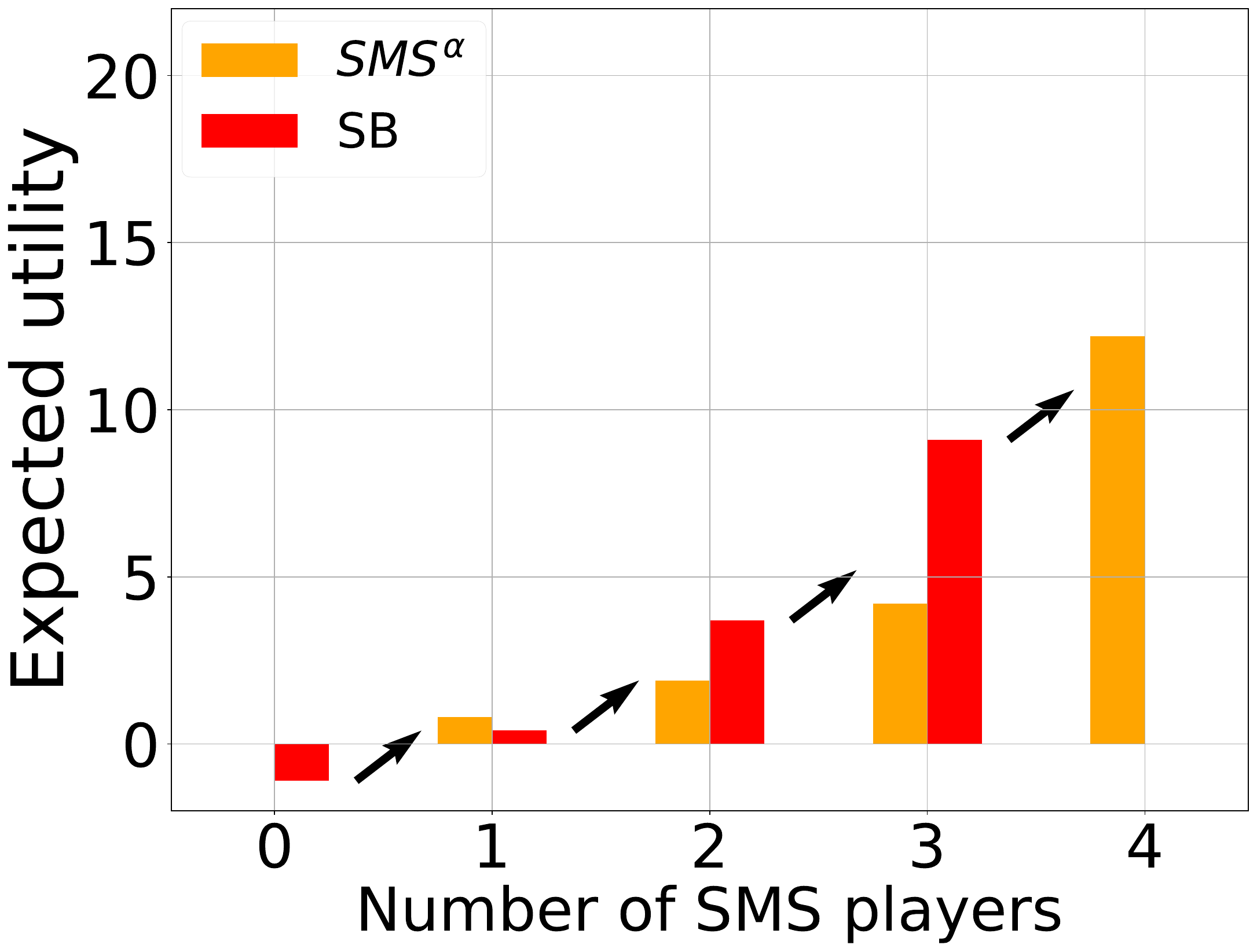}%
\label{5}}
\caption{Normal-form expected utility for an SAA-c game with five 
strategies ($n=4$, $m=11$, $\varepsilon=1$)}
\label{Expected utility}
\end{figure*}

To facilitate our analysis, we study the normal form game in expected utility where each player has the choice between playing $SMS^\alpha$ or another strategy $A$. The same empirical game analysis approach was employed by Wellman et al. in \cite{wellman2008bidding}. More precisely, we map each strategy profile to the estimated expected utility obtained by each player in the 1000 SAA-c instances. The four resulting empirical games for each possible strategy $A$ are given in Figure \ref{Expected utility}.

For example, in Figure \ref{Expected utility EPE}, if all bidders play EPE, each bidder obtains an expected utility of $10.8$. In the case of three EPE bidders and one $SMS^\alpha$ bidder, the $SMS^\alpha$ bidder obtains an expected utility of $21.5$. Hence, if all bidders play EPE, a bidder can double its expected utility by switching to $SMS^\alpha$. Therefore, deviating to $SMS^\alpha$ is profitable if all bidders play EPE. This is also the case for the three other possible deviations in Figure \ref{Expected utility EPE}. Hence, in the empirical game where bidders have the choice between playing $SMS^\alpha$ or EPE, each bidder has interest in playing $SMS^\alpha$. We can clearly see that all deviations to $SMS^\alpha$ are also strictly profitable in the three other empirical games. Hence, in each empirical game, a bidder should play $SMS^\alpha$ to maximise its expected utility. Therefore, the strategy profile ($SMS^\alpha$, $SMS^\alpha$, $SMS^\alpha$, $SMS^\alpha$) is a Nash equilibrium of the normal-form SAA-c game in expected utility with strategy set \{$SMS^\alpha$, $MS^\lambda$, EPE, SCPD, SB\}.

Moreover, the strategy profile where all bidders play $SMS^\alpha$ has a substantially higher expected utility than any other strategy profile where all bidders play the same strategy. This is mainly due to the fact that $SMS^\alpha$ tackles efficiently the own price effect. For instance, in Figure \ref{Expected utility}, the expected utility of the strategy profile where all bidders play $SMS^\alpha$ is respectively $1.13$, $1.68$ and $3.94$ times higher than the ones where all bidders play EPE, $MS^\lambda$ and SCPD. 

The fact that the expected utility obtained by the strategy profile where all bidders play EPE is relatively close to the one where all bidders play $SMS^\alpha$ can be explained as follows. To compute their expected price equilibrium as initial prediction of closing prices, all EPE bidders in our experiments share the same initial price vector and adjustment parameter in their tâtonnement process. This tâtonnement process is independent of the auction's mechanism and only relies on the estimated valuations of the players. Hence, as SAA-c is a game with complete information, all EPE bidders share the same initial prediction of closing prices and can therefore split up the items between them more or less efficiently.

Not all algorithms have the ability of achieving good coordination between bidders. For instance, the strategy profile where all bidders play SB leads to a negative expected utility. Hence, in this specific case, bidders would have preferred not to participate in the auction. This highlights the fact that playing SB is a very risky strategy and mainly leads to exposure. 

We believe that $SMS^\alpha$ outperforms the four other strategies for the three following reasons:
\begin{itemize}
    \item it better judges when to perform demand reduction or to bid greedily.
    \item it better tackles the own price effect without putting itself in a vulnerable position because of eligibility constraints.
    \item it better tackles the exposure problem.
\end{itemize}

\subsubsection{Own Price Effect}

\begin{figure*}[!t]
\centering
\subfloat[Average price paid per item won]{\includegraphics[width=3in]{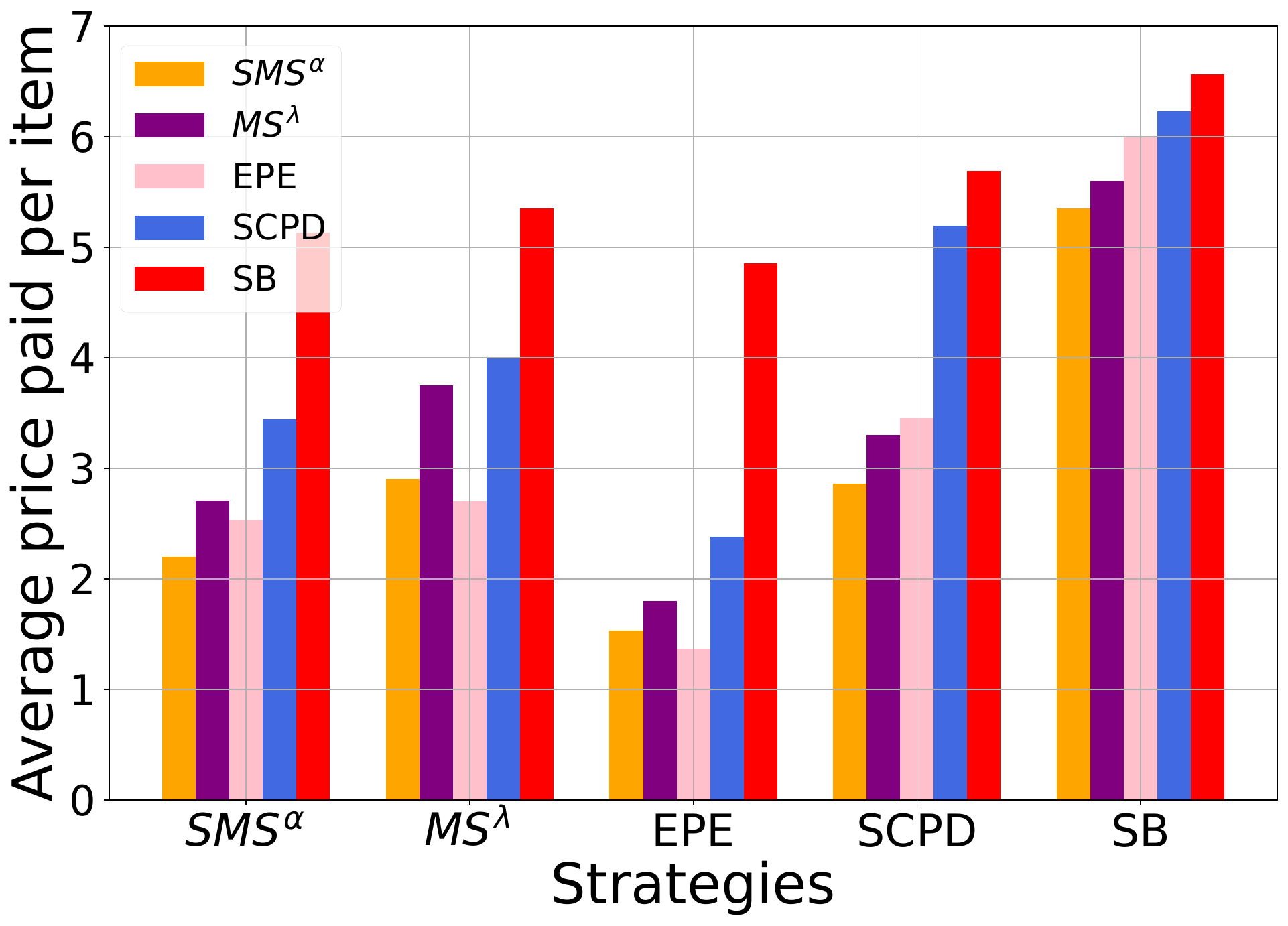}%
\label{avg price}}
\hfil
\subfloat[Ratio of items won]{\includegraphics[width=3.1in]{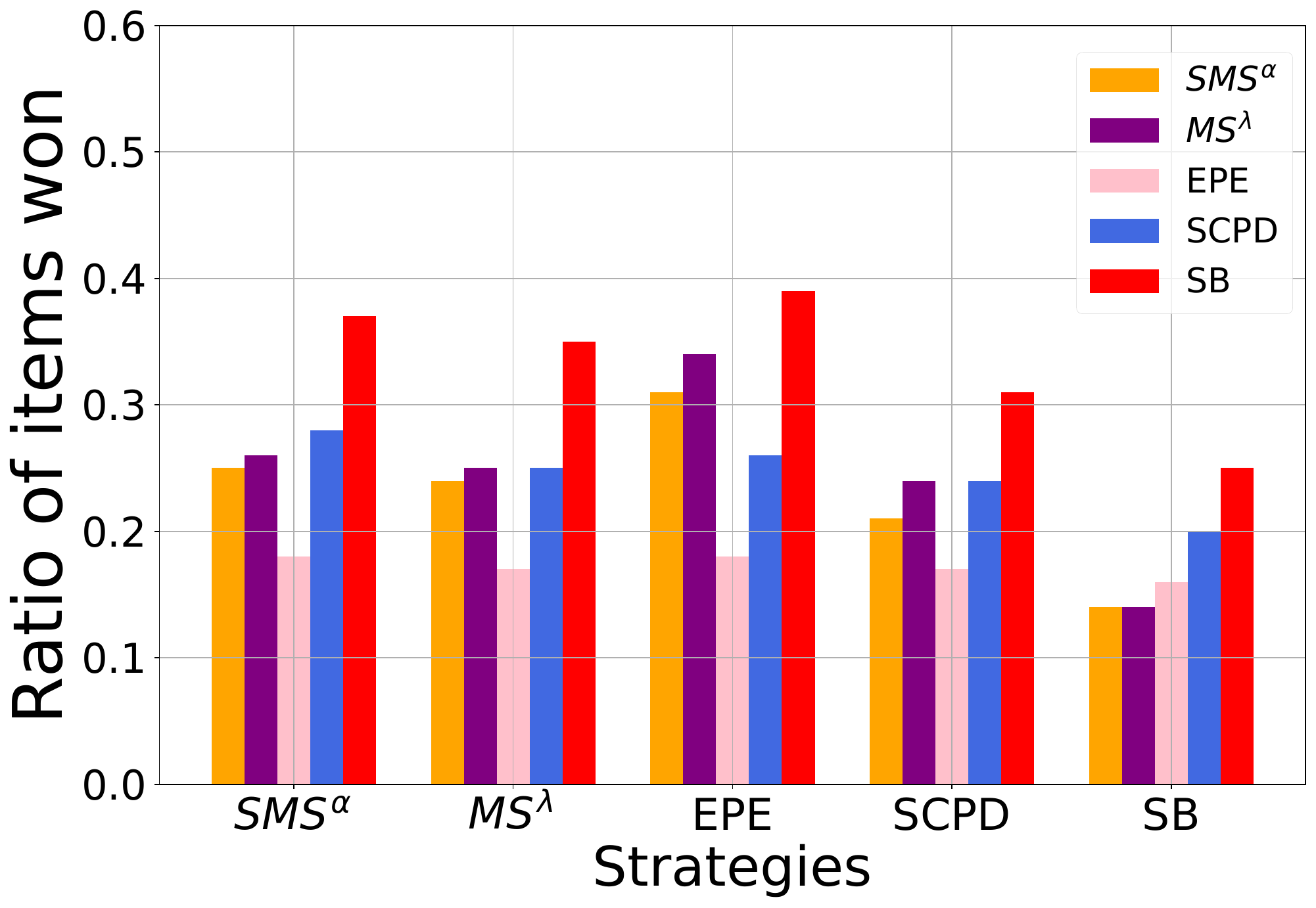}%
\label{ratio itm}}
\caption{Own price effect analysis for an SAA-c game with five strategies ($n=4$, $m=11$, $\varepsilon=1$)}
\label{fig:Own_price_effect}
\end{figure*}

\begin{figure*}[!t]
\centering
\subfloat[Expected exposure]{\includegraphics[width=3.1in]{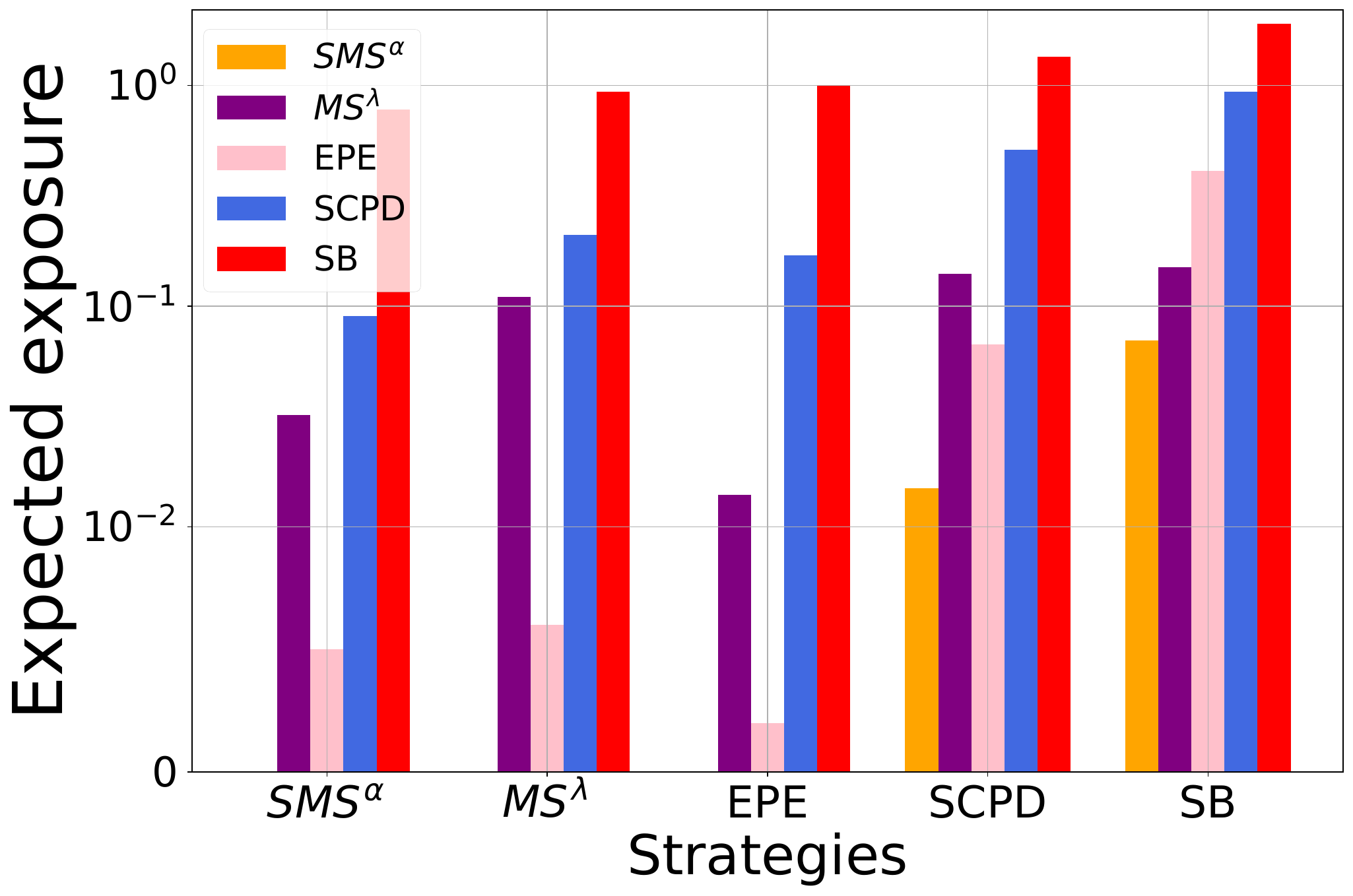}%
\label{fig:Estimated Exposure}}
\hfil
\subfloat[Exposure frequency]{\includegraphics[width=3in]{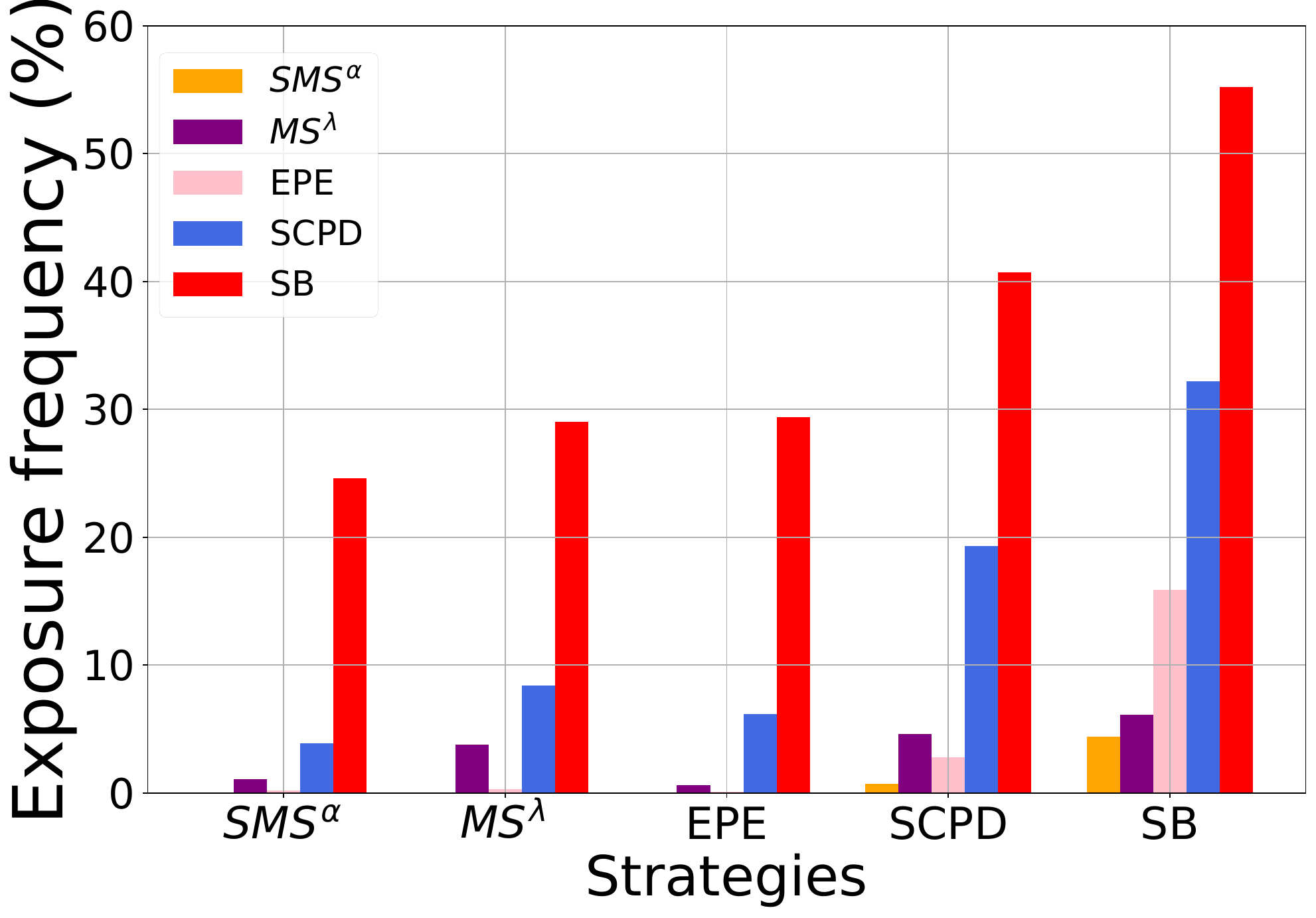}%
\label{fig:Exposure frequency}}
\caption{Exposure analysis for an SAA-c game with five strategies ($n=4$, $m=11$, $\varepsilon=1$)}
\label{fig:Exposure}
\end{figure*}

To analyse the own price effect, we plot in Figure \ref{avg price} the average price paid per item won and in Figure ~\ref{ratio itm} the ratio of items won by each strategy $A$ against every strategy $B$ displayed on the x-axis. For instance, if $A=SMS^\alpha$ and $B=\text{EPE}$, the average price paid per item won and the ratio of items won by $SMS^\alpha$ against EPE are respectively $1.53$ and $0.31$. They both correspond respectively  to the orange bar above index EPE on the x-axis. If $A=\text{EPE}$ and $B=SMS^\alpha$, then the average price paid per item won and the ratio of items won by EPE against $SMS^\alpha$ are respectively $2.53$ and $0.18$. They both correspond respectively to the pink bar above index $SMS^\alpha$ on the x-axis.

In Figure \ref{avg price}, we can clearly see that $SMS^\alpha$ acquires items at a lower price in average than the other strategies against $SMS^\alpha$, SCPD and SB. For instance, $SMS^\alpha$ spends $13.3\%$, $17.1\%$, $44.9\%$ and $49.8\%$ less per item won against SCPD than $MS^\lambda$, EPE, SCPD and SB respectively. Moreover, against $MS^\lambda$ and EPE, only EPE spends slightly less than $SMS^\alpha$ per item won.

Regarding strategy profiles where all bidders play the same strategy, the one corresponding to $SMS^\alpha$ has an average price paid per item won $1.70$, $2.35$ and $2.98$ times lower than $MS^\lambda$, SCPD and SB respectively. Moreover, by looking at Figure \ref{ratio itm}, we can see that all items are allocated when all bidders play $SMS^\alpha$. Being capable of splitting up all items at a relatively low price explains why the expected utility of the strategy profile where all bidders play $SMS^\alpha$ is substantially higher than the ones where all bidders play a same other strategy. Only obtaining items at a low price is not sufficient. For instance, when all bidders play EPE, the average price paid per item won is $1.6$ times lower than when all bidders play $SMS^\alpha$. However, only $72\%$ of all items are allocated. Hence, this strategy profile achieves a lower expected utility than if all bidders had played $SMS^\alpha$.

Moreover, the fact that the average price per item won when all bidders play EPE is relatively close to $\varepsilon$ raises an important strategic issue. Indeed, to obtain such low price, EPE bidders drastically reduce their eligibility during the first round without considering the fact that they might end up in a vulnerable position. Hence, an EPE bidder can easily be deceived. This explains why a bidder doubles its expected utility if it decides to play $SMS^\alpha$ instead of EPE when all its opponents are playing EPE in Figure \ref{Expected utility EPE}. After the first round, $SMS^\alpha$ easily takes advantage of the weak position of its opponents. By gradually decreasing its eligibility, a $SMS^\alpha$ bidder tackles efficiently the own price effect and avoids putting itself in vulnerable positions.

\subsection{Exposure} \label{Section exposure}

To analyse exposure, we plot in Figure \ref{fig:Estimated Exposure} the expected exposure of each strategy $A$ against every strategy $B$ displayed on the x-axis. Similarly, we plot in Figure \ref{fig:Exposure frequency} the exposure frequency of each strategy $A$ against every strategy $B$ displayed on the x-axis. For instance, if $A=SMS^\alpha$ and $B=\text{SB}$, the expected exposure and exposure frequency of $SMS^\alpha$ against $SB$ are respectively $0.07$ and $4.4\%$. They both correspond respectively to the orange bar above index SB on the x-axis in Figure \ref{fig:Estimated Exposure} and Figure \ref{fig:Exposure frequency}.

Firstly, in the situation where all bidders decide to play the same strategy, $SMS^\alpha$ has the remarkable property of never leading to exposure. This is not the case for the four other strategies. Secondly, $SMS^\alpha$ is the only strategy which never suffers from exposure against $MS^\lambda$ and EPE. Thirdly, even against SCDP and SB, $SMS^\alpha$ is rarely exposed. It has the lowest expected exposure and exposure frequency. For instance, $SMS^\alpha$ induces $9.3$, $4.5$, $34$ and $90$ times less expected exposure against SCPD than $MS^\lambda$, EPE, SCPD and SB respectively. Moreover, regarding exposure frequency, by playing $SMS^\alpha$ a bidder has $6.6$, $4$, $27.6$ and $58.1$ times less chance of ending up exposed against SCPD than $MS^\lambda$, EPE, SCPD and SB respectively. 

Hence, not only does $SMS^\alpha$ achieve higher expected utility than state-of-the-art algorithms but it also takes less risks. 

In this analysis, we exclusively focus on empirical games where bidders have the choice between only two distinct bidding strategies. In Appendix \ref{Appendix A}, we show that $SMS^\alpha$ still remains highly efficient even when bidders have access to a wider range of bidding strategies. Moreover, we also ran additional experiments on instances with different numbers of players and items in Appendix \ref{Appendix B}. The same conclusions can be drawn as for instances of size ($n=4$, $m=11$, $\varepsilon=1$), such as $SMS^\alpha$ achieves higher expected utility than state-of-the-art algorithms, notably by better tackling the own price effect and the exposure problem in eligibility and budget-constrained environments.

\subsection{Influence and selection of $N_{act}$ and $\alpha$}
\label{choice hyperparameters}

\begin{figure*}[!t]
\centering
\subfloat[Expected utility when all bidders play $SMS^\alpha$ versus $N_{act}$]{\includegraphics[width=2.5in]{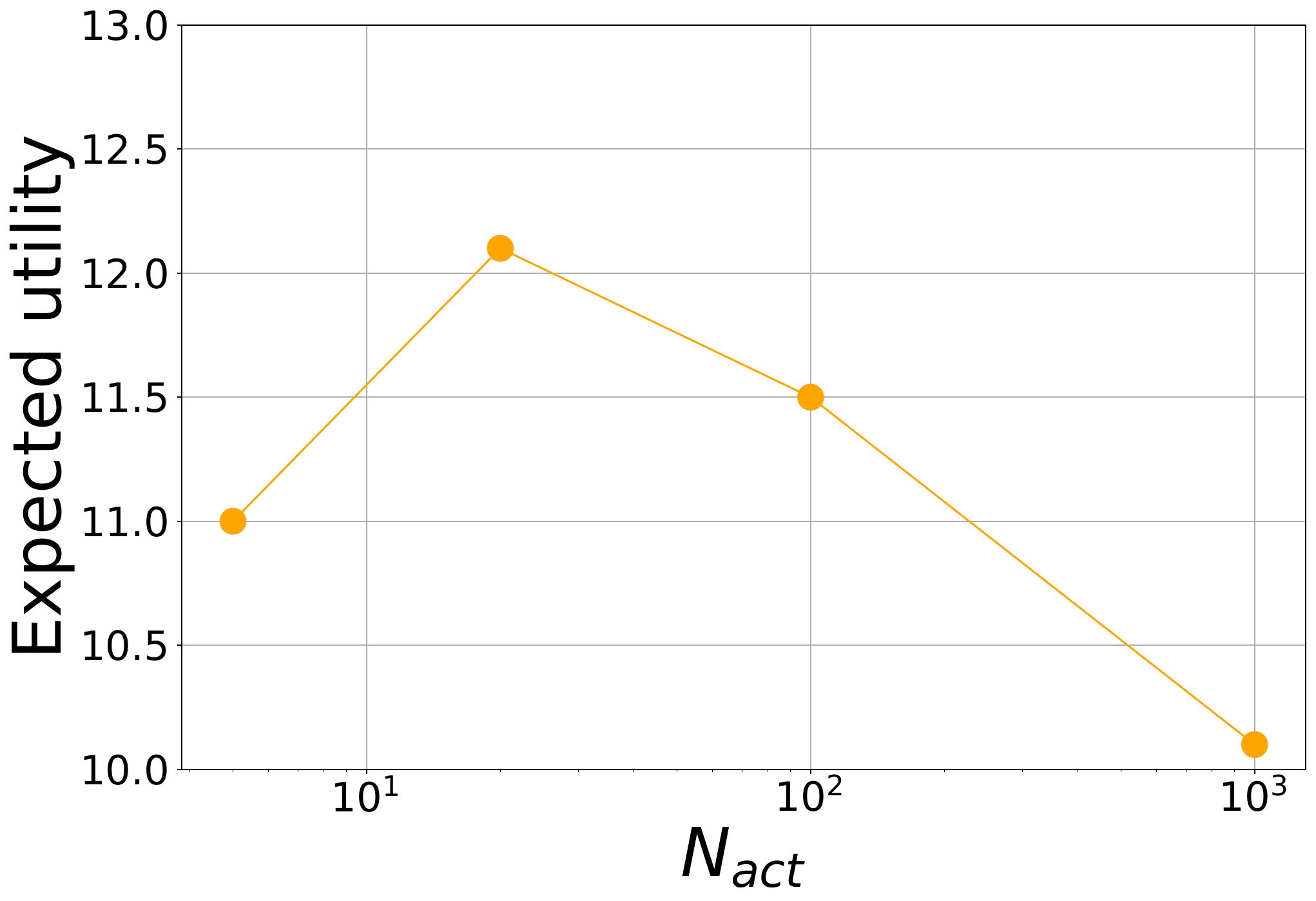}%
\label{N_act expected}}
\hfil
\subfloat[Exposure frequency ($\%$) of $SMS^\alpha$ against SB versus $N_{act}$]{\includegraphics[width=2.5in]{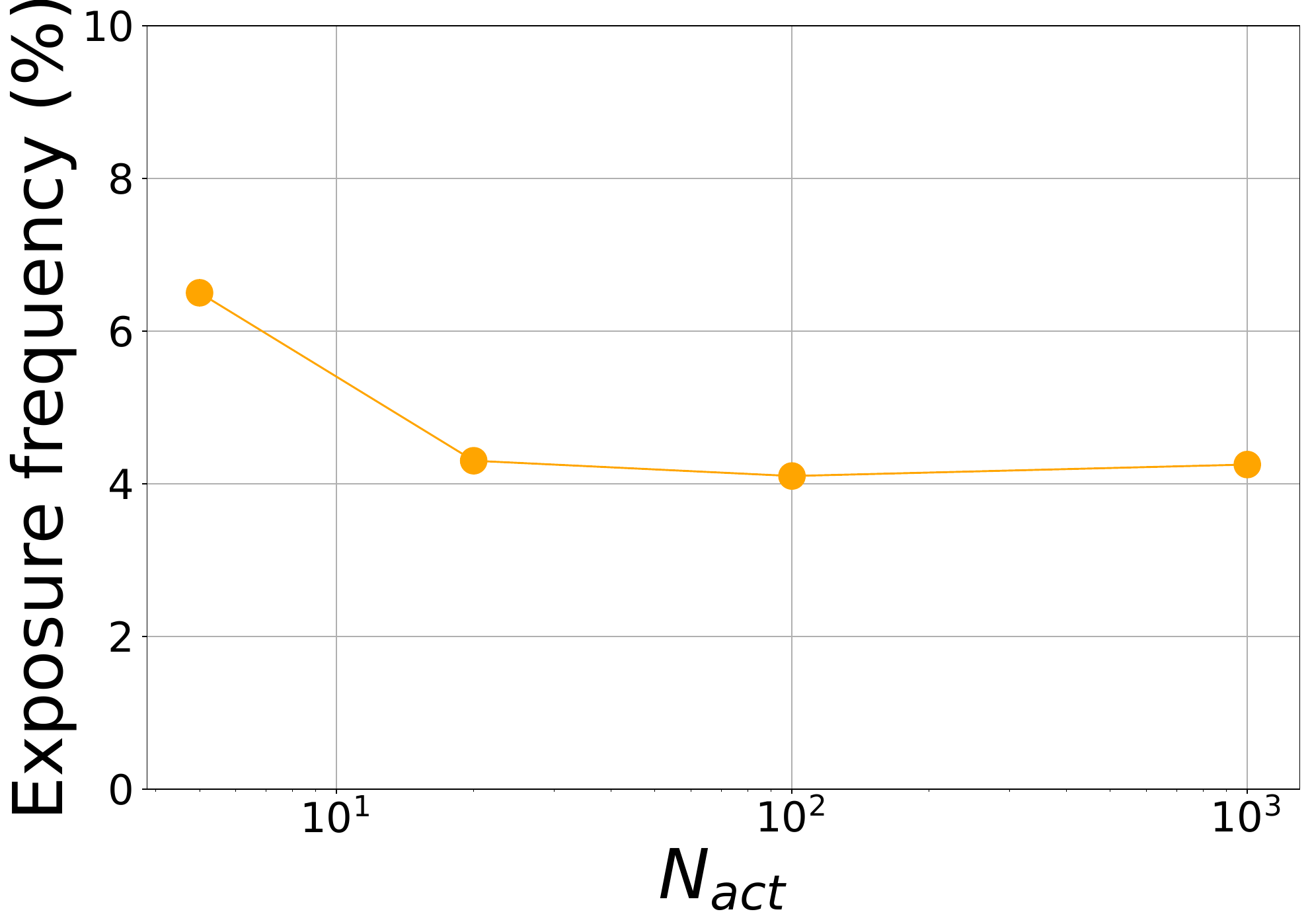}%
\label{N_act exposure}}

\subfloat[Exposure frequency ($\%$) of $SMS^\alpha$ against SB versus $\alpha$]{\includegraphics[width=2.5in]{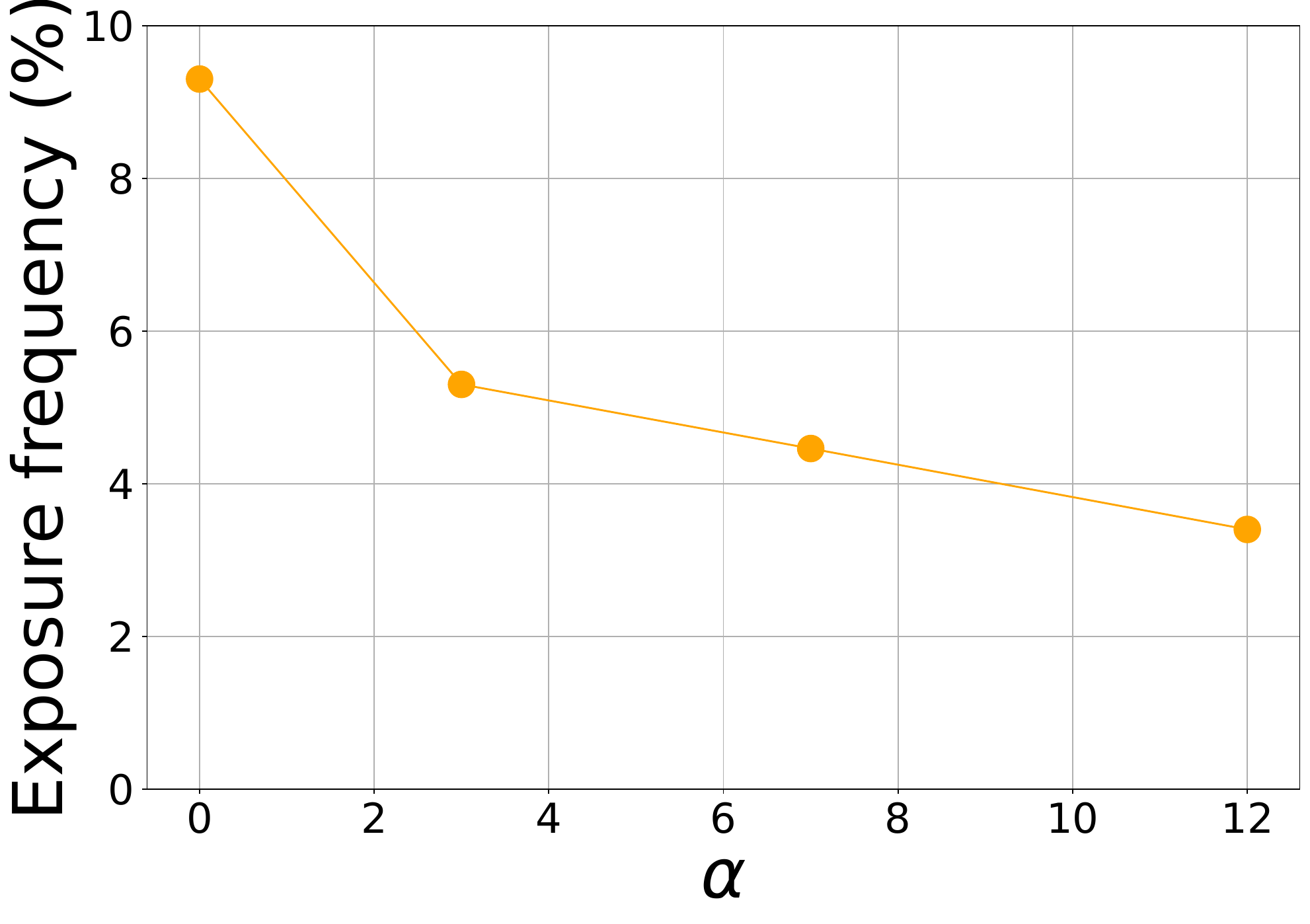}%
\label{alpha exposure}}
\hfil
\subfloat[Increase in expected utility ($\%$) when deviating from SB to $SMS^\alpha$ versus $\alpha$]{\includegraphics[width=2.5in]{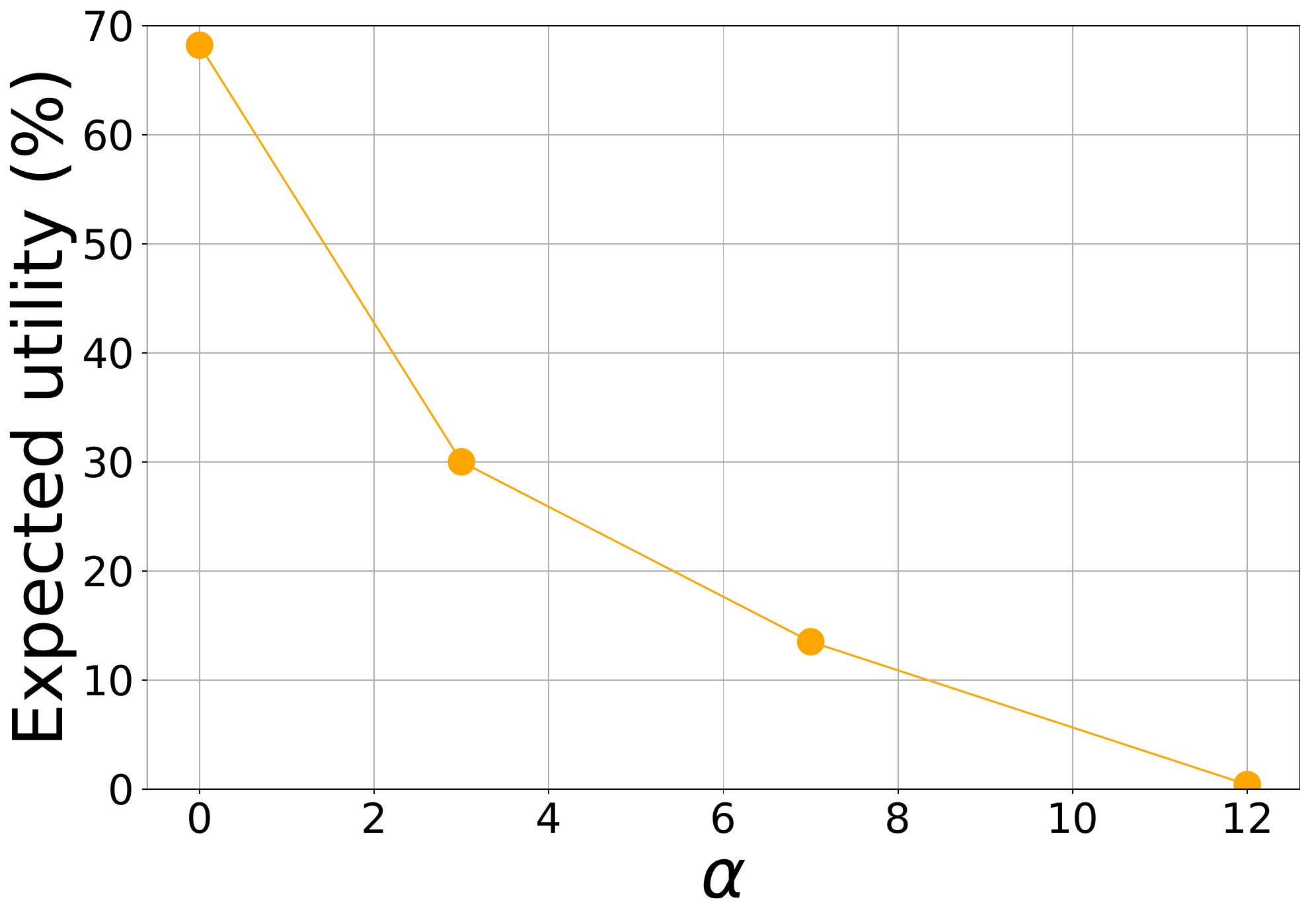}%
\label{alpha expected}}

\caption{Impact of $N_{act}$ and $\alpha$ on $SMS^\alpha$} 
\end{figure*}

The values of hyperparameters $N_{act}$ and $\alpha$ are obtained through grid search. We propose hereafter an analysis of the impact of both hyperparameters on the performance of $SMS^\alpha$ and explain our choice of setting $N_{act}$ to $20$ and $\alpha$ to $7$. 

First, we compare $SMS^\alpha$ when $\alpha=7$ for the following values of $N_{act}$: $5$, $20$, $100$ and $1000$. We plot the expected utility when all bidders play $SMS^\alpha$ versus $N_{act}$ in Figure~\ref{N_act expected}. We can see that the expected utility drastically increases between $5$ and $20$ and, then, slowly decreases between $20$ and $1000$. The fact that the expected utility seems to be a unimodal function presenting a maximum for $N_{act}=20$ illustrates perfectly the tradeoff to be made between in-depth inspection and a wide range of possible moves per player. Moreover, in Figure \ref{N_act exposure}, we represent the exposure frequency of $SMS^\alpha$ against SB versus $N_{act}$. We can see that between $5$ and $20$ the exposure frequency drops by $50\%$ whereas it remains relatively constant afterwards. This highlights the fact that having an overly restricted range of moves reduces the possibilities to avoid exposure. Hence, in order to maximise coordination between bidders and minimise exposure, we set $N_{act}$ to $20$.

Secondly, we compare $SMS^\alpha$ with $N_{act}=20$ for the following values of $\alpha$: $0$, $3$, $7$ and $12$. We represent the exposure frequency of $SMS^\alpha$ against SB versus $\alpha$ in Figure~\ref{alpha exposure}. We clearly see that the exposure frequency decreases when $\alpha$ grows. Thus, by increasing $\alpha$, $SMS^\alpha$ tackles better the exposure problem. In Figure \ref{alpha expected}, we plot the increase in expected utility when deviating from SB to $SMS^\alpha$ when two opponents play $SB$ and the last opponent plays $SMS^\alpha$. We observe that the profitability of the deviation considerably decreases when $\alpha$ grows and nearly reaches $0$ when $\alpha=12$. Hence, the reduction of exposure when $\alpha$ increases is accompanied by a decrease in profitability of deviating to $SMS^\alpha$. The hyperparameter $\alpha$ thus allows the bidder to arbitrate between expected utility and risk-aversion. In order to maintain all deviations towards $SMS^\alpha$ profitable in each empirical game illustrated in Figure \ref{Expected utility} and minimise the risk of exposure, we set $\alpha$ to $7$.

It is also important to note that increasing $\alpha$ reduces the own price effect and enables a better coordination between bidders. This is a natural effect of risk-aversion where a bidder tends to avoid a rise in price. For instance, the average price per item won all bidders play $SMS^\alpha$ is respectively $3.1$ and $2$ for $\alpha$ equal to $0$ and $12$. This explains why the expected utility when all bidders play $SMS^{12}$ is $1.36$ times greater than the expected utility when all bidders play $SMS^{0}$.

\section{Conclusions and Future Work}

This paper introduces the first efficient bidding strategy that tackles simultaneously the \textit{exposure problem}, the \textit{own price effect}, \textit{budget constraints} and the \textit{eligibility management problem} in a simplified version of SAA (SAA-c). It is a SM-MCTS whose expansion and rollout phase relies on a new method for the prediction of closing prices. By introducing hyperparameter $\alpha$, we give the freedom to bidders to arbitrate between expected profit and risk-aversion. Our solution $SMS^\alpha$ largely outperforms state-of-the-art algorithms on instances of realistic size in generic settings. 

For future works, it would be interesting to see if $SMS^\alpha$ can be improved by modifying some of its search iteration phases. For instance, another selection index could be used such as \textit{EXP3} or \textit{Regret Matching} \cite{bovsansky2016algorithms}. Moreover, following the successful model of AlphaZero \cite{silver2018general}, exploring approaches which combine MCTS with neural networks might also be a promising direction.

An important future work is the design of an efficient bidding strategy for SAA with incomplete information. To address this challenge, one could investigate determinization approaches. These solutions are based on complete information schemes like $SMS^\alpha$, which thus appears as an important building block for future advancements.

Finally, although $SMS^\alpha$ was initially designed for bidders, it could also be used by mechanism designers to verify the impact of new rules on the efficiency and revenue of SAA.

\bibliographystyle{plain}
\bibliography{sample-base}

\cleardoublepage
\appendix

\subsection{Extensive experiments with more than two distinct strategies} \label{Appendix A}

In this section, we show that $SMS^\alpha$ still computes an efficient bidding strategy when there are more than two distinct strategies present among bidders. For instance, we provide hereafter results for the strategy profiles (A, $MS^\lambda$, EPE, SCPD) with A$\in$\{$SMS^\alpha$, $MS^\lambda$, EPE, SCPD, SB\}. Each experimental result has been run on the same $1000$ SAA-c instances of size ($n=4$, $m=11$, $\varepsilon=1$) than in Section \ref{extensive experiments}. We fix the strategies played by the opponents (bidder 2 plays $MS^\lambda$, bidder 3 plays EPE, bidder 4 plays SCPD) and analyse the effects of deviating from one strategy to another for the first player. Our analysis is divided into three parts: expected utility, own price effect and exposure.
\begin{enumerate}
    \item \textit{Expected utility:} We plot in Figure \ref{fig:Expected util Trans} the expected utility obtained by the first player depending on the strategy it decides to play. We observe that, by playing $SMS^\alpha$, the first player obtains an expected utility which is respectively $1.04$, $1.09$, $1.51$ and $1.63$ times higher than if it had played $MS^\lambda$, SCPD, EPE and SB. Thus, if bidder 2 plays $MS^\lambda$, bidder 3 plays EPE and bidder 4 plays SCPD, then the first player should play $SMS^\alpha$ to maximise its expected utility. 
    \begin{figure}[h]
    \centering
    \includegraphics[width=0.9 \columnwidth]{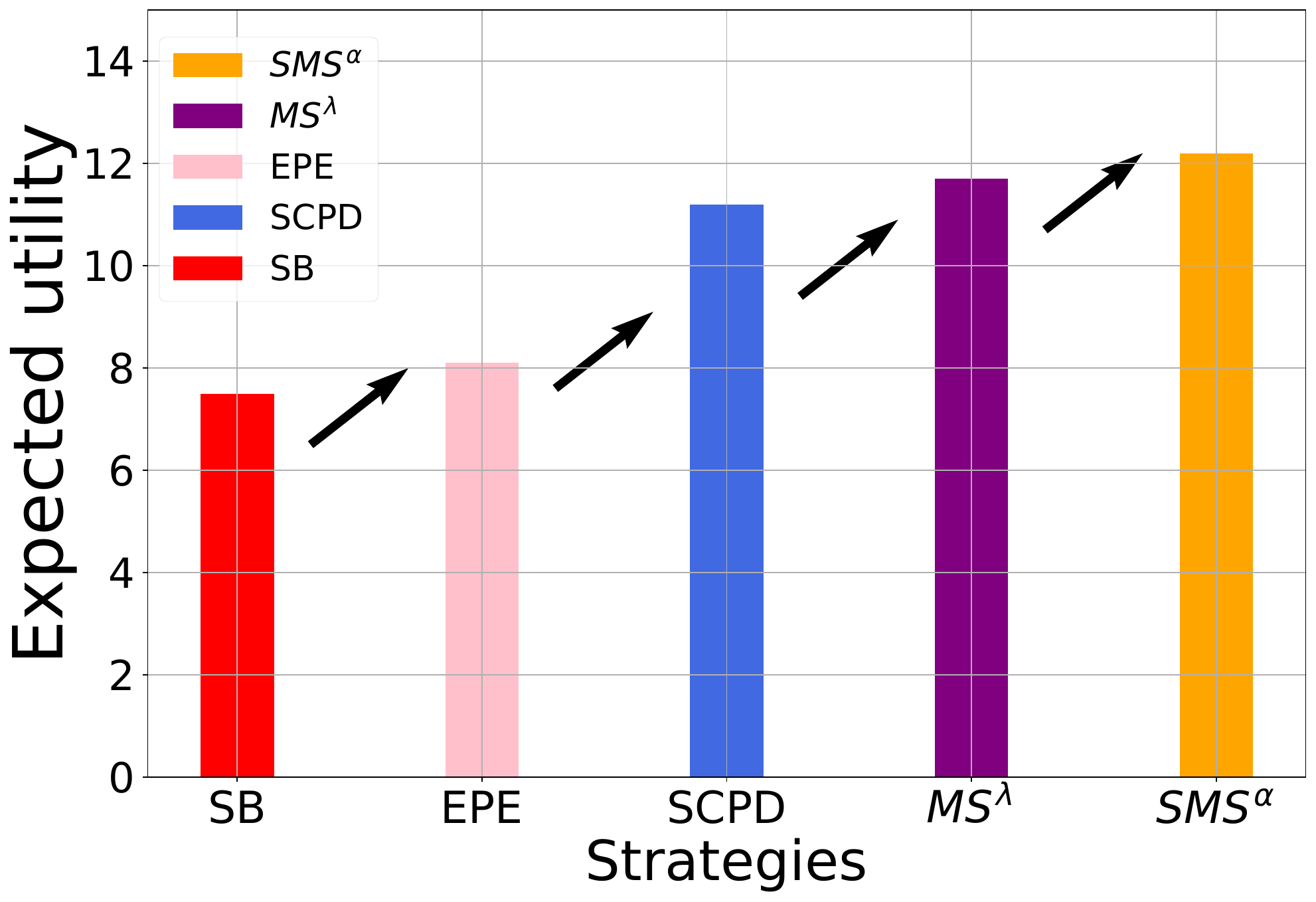}
    \caption{Expected utility obtained by the first player in strategy profile (A, $MS^\lambda$, EPE, SCPD) depending on the strategy A$\in$\{ $SMS^\alpha$, $MS^\lambda$, EPE, SCPD, SB\} played.} 
    \label{fig:Expected util Trans}
    \end{figure}
    
    \item \textit{Own price effect:} We plot in Figure \ref{fig:Own_price_effect Trans} the average price payed per item won and the ratio of items won of the first player depending on which strategy it plays. We see that, by playing $SMS^\alpha$, the first player spends $2\%$, $30\%$, $33\%$ and $53\%$ less per item won than respectively $MS^\lambda$, EPE, SCPD and SB. Moreover, regarding the ratio of items won, by playing $SMS^\alpha$, the first player obtains $26\%$ of all items. The highest ratio of items won is obtained by playing SB. In this specific case, the first player obtains $42\%$ of all items. This highlights the fact that reducing demand to obtain items at a lower price is far more effective to increase one's expected utility than bidding aggressively to acquire a maximum number of items. Hence, by playing $SMS^\alpha$, the first player tackles better the own price effect than the other strategies when bidder 2 plays $MS^\lambda$, bidder 3 plays EPE and bidder 4 plays SCPD,
    \begin{figure*}[h]
    \centering
    \subfloat[Average price per item won]{\includegraphics[width=3in]{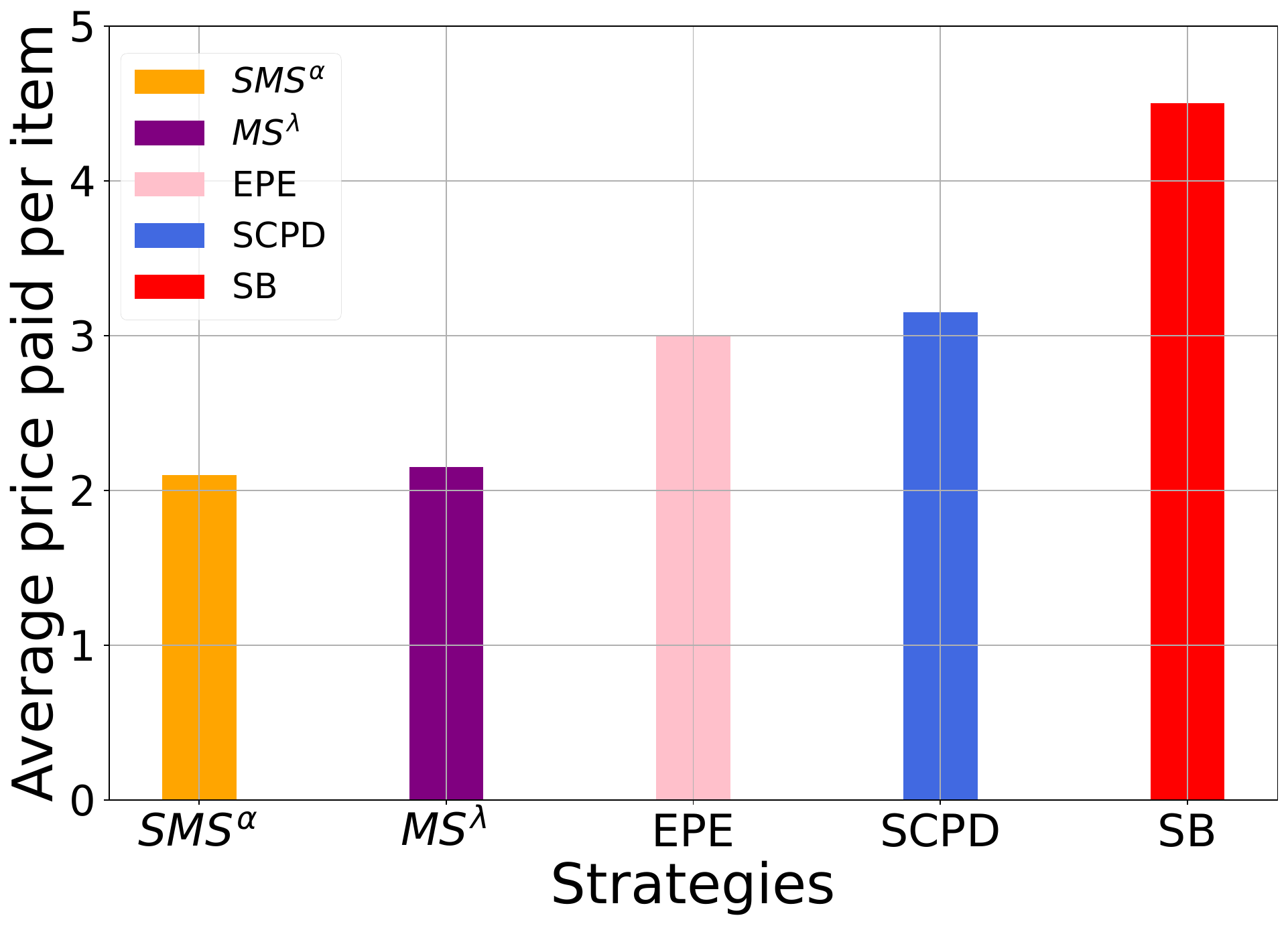}}
    \hfil
    \subfloat[Ratio items won]{\includegraphics[width=3.1in]{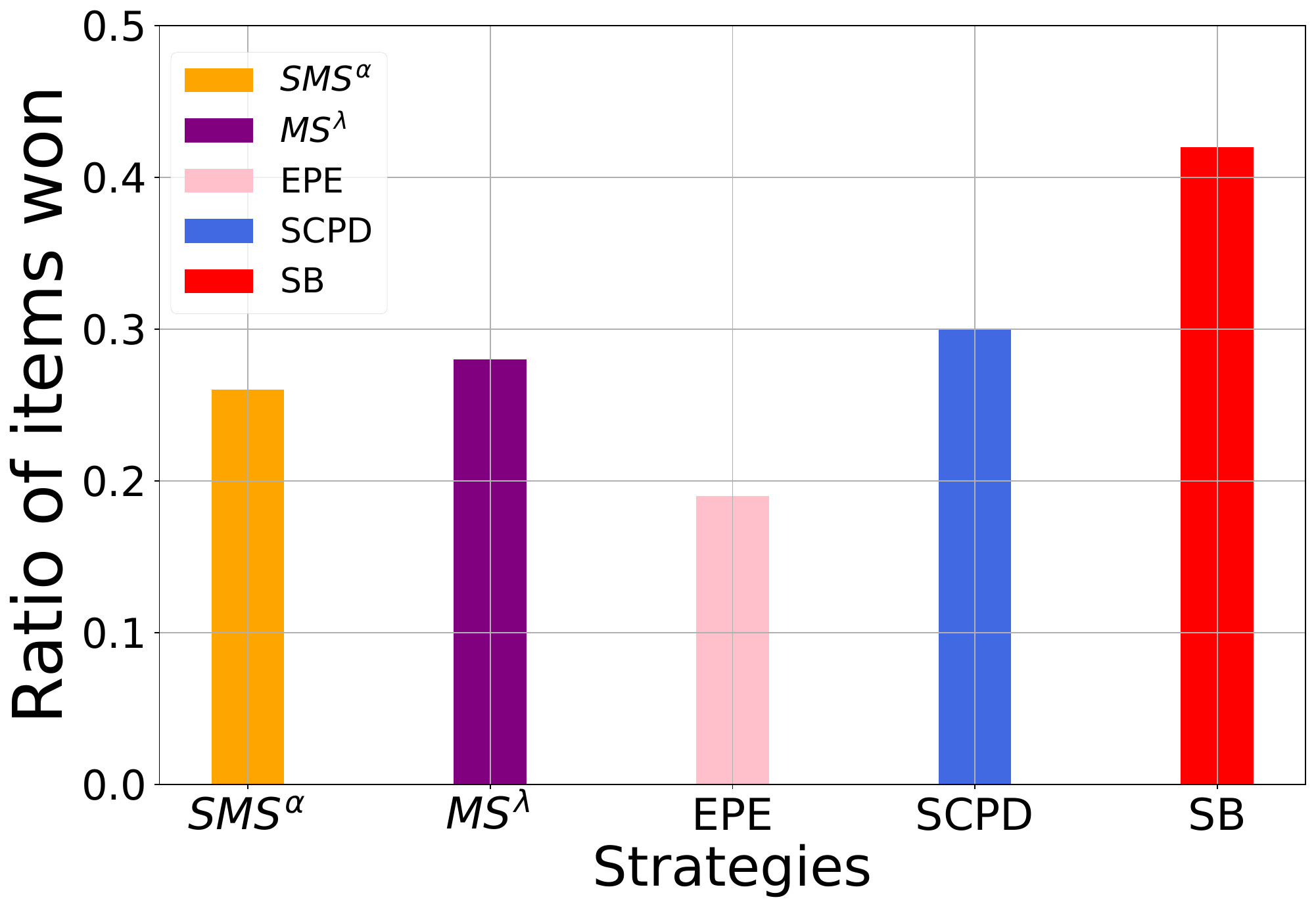}}
    \caption{Own price effect analysis for the first player in strategy profile (A, $MS^\lambda$, EPE, SCPD) depending on the strategy A$\in$\{ $SMS^\alpha$, $MS^\lambda$, EPE, SCPD, SB\} played.}
    \label{fig:Own_price_effect Trans}
    \end{figure*}
    
    \item \textit{Exposure:} In Figure \ref{fig:Exposure Trans}, we plot the expected exposure and the exposure frequency incurred by the first player depending on the strategy played. By playing $SMS^\alpha$, the first player incurs $10$, $5$, $24$ and $133$ times less expected exposure than by playing respectively $MS^\lambda$, EPE, SCPD and SB. Moreover, by playing $SMS^\alpha$, the first player has $4.3$, $3.3$, $11.7$ and $49.7$ less chance of ending up exposed than by playing respectively $MS^\lambda$, EPE, SCPD and SB. Thus, by playing $SMS^\alpha$, the first player considerably reduces the risk of exposure compared to other strategies when bidder 2 plays $MS^\lambda$, bidder 3 plays EPE and bidder 4 plays SCPD.
    \begin{figure*}[t]
    \centering
    \subfloat[Expected exposure]{\includegraphics[width=3.1in]{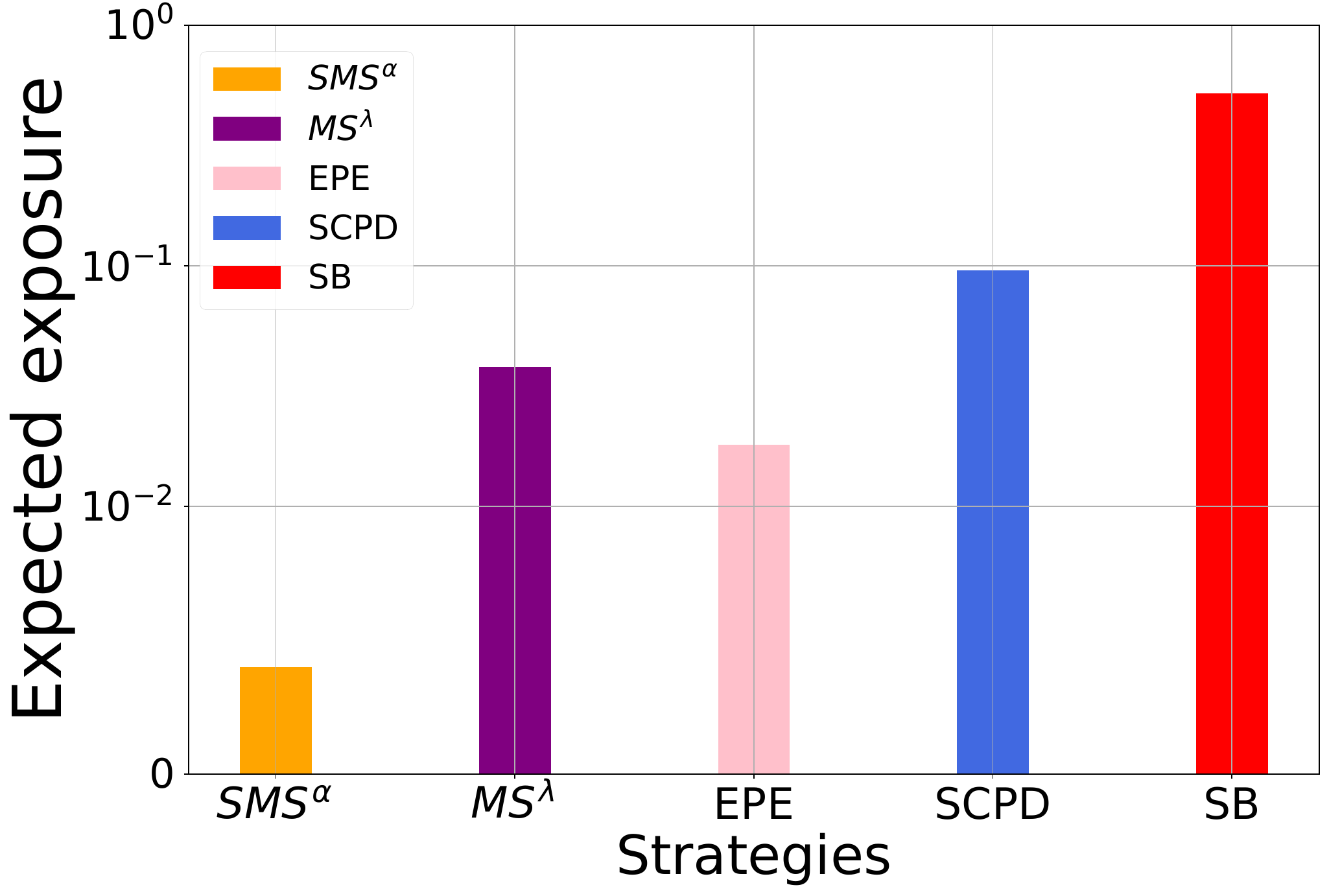}}
    \hfil
    \subfloat[Exposure frequency]{\includegraphics[width=3in]{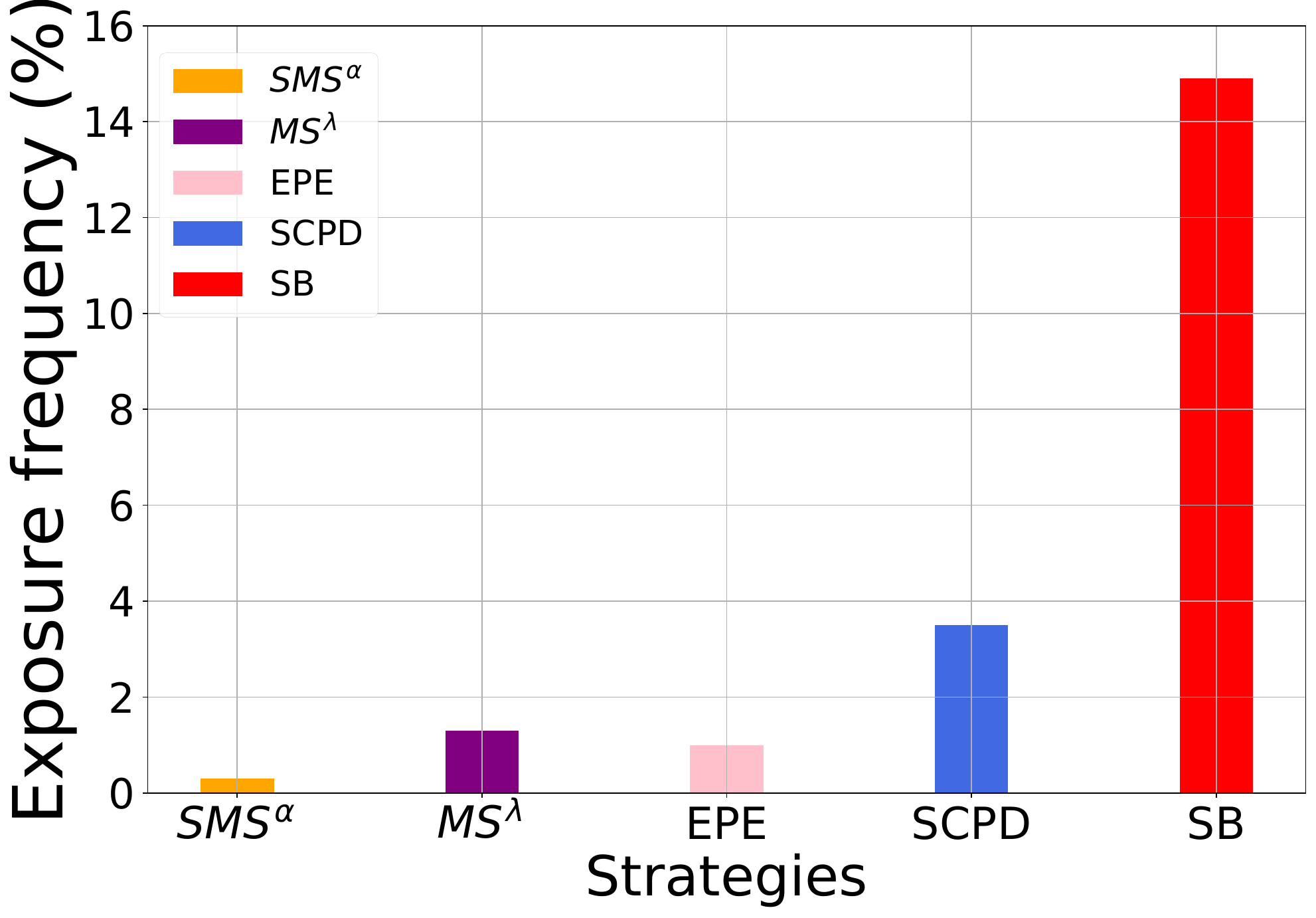}}
    \caption{Exposure analysis for the first player in strategy profile (A, $MS^\lambda$, EPE, SCPD) depending on the strategy A$\in$\{ $SMS^\alpha$, $MS^\lambda$, EPE, SCPD, SB\} played.}
    \label{fig:Exposure Trans}
    \end{figure*}
\end{enumerate}

Through this example, we clearly see that $SMS^\alpha$ remains an efficient bidding strategy when there are more than two distinct strategies. Indeed, $SMS^\alpha$ achieves higher expected utility than state-of-the-art strategies by better tackling the own price effect and the exposure problem in eligibility and budget constrained environments.

\subsection{Extensive experiments on instances of size ($n=2$, $m=7$, $\varepsilon=1$) and ($n=3$, $m=9$, $\varepsilon=1$)} \label{Appendix B}

We provide hereafter results obtained on two different sizes of instances: ($n=2$, $m=7$, $\varepsilon=1$) and ($n=3$, $m=9$, $\varepsilon=1$). Results are computed through 1000 SAA-c instances. Value functions and budgets are generated the same way as for instances of size ($n=4$, $m=11$, $\varepsilon=1$), i.e. with $b_{min}=10$, $b_{max}=40$ and $V=5$. Our analysis is divided into three parts: expected utility, own price effect and exposure.\\


For instances of size ($n=2$, $m=7$, $\varepsilon=1$), the hyperparameter $\alpha$ of $SMS^\alpha$ takes the value $12$ and the risk-aversion hyperparameters $\lambda^r$ and $\lambda^o$ of $MS^\lambda$ both take the value $0.025$. The maximum number of expanded actions per information set $N_{act}$ of $SMS^\alpha$ is set to $20$. Each algorithm is given respectively $50$ seconds of thinking time. 
\begin{enumerate}
    \item \textit{Expected Utility:} We represent in Figure \ref{Expected utility 2 pl} the four normal form games in expected utility where each player has the choice between either playing $SMS^\alpha$ or another strategy A. In each empirical game, playing $SMS^\alpha$ is always profitable. Thus, the strategy profile ($SMS^\alpha$, $SMS^\alpha$) is a Nash equilibrium of the normal form game in expected utility with strategy set \{$SMS^\alpha$, $MS^\lambda$, EPE, SCPD, SB\}. 
    \begin{figure*}[t]
    \centering
    \subfloat[$SMS^\alpha$ vs $MS^\lambda$]{\includegraphics[width=2.1in]{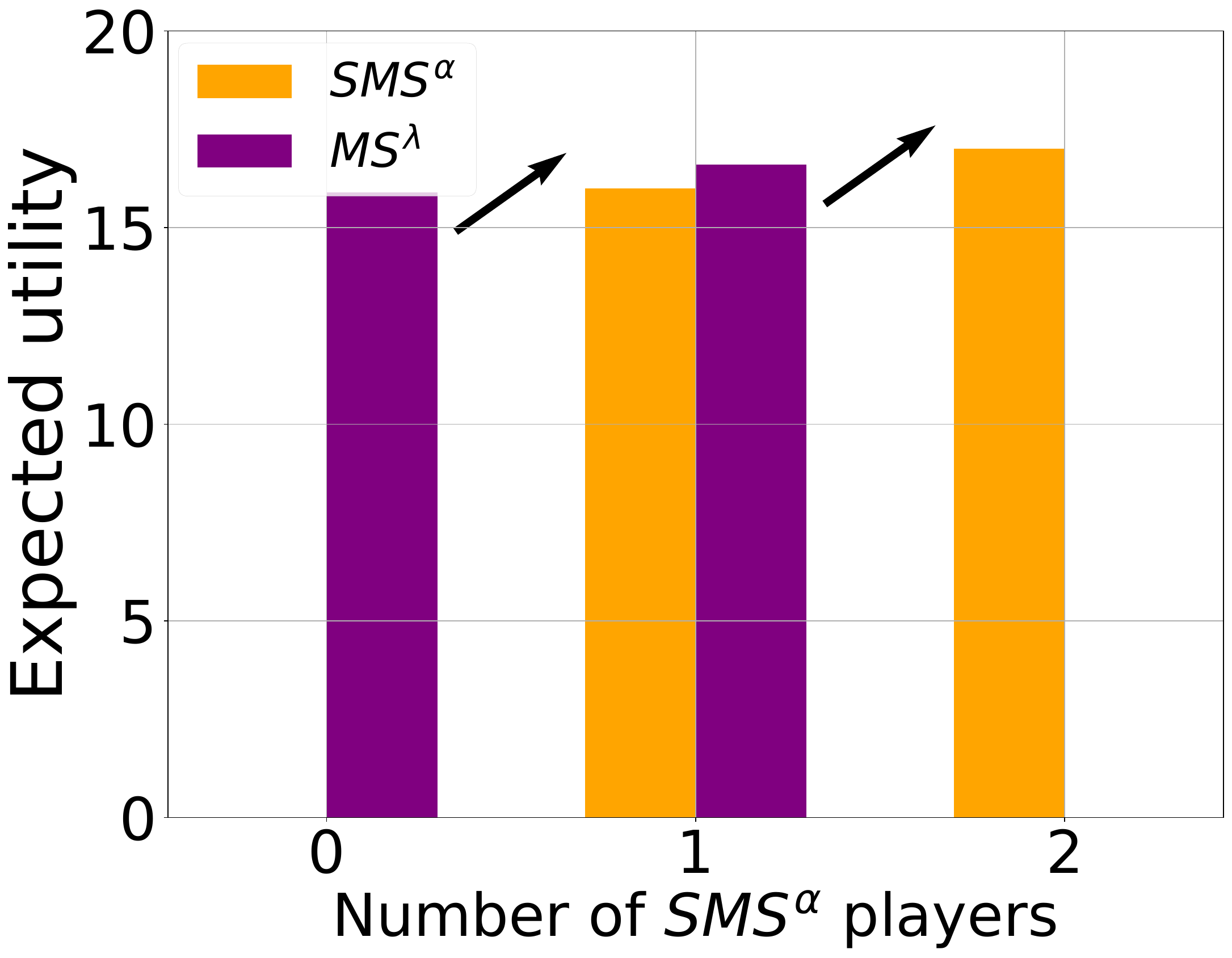}}
    \hfil
    \subfloat[$SMS^\alpha$ vs EPE]{\includegraphics[width=2.1in]{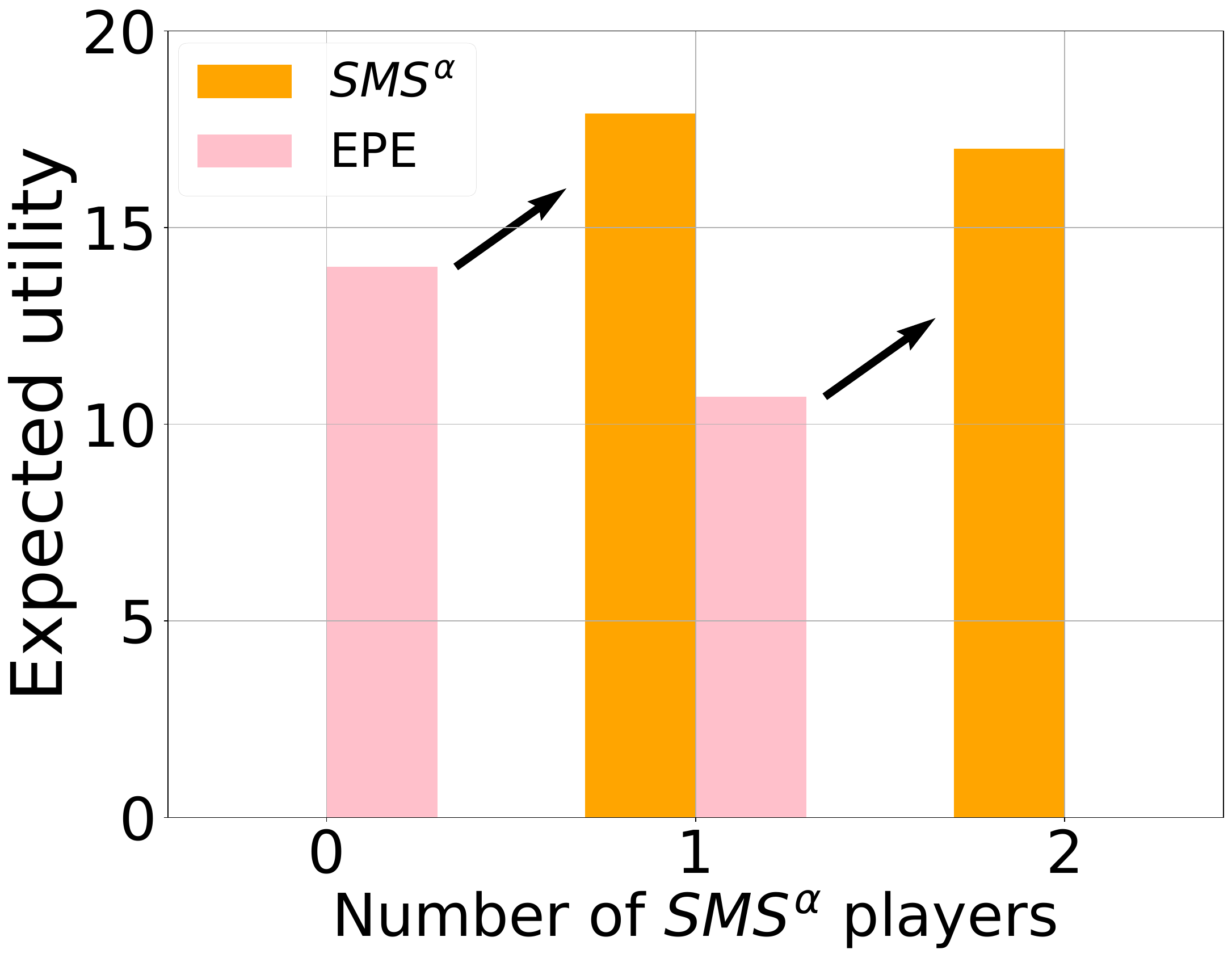}}

    \subfloat[$SMS^\alpha$ vs SCPD]{\includegraphics[width=2.1in]{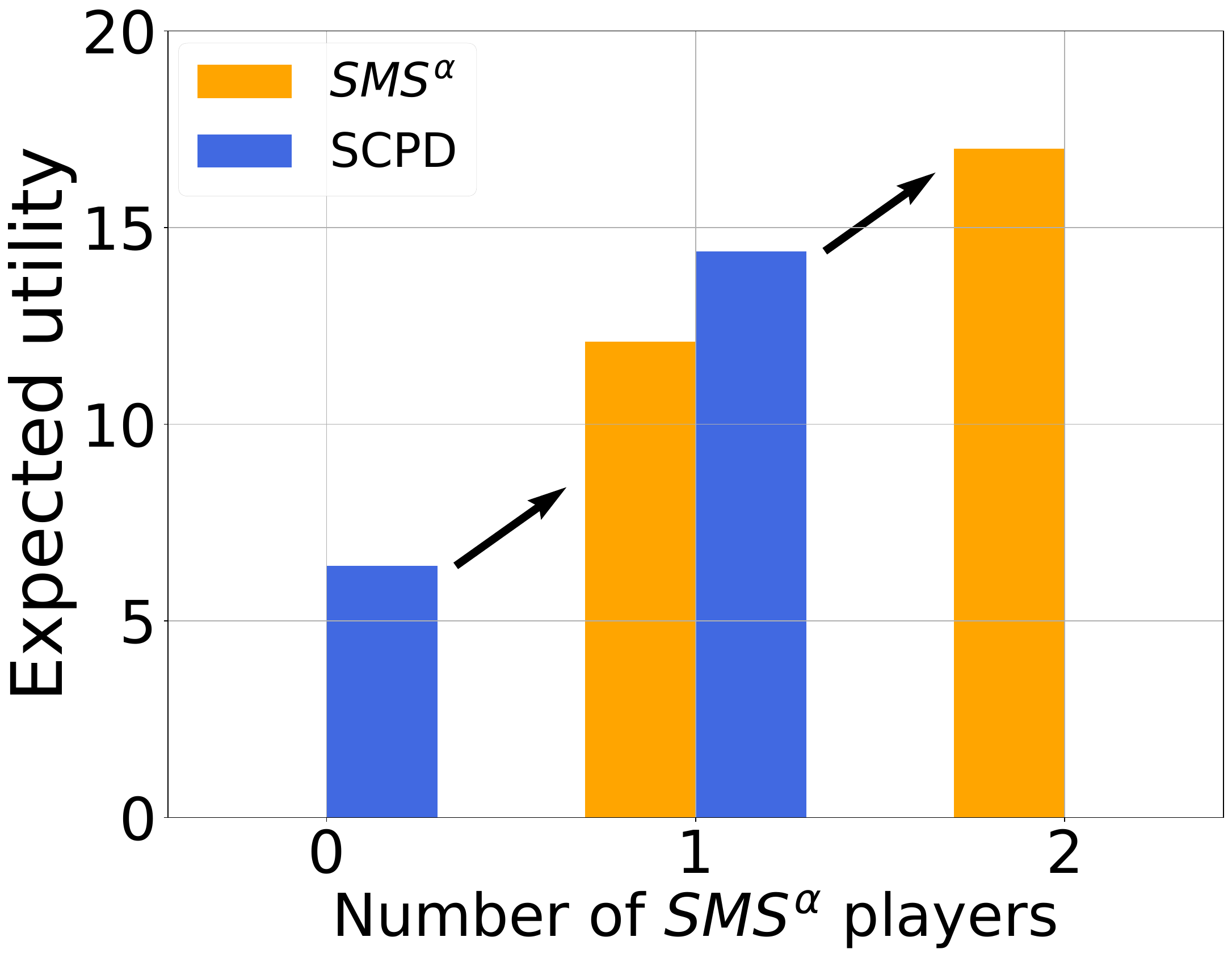}}
    \hfil
    \subfloat[$SMS^\alpha$ vs SB]{\includegraphics[width=2.1in]{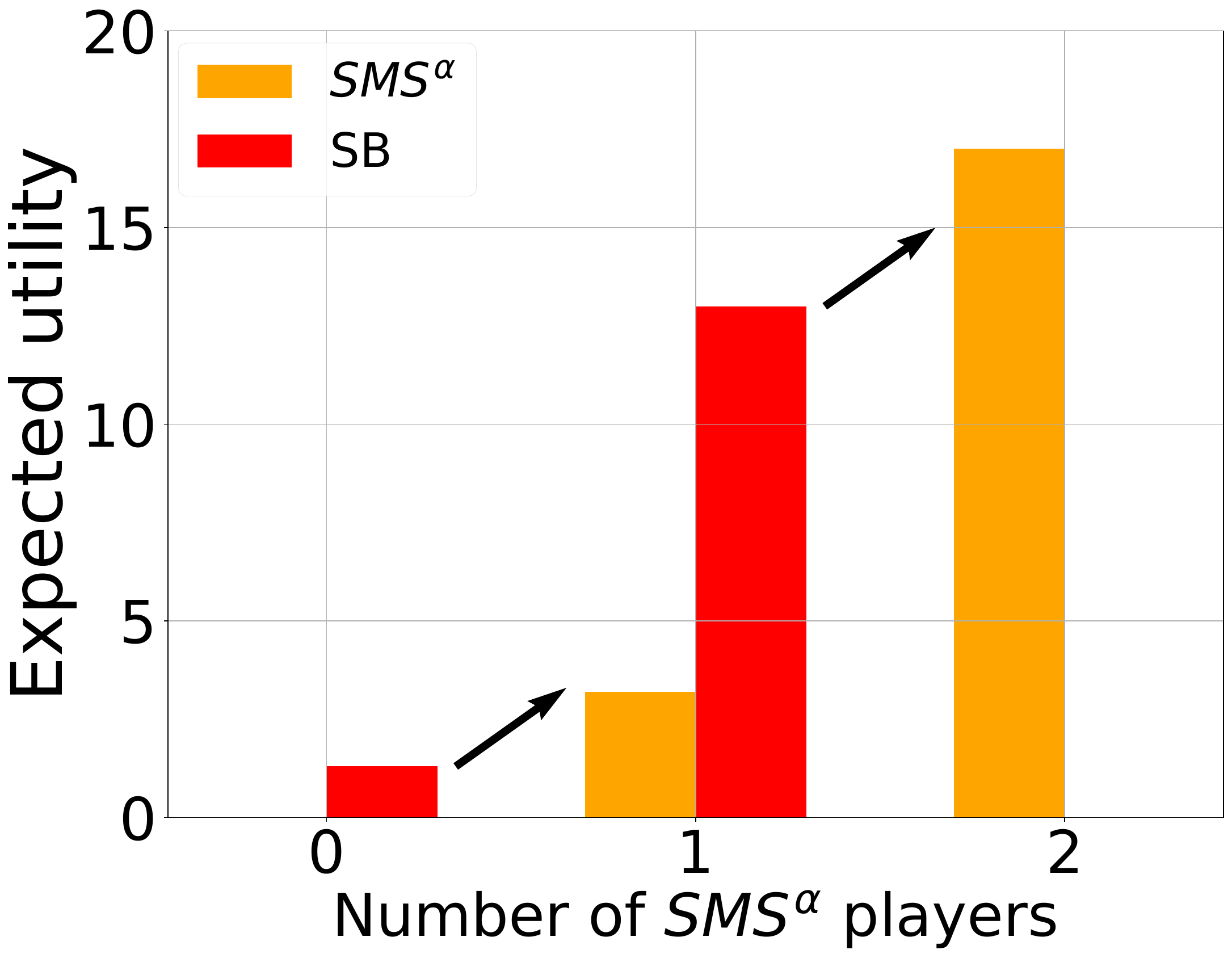}}
    \caption{Normal-form SAA-c game in expected utility with five 
    strategies ($n=2$, $m=7$, $\varepsilon=1$)}
    \label{Expected utility 2 pl}
    \end{figure*}
    Moreover, regarding strategy profiles where all bidders play the same strategy, $SMS^\alpha$ has a significantly higher expected utility than the other bidding strategies. For instance, it obtains an expected utility $1.07$, $1.21$, $2.66$ and $13.28$ times higher than respectively $MS^\lambda$, EPE, SCPD and SB. 
    
    \item \textit{Own price effect}: We plot in Figure \ref{avg price 2} the average price paid per item won and in Figure \ref{ratio itm 2} the ratio of items won by each strategy $A$ against every strategy $B$ displayed on the x-axis. In Figure \ref{avg price 2}, we can see that $SMS^\alpha$ acquires items at a lower price in average than the other strategies against $SMS^\alpha$, SCPD and SB. For instance, $SMS^\alpha$ spends $2.7\%$, $28.1\%$, $52.6\%$ and $56.1\%$ less per item won against SCPD than $MS^\lambda$, EPE, SCPD and SB respectively. Moreover, against $MS^\lambda$ and EPE, only EPE spends slightly less than $SMS^\alpha$ per item won. 
    
    Moreover, regarding strategy profiles where all bidders play the same strategy, the one corresponding to $SMS^\alpha$ has an average price payed per item won $1.19$, $3.33$ and $4.65$ times lower than $MS^\lambda$, SCPD and SB respectively. All items are allocated when all bidders play $SMS^\alpha$ whereas only $78\%$ are allocated when all bidders play EPE.
    
    \begin{figure*}[t]
    \centering
    \subfloat[Average price per item won]{\includegraphics[width=3in]{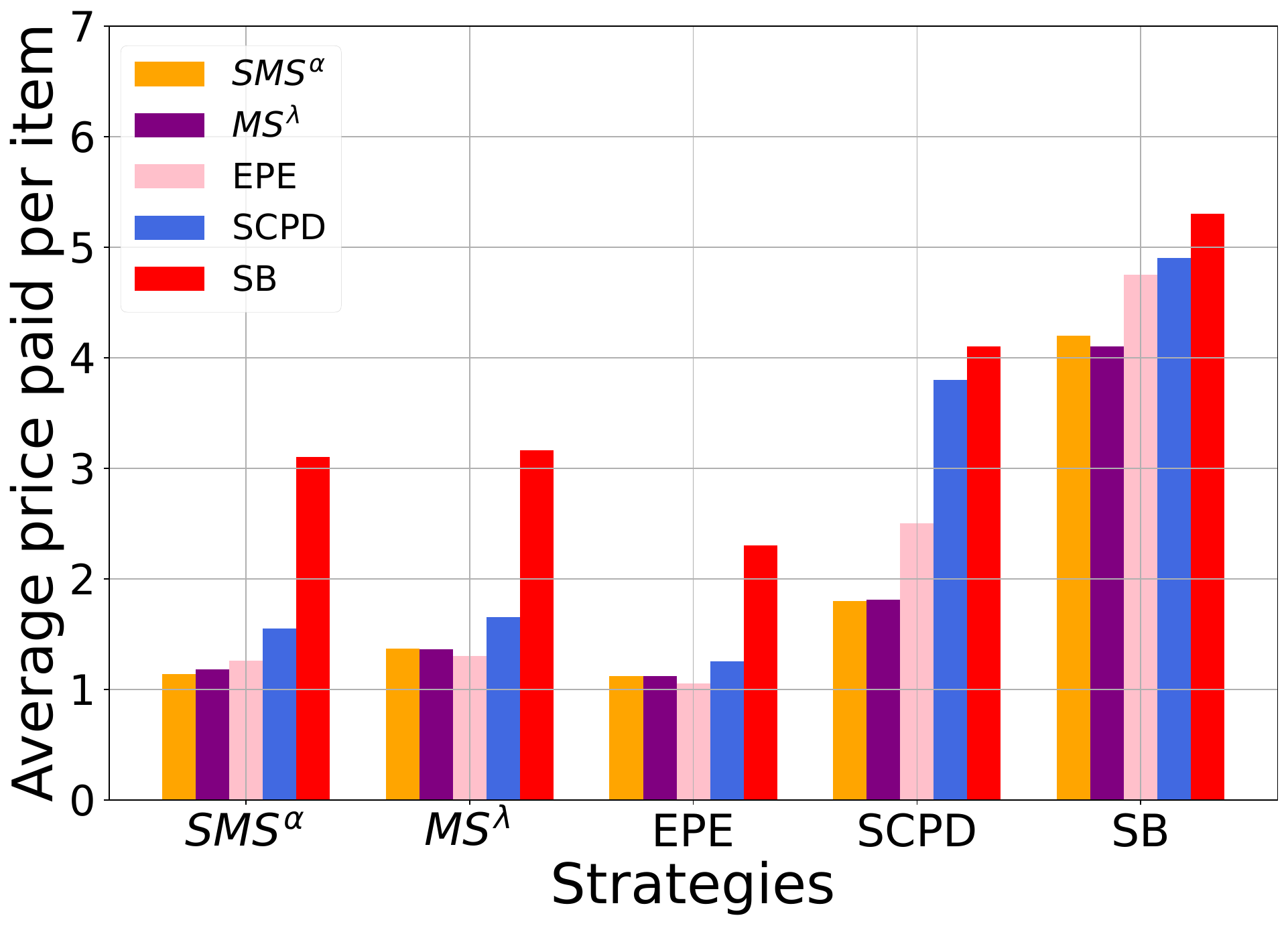}%
    \label{avg price 2}}
    \hfil
    \subfloat[Ratio of items won]{\includegraphics[width=3.1in]{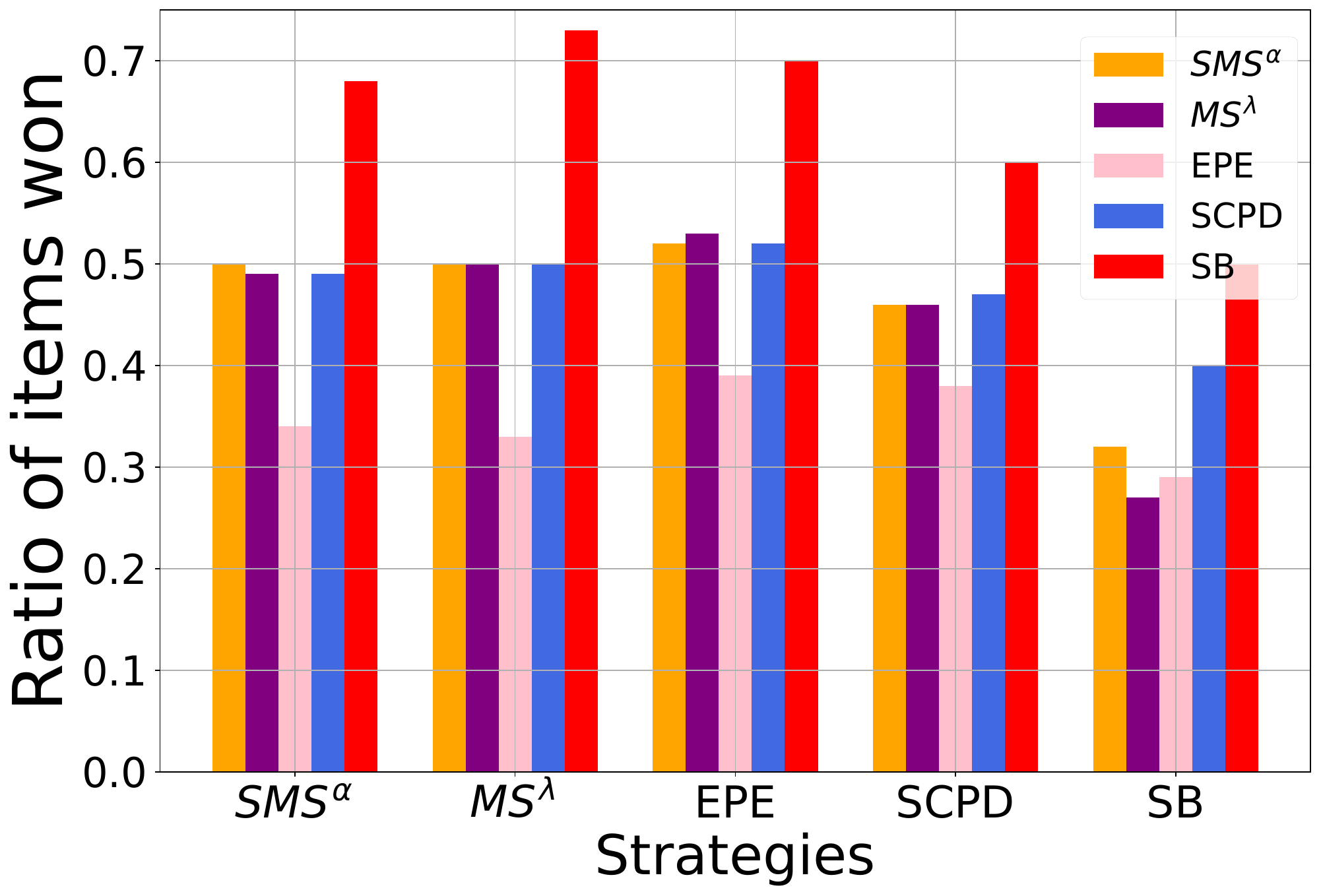}%
    \label{ratio itm 2}}
    \caption{Own price effect analysis for an SAA-c game with five strategies ($n=2$, $m=7$, $\varepsilon=1$)}
    \label{fig:Own_price_effect 2}
    \end{figure*}
    
    \item \textit{Exposure}: In Figure \ref{fig:Exposure 2}, we plot the expected exposure and exposure frequency of each strategy A against every strategy B displayed on the x-axis. When all bidders play the same strategy, $SMS^\alpha$ has the remarkable property of never leading to exposure. Moreover, $SMS^\alpha$ never suffers from exposure against $MS^\lambda$ and EPE. Thirdly, even against SCPD and SB, $SMS^\alpha$ is rarely exposed. It has the lowest expected exposure and exposure frequency (with some ties with EPE and $MS^\lambda$). For instance, $SMS^\alpha$ induces $1.8$, $3.6$ and $7.0$ times less expected exposure against SB than EPE, SCPD and SB respectively. Strategy $MS^\lambda$ obtains roughly the same expected exposure than $SMS^\alpha$ against every strategy. Moreover, regarding exposure frequency, by playing $SMS^\alpha$, a bidder has $1.1$, $1.7$, $3.2$ and $4.9$ times less chance of ending up exposed against SB than $MS^\lambda$, EPE, SCPD and SB respectively.

    \begin{figure*}[t]
    \centering
    \subfloat[Expected exposure]{\includegraphics[width=3.1in]{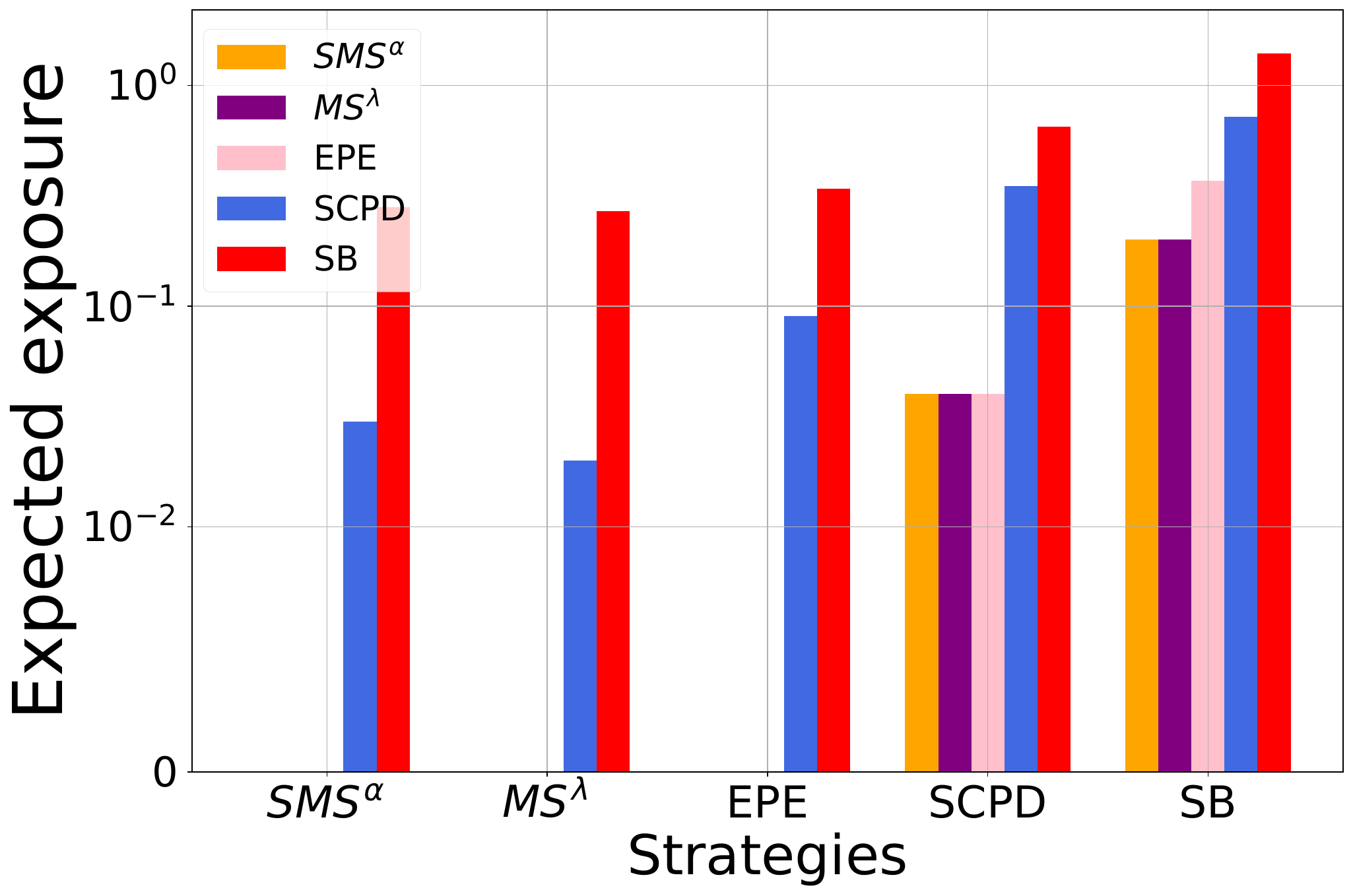}%
    \label{fig:Estimated Exposure 2}}
    \hfil
    \subfloat[Exposure frequency]{\includegraphics[width=3in]{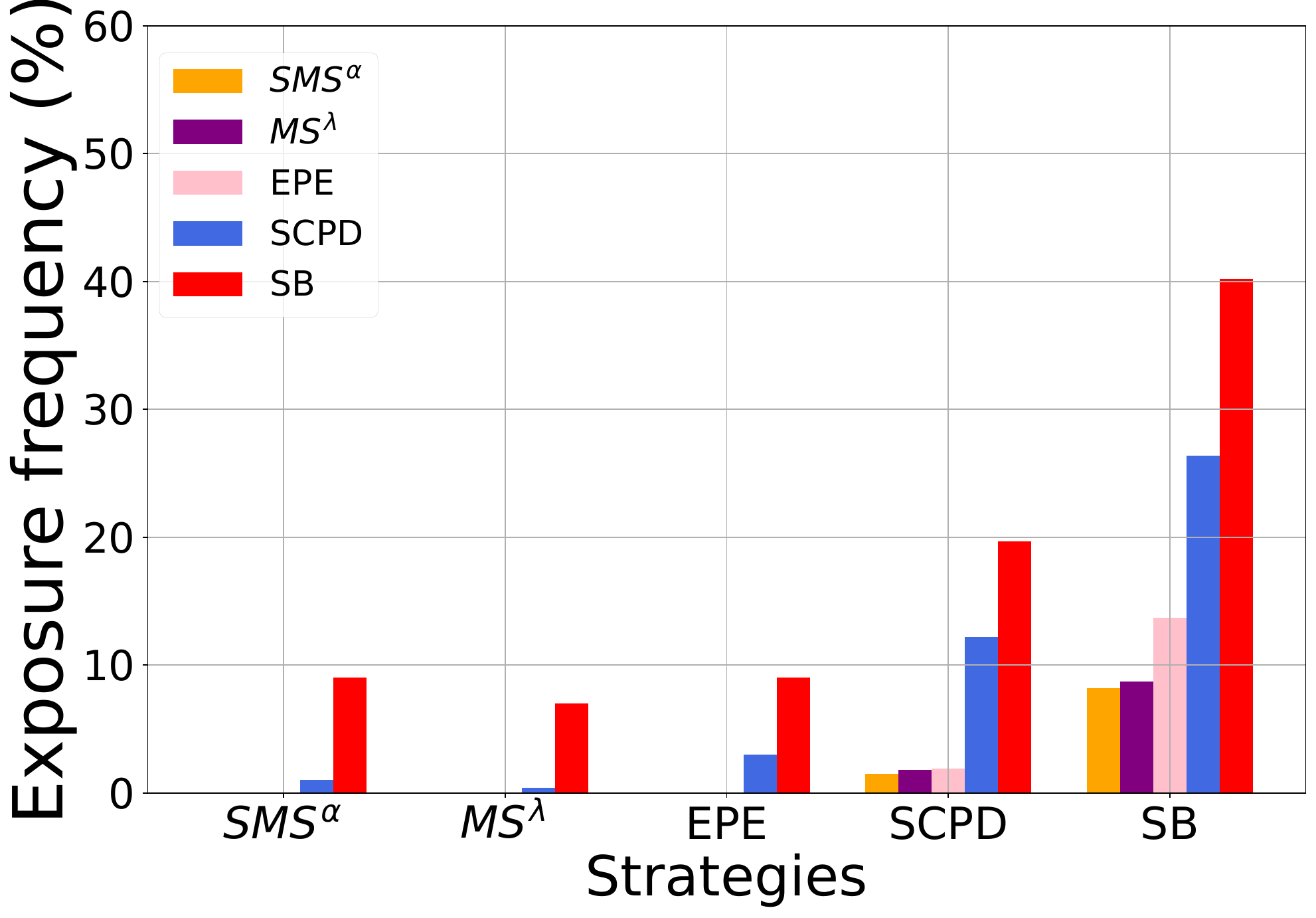}%
    \label{fig:Exposure frequency 2}}
    \caption{Exposure analysis for an SAA-c game with five strategies ($n=2$, $m=7$, $\varepsilon=1$)}
    \label{fig:Exposure 2}
    \end{figure*}
\end{enumerate}

Hence, we can draw the same conclusions for instances of size ($n=2$, $m=7$, $\varepsilon=1$) than for instances of size ($n=4$, $m=11$, $\varepsilon=1$) such as $SMS^\alpha$ achieves a higher expected utility than state-of-the-art algorithms, notably by better tackling the own price effect and the exposure problem in eligibility and budget constrained environments. \\

For instances of size ($n=3$, $m=9$, $\varepsilon=1$), the hyperparameter $\alpha$ of $SMS^\alpha$ takes the value $4$ and the risk-aversion hyperparameters $\lambda^r$ and $\lambda^o$ of $MS^\lambda$ both take the value $0.01$. The maximum number of expanded actions per information set $N_{act}$ of $SMS^\alpha$ is set to $20$. Each algorithm is given respectively $100$ seconds of thinking time. 
\begin{enumerate}
    \item \textit{Expected Utility:} We represent in Figure \ref{Expected utility 3 pl} the four normal form games in expected utility where each player has the choice between either playing $SMS^\alpha$ or another strategy A. In each empirical game, playing $SMS^\alpha$ is always profitable. Thus, the strategy profile ($SMS^\alpha$, $SMS^\alpha$, $SMS^\alpha$) is a Nash equilibrium of the normal form game in expected utility with strategy set \{$SMS^\alpha$, $MS^\lambda$, EPE, SCPD, SB\}. 
    
    Moreover, regarding strategy profiles where all bidders play the same strategy, $SMS^\alpha$ has a significantly higher expected utility than the other bidding strategies. For instance, it obtains an expected utility $1.09$, $1.17$ and $3.55$ times higher than respectively $MS^\lambda$, EPE and SCPD. When all bidders play SB, their expected utility is negative which highlights the risk of exposure. 

    \begin{figure*}[t]
    \centering
    \subfloat[$SMS^\alpha$ vs $MS^\lambda$]{\includegraphics[width=2.1in]{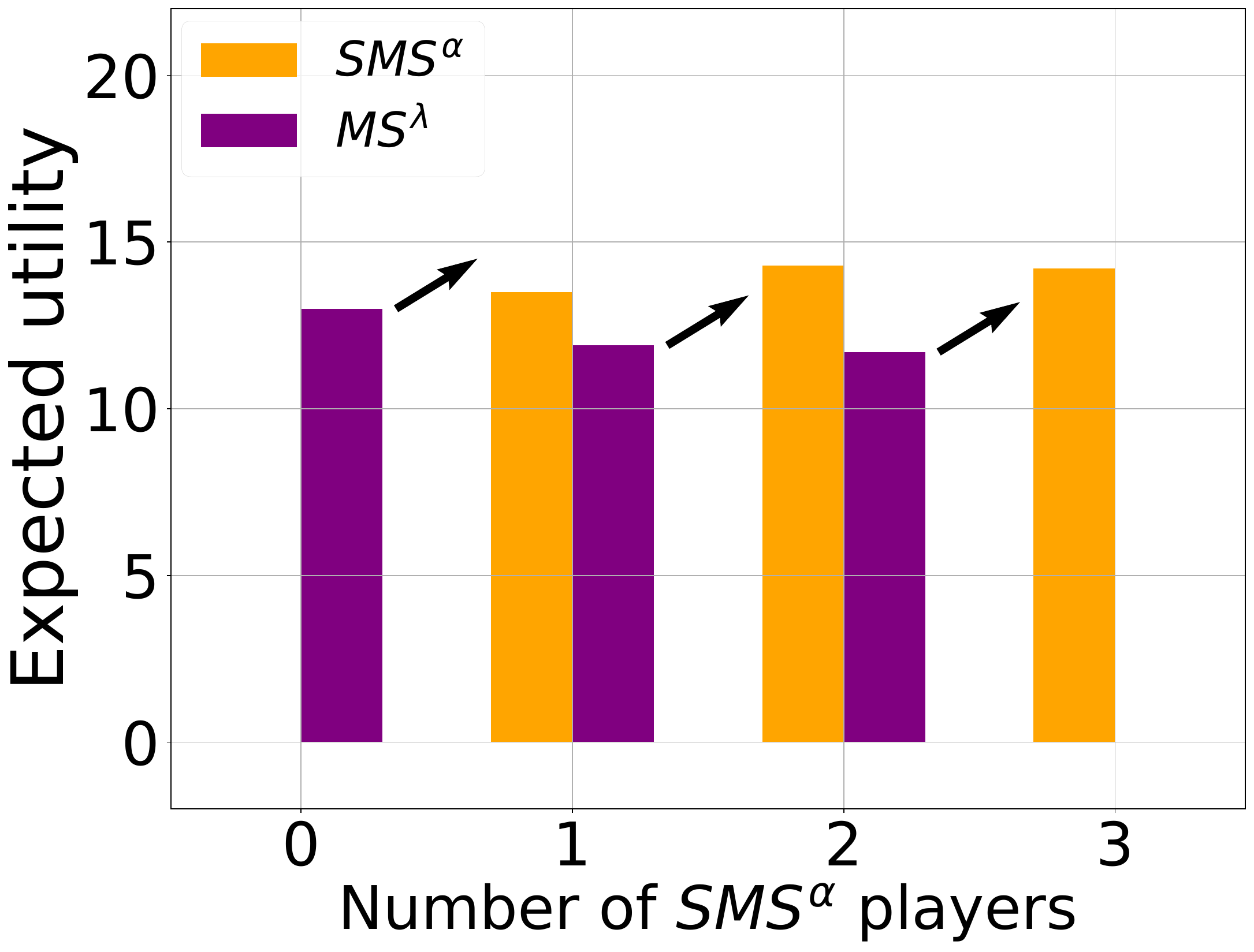}}
    \hfil
    \subfloat[$SMS^\alpha$ vs EPE]{\includegraphics[width=2.1in]{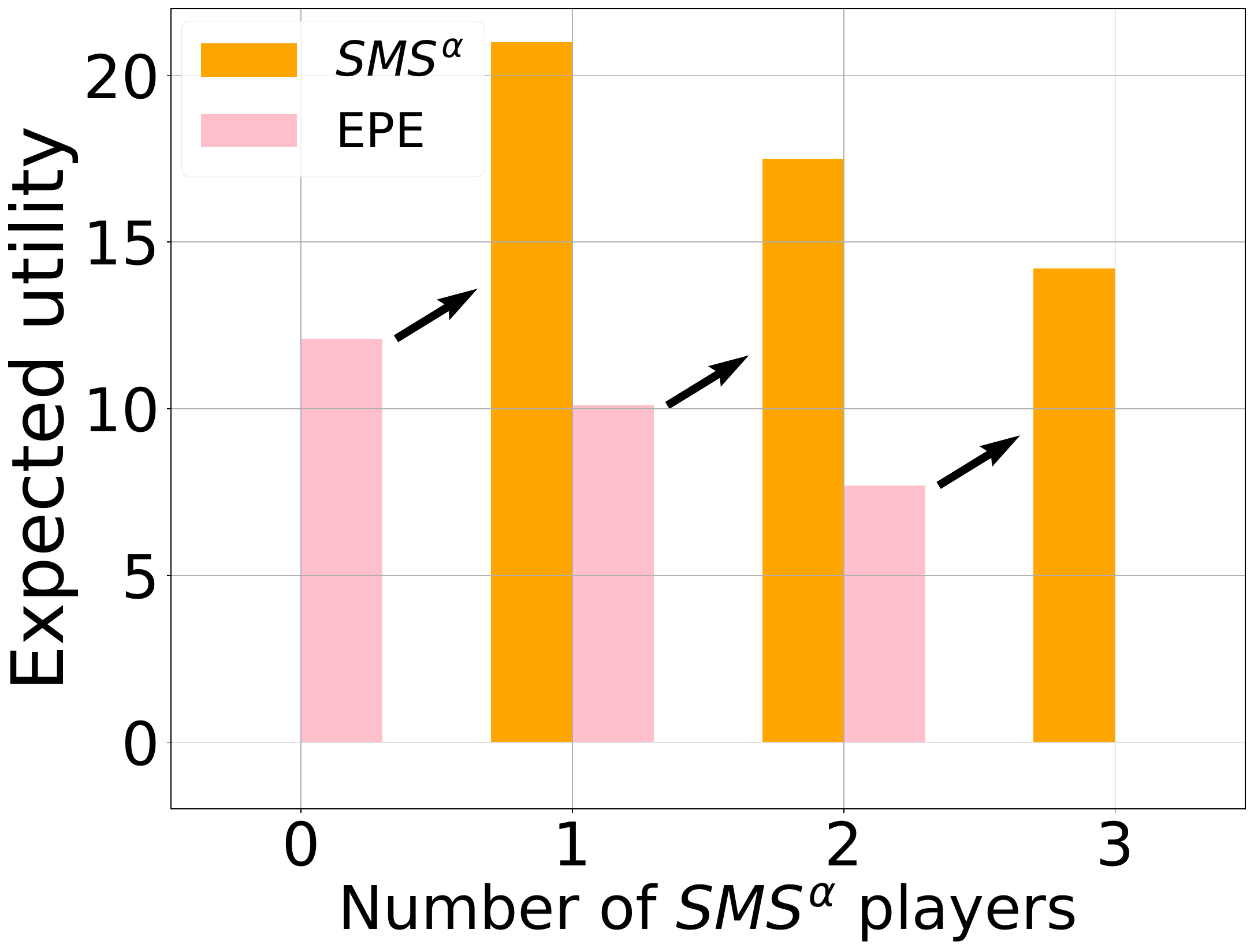}}

    \subfloat[$SMS^\alpha$ vs SCPD]{\includegraphics[width=2.1in]{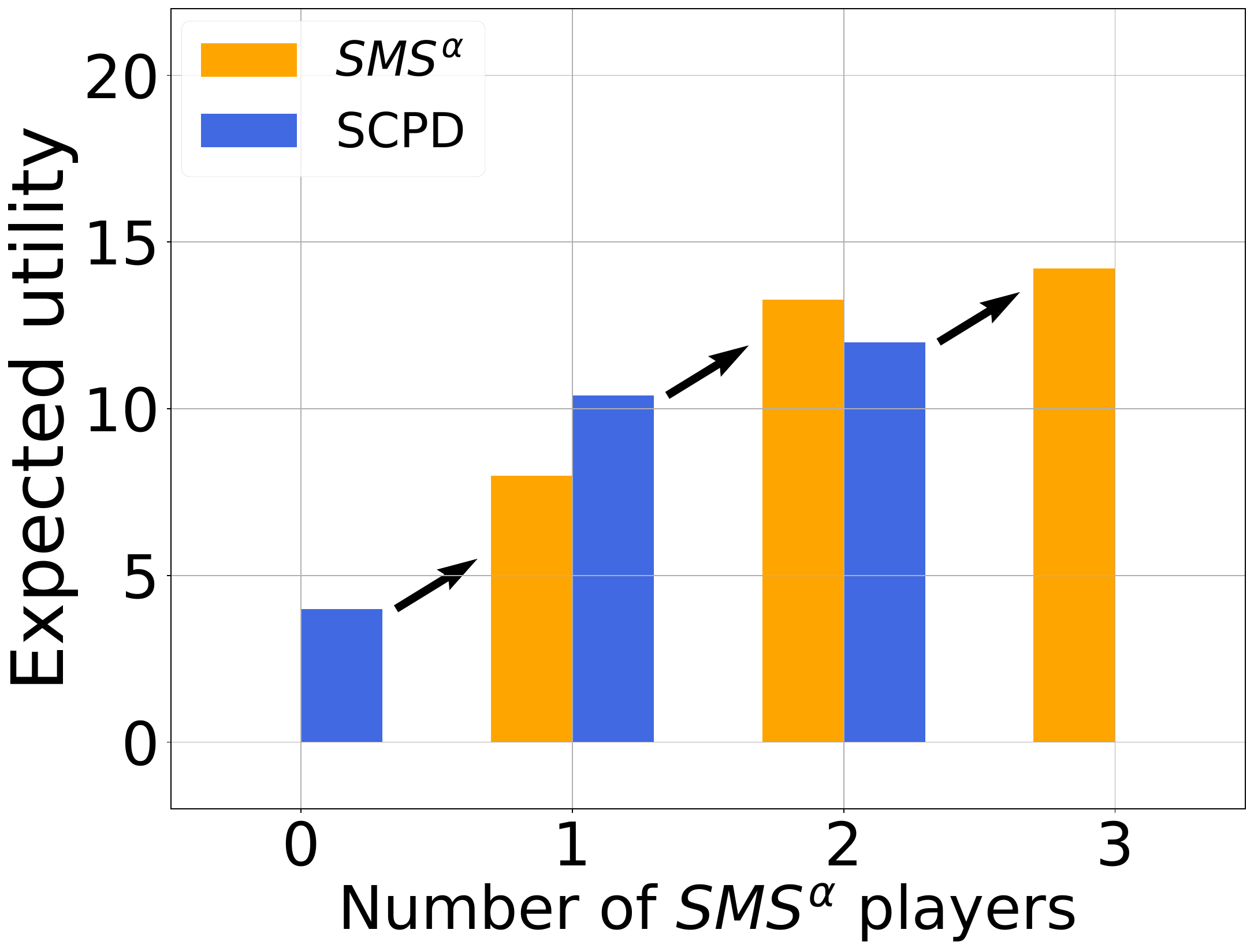}}
    \hfil
    \subfloat[$SMS^\alpha$ vs SB]{\includegraphics[width=2.1in]{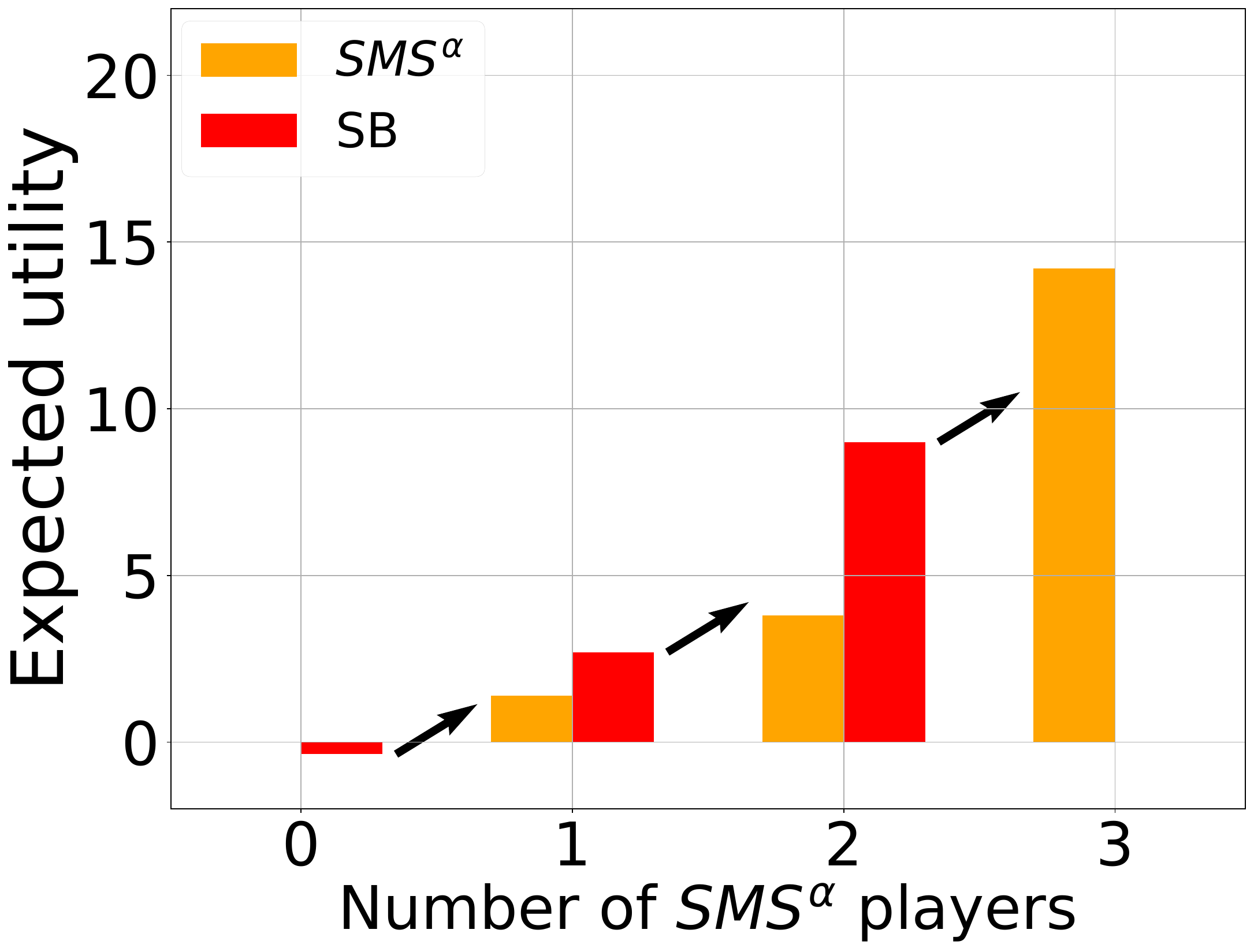}}
    \caption{Normal-form SAA-c game in expected utility with five 
    strategies ($n=3$, $m=9$, $\varepsilon=1$)}
    \label{Expected utility 3 pl}
    \end{figure*}
    
    \item \textit{Own price effect}: We plot in Figure \ref{avg price 3} the average price paid per item won and in Figure \ref{ratio itm 3} the ratio of items won by each strategy $A$ against every strategy $B$ displayed on the x-axis. In Figure \ref{avg price 3}, we can see that $SMS^\alpha$ acquires items at a lower price in average than the other strategies against $SMS^\alpha$, EPE, SCPD and SB (some ties occur for $SMS^\alpha$ and EPE). For instance, $SMS^\alpha$ spends $4.5\%$, $33.1\%$, $55.3\%$ and $59.6\%$ less per item won against SCPD than $MS^\lambda$, EPE, SCPD and SB respectively. Moreover, against $MS^\lambda$, only EPE spends slightly less than $SMS^\alpha$ per item won. 

    Moreover, regarding strategy profiles where all bidders play the same strategy, the one corresponding to $SMS^\alpha$ has an average price payed per item won $1.15$, $3$ and $3.89$ times lower than $MS^\lambda$, SCPD and SB respectively. All items are allocated when all bidders play $SMS^\alpha$ whereas only $75\%$ are allocated when all bidders play EPE. The fact that items are divided more efficiently with $SMS^\alpha$ than with EPE explains why the strategy profile where all bidders play $SMS^\alpha$ has a higher expected utility than when all bidders play EPE even though the items are purchased at a higher price.
    
    \begin{figure*}[t]
    \centering
    \subfloat[Average price per item won]{\includegraphics[width=3in]{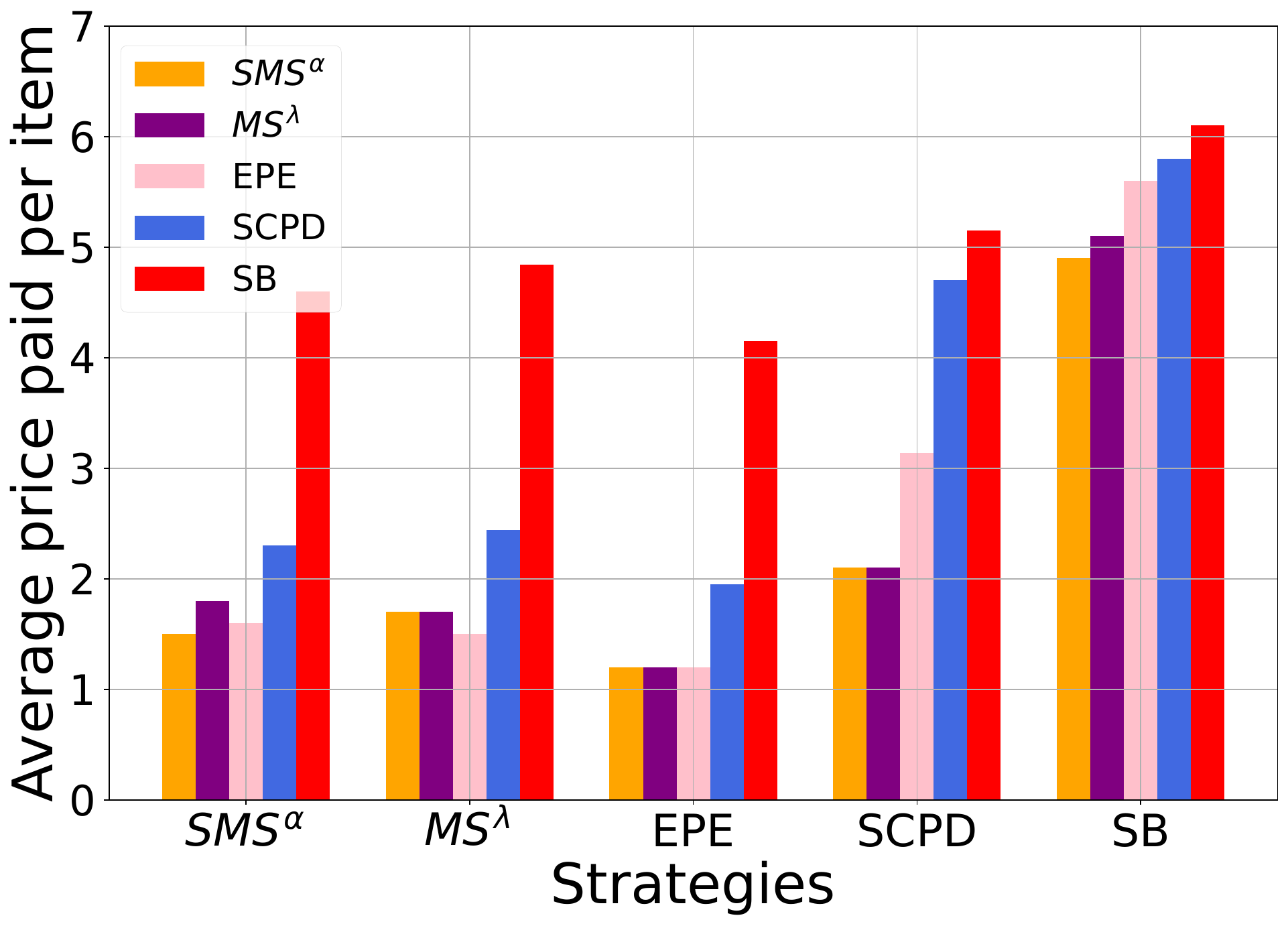}%
    \label{avg price 3}}
    \hfil
    \subfloat[Ratio of items won]{\includegraphics[width=3.1in]{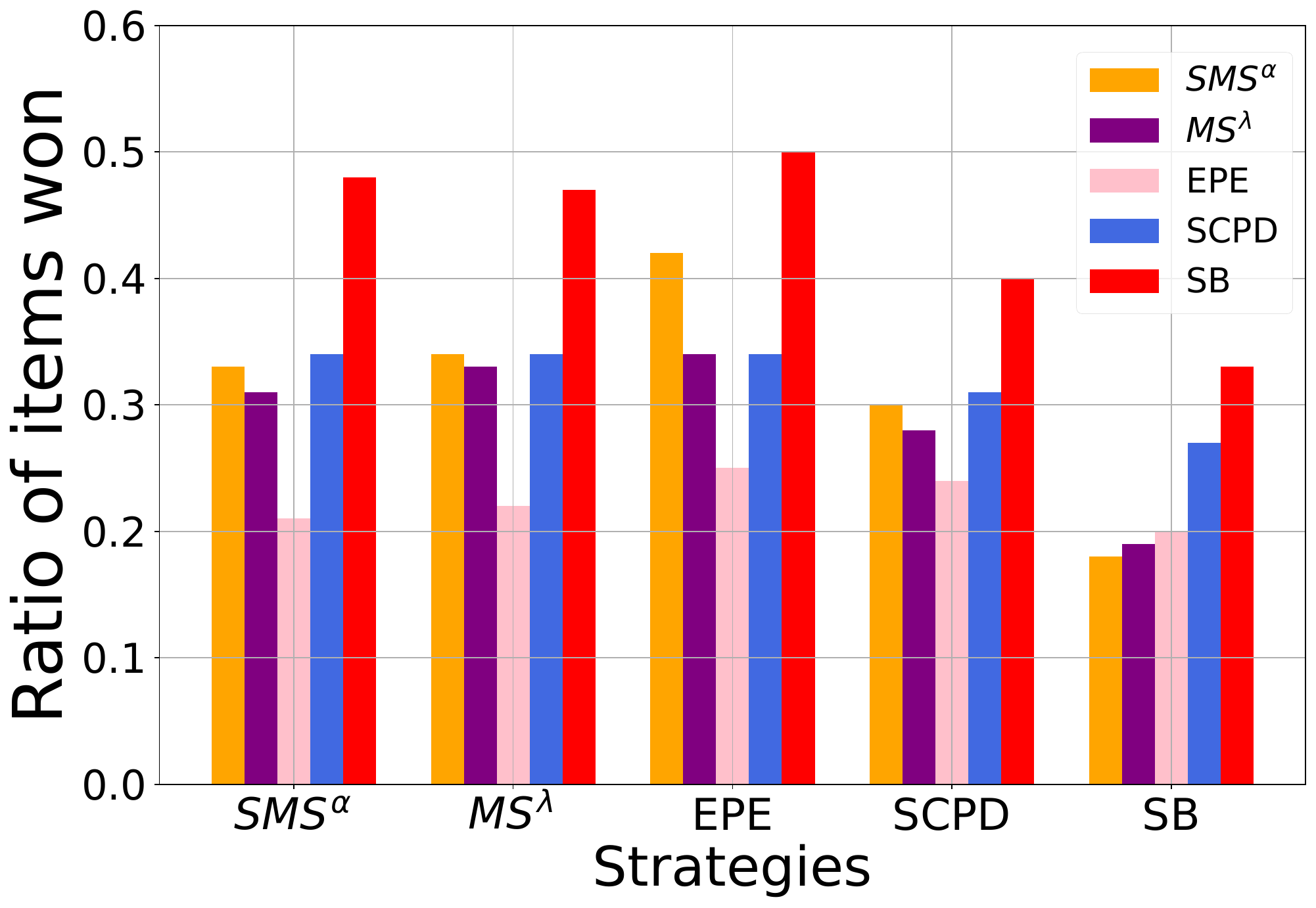}%
    \label{ratio itm 3}}
    \caption{Own price effect analysis for an SAA-c game with five strategies ($n=3$, $m=9$, $\varepsilon=1$)}
    \label{fig:Own_price_effect 3}
    \end{figure*}
    
    \item \textit{Exposure}: In Figure \ref{fig:Exposure 3}, we plot the expected exposure and exposure frequency of each strategy A against every strategy B displayed on the x-axis. When all bidders play the same strategy, $SMS^\alpha$ has the remarkable property of never leading to exposure. Moreover, $SMS^\alpha$ never suffers from exposure against $MS^\lambda$ and EPE. Thirdly, even against SCPD and SB, $SMS^\alpha$ is rarely exposed. It has the lowest expected exposure and exposure frequency. For instance, $SMS^\alpha$ induces $2.7$, $3.1$, $30.7$ and $77.3$ times less expected exposure against SCPD than $MS^\lambda$, EPE, SCPD and SB respectively. Moreover, regarding exposure frequency, by playing $SMS^\alpha$, a bidder has $2.3$, $2.2$, $19.1$ and $38.9$ times less chance of ending up exposed against SCPD than $MS^\lambda$, EPE, SCPD and SB respectively.

    \begin{figure*}[t]
    \centering
    \subfloat[Expected exposure]{\includegraphics[width=3.1in]{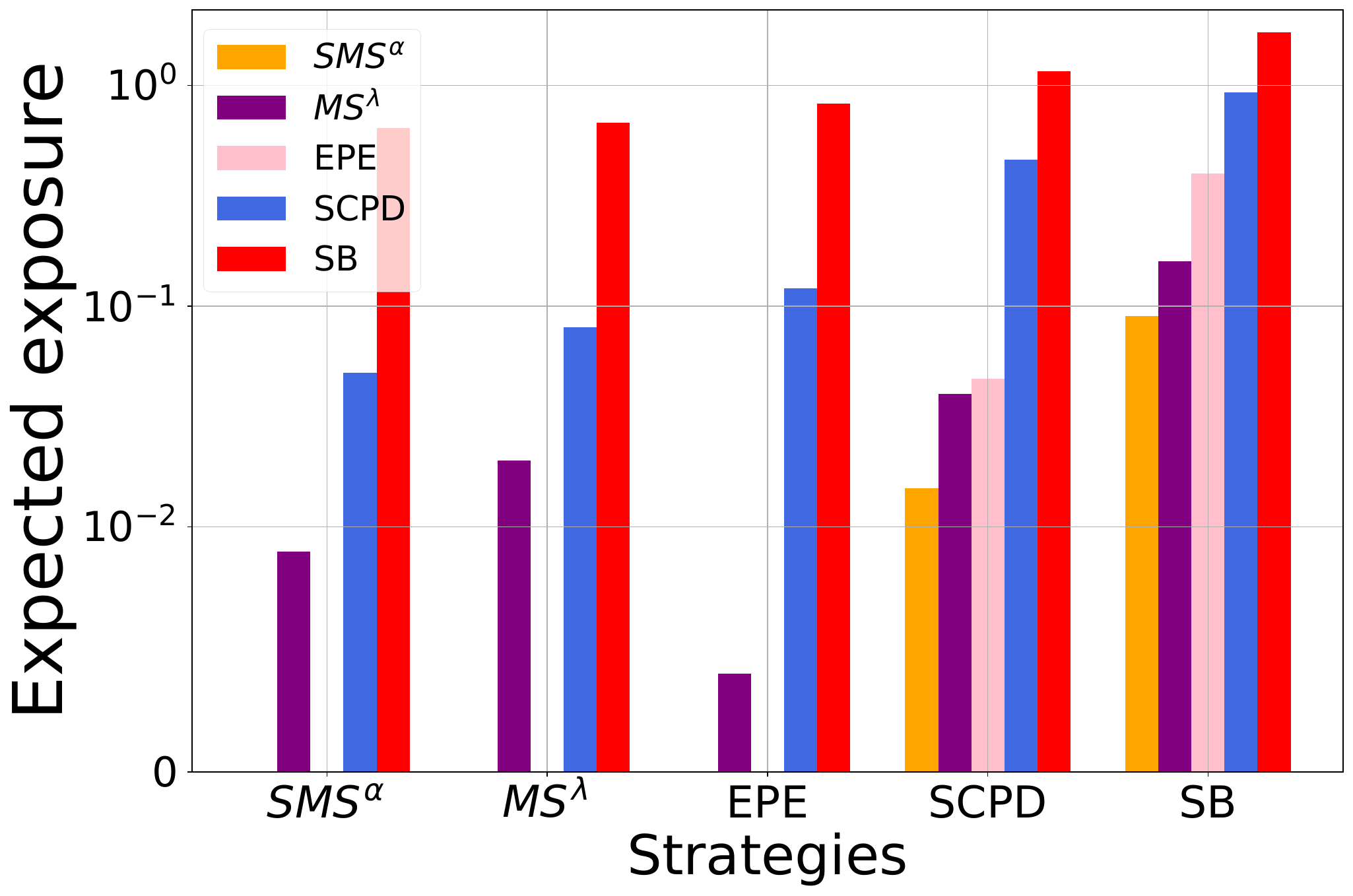}%
    \label{fig:Estimated Exposure 3}}
    \hfil
    \subfloat[Exposure frequency]{\includegraphics[width=3in]{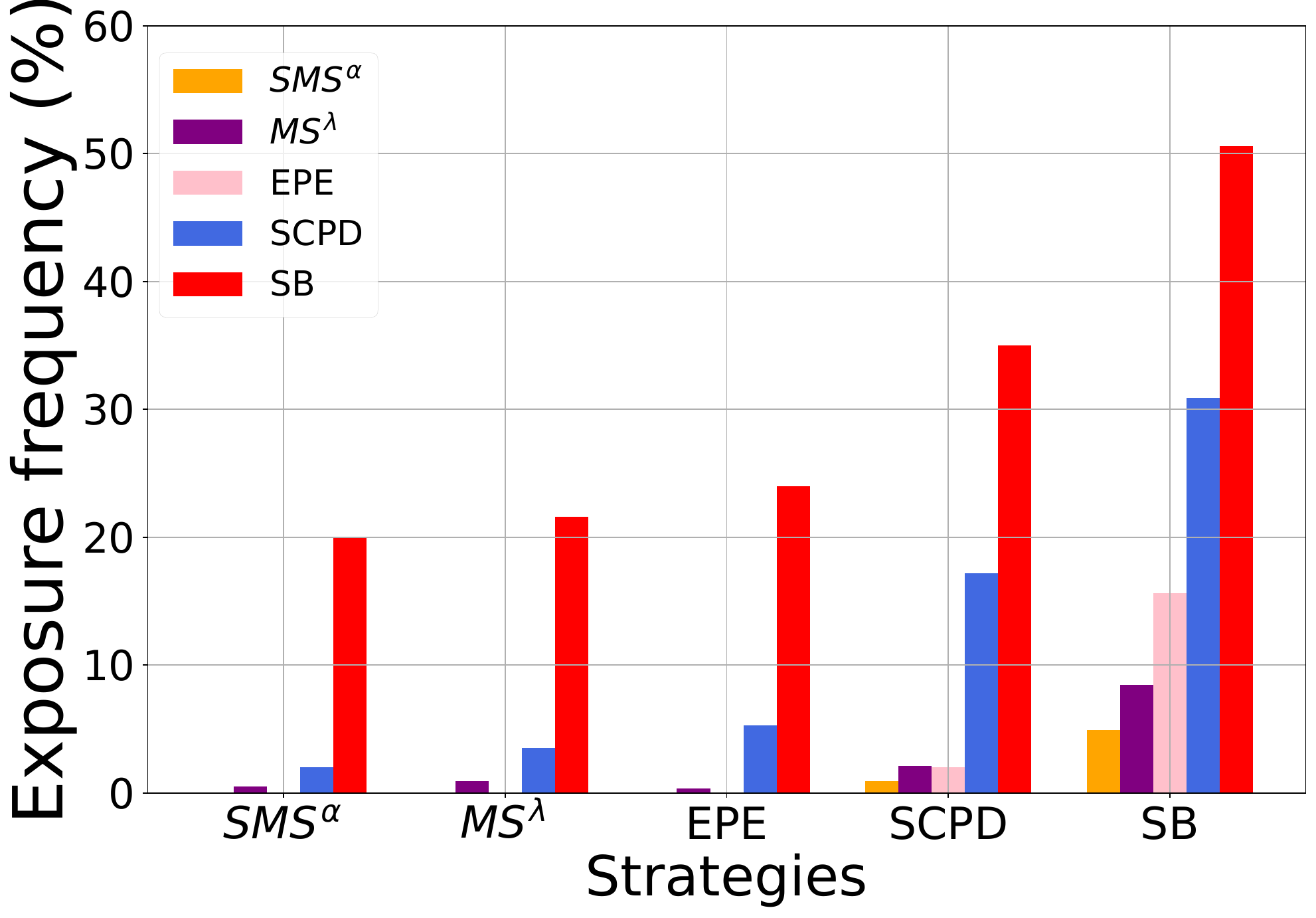}%
    \label{fig:Exposure frequency 3}}
    \caption{Exposure analysis for an SAA-c game with five strategies ($n=3$, $m=9$, $\varepsilon=1$)}
    \label{fig:Exposure 3}
    \end{figure*}
    
\end{enumerate}

Hence, we can draw the same conclusions for instances of size ($n=3$, $m=9$, $\varepsilon=1$) than for instances of size ($n=4$, $m=11$, $\varepsilon=1$) such as $SMS^\alpha$ achieves higher expected utility than state-of-the-art algorithms, notably by better tackling the own price effect and the exposure problem in eligibility and budget constrained environments. \\

\vfill

\end{document}